\documentclass[12pt]{article}
\usepackage[cp1250]{inputenc}
\usepackage[verbose,a4paper,tmargin=0.5in,lmargin=1in,rmargin=1in]{geometry}
\usepackage[pdftex]{graphicx}
\usepackage{amsfonts}
\usepackage{indentfirst}
\usepackage{latexsym}
\usepackage{amsmath}
\usepackage{amsthm}
\usepackage{amssymb}
\usepackage{psfrag}
\usepackage[mathscr]{eucal} 
\usepackage{color}
\usepackage{cite}
\usepackage{longtable}
\usepackage{verbatim}
\usepackage{enumitem}
\usepackage[font=normalsize]{caption}
\usepackage{bigints}

\usepackage{graphicx}
\usepackage[all]{xy}
\usepackage{float}
\usepackage{afterpage}

\newtheorem{Twierdzenie}{Theorem}[section]
\newtheorem{Wniosek}{Corollary}[section]
\newtheorem{Definicja}{Definition}[section]
\newtheorem{Lemat}{Lemma}[section]
\newtheorem{Uwaga}{Remark}[section]

\newfloat{Scheme}{tbp}{muj}

\title{Hyperheavenly spaces and their application in Walker and para-Kähler geometries: part II}

\author{$\textrm{Adam Chudecki}^{*}$}

\begin{document}
\sloppy

\maketitle

$*$ Center of Mathematics and Physics, Lodz University of Technology, 
\newline
$\ \ \ \ \ $ Al. Politechniki 11, 90-924 Łódź, Poland, adam.chudecki@p.lodz.pl
\newline
\newline
\newline
\textbf{Abstract}.  
\newline
4-dimensional spaces equipped with congruences of null strings are considered. It is assumed that a space admits a congruence of expanding self-dual null strings and its self-dual part of the Weyl tensor is algebraically degenerate. Different Petrov-Penrose types of such spaces are analyzed. A special attention is paid to para-Kähler Einstein spaces. All para-Kähler Einstein metrics of spaces with algebraically degenerate self-dual Weyl spinor are found in all the generality.


\renewcommand{\arraystretch}{1.5}
\setlength\arraycolsep{2pt}
\setcounter{equation}{0}

\section{Introduction}

This is the second part of the work devoted to para-Kähler and Walker geometries. In the first part \cite{Chudecki_Kahler_1} we considered complex and real neutral, 4-dimensional spaces equipped with a nonexpanding congruence\footnote{\textsl{A congruence (a foliation) of null strings} is abbreviated by $\mathcal{C}$. Also, an intersection of self-dual and anti-self-dual congruences of null strings which constitutes \textsl{a congruence of null geodesics} is abbreviated by $\mathcal{I}$. The list of all the abbreviations which we use in the present paper can be found in Ref. \cite{Chudecki_Kahler_1}.} of self-dual (SD) null strings. These are spaces of the types $[\textrm{deg}]^{n} \otimes [\textrm{any}]$. This part is devoted to spaces with expanding SD $\mathcal{C}$, i.e., spaces of the types $[\textrm{deg}]^{e} \otimes [\textrm{any}]$. 

[At this point an important remark is needed. It the present paper we use the spinorial formalism, hyperheavenly spaces ($\mathcal{HH}$-spaces) formalism and abbreviations introduced in Ref. \cite{Chudecki_Kahler_1}. We encourage the Reader of this paper to read Ref. \cite{Chudecki_Kahler_1} first. More detailed approach to the spinorial formalism can be found in \cite{Penrose, Plebanski_Spinors, Pleban_formalism_2}. A spinorial approach to Walker and sesquiWalker geometries was presented in \cite{Law_3, Law_2, Law, Chudecki_Przanowski_Walkery}. \cite{Finley_Plebanski_intrinsic,Boyer_Finley_Plebanski} are fundamental references devoted to $\mathcal{HH}$-spaces. Congruences of null strings are analyzed in \cite{Plebanski_Rozga, Chudecki_struny, Plebanski_surf}. Applications of complex methods in real geometries, especially in para-Kähler geometry can be found in \cite{An, Bor, Bor_Makhmali_Nurowski,Chudecki_przyklady, Chudecki_Ricci}.]

If a space is equipped with a $\mathcal{C}^{n}$ which is SD then the SD Weyl spinor is algebraically special. However, if SD $\mathcal{C}$ is expanding then the SD Weyl spinor is not algebraically special in general. Thus, in what follows we assume two basic facts
\begin{enumerate}[label=(\roman*)]
\item a space admits an expanding congruence of null strings and the orientation is chosen in such a manner that this congruence is self-dual,
\item the self-dual part of the Weyl tensor is algebraically special.
\end{enumerate}
In Einstein spaces (i) and (ii) are equivalent (it is a statement of the complex Goldberg-Sachs theorem, see \cite{Plebanski_surf}). However, in Sections \ref{subsekcja_dege_x_anye} - \ref{subsekcja_dege_x_degne}, \ref{subsekcja_general_result} and \ref{sekcja_para_Kahler_complicated_nonEinstein} we consider spaces with an arbitrary traceless Ricci tensor. For such spaces (i) and (ii) are not equivalent in general.

The paper is organized, as follows. In Section \ref{weak_hyperheavenly_spaces} the structure of weak (expanding) $\mathcal{HH}$-spaces is presented. Weak $\mathcal{HH}$-spaces are equipped with a single $\mathcal{C}$ \textsl{ex definitione}. This $\mathcal{C}$ is assumed to be expanding and it is a common property of all the spaces which metrics are presented in this paper. Then we follow the pattern of Ref. \cite{Chudecki_Kahler_1} (Scheme \ref{Structure} shows a structure of the present paper and compares it to Ref. \cite{Chudecki_Kahler_1}.). 

\begin{Scheme}[!ht]
\begin{displaymath}
\xymatrixcolsep{0.0cm}
\xymatrixrowsep{0.0cm}
\xymatrix{
\textbf{Ref. \cite{Chudecki_Kahler_1}:} &   & \textbf{This paper:}  \\
[\textrm{deg}]^{n} \otimes [\textrm{any}] \ar[dd] &   & [\textrm{deg}]^{e} \otimes [\textrm{any}] \ar[dd]^{\textrm{Section } \ref{subsekcja_dege_x_anye}} \ \  \\ 
      &  \ar[l] \textrm{let } \mathcal{C}^{e}_{m^{\dot{A}}} \textrm{ exists} \ar[r] & \\
  \{ [\textrm{deg}]^{n} \otimes [\textrm{any}]^{e},[--]  \} & &\{ [\textrm{deg}]^{e} \otimes [\textrm{any}]^{e},[--]  \} \ (\textrm{Thm } \ref{theorem_dege_x_anye_mm_pm}) \\  
   \{ [\textrm{deg}]^{n} \otimes [\textrm{any}]^{e},[++]  \} \ar[dddd] & &\{ [\textrm{deg}]^{e} \otimes [\textrm{any}]^{e},[+-]  \} \ (\textrm{Thm } \ref{theorem_dege_x_anye_mm_pm}) \\  
    & & \{ [\textrm{deg}]^{e} \otimes [\textrm{any}]^{e},[-+]  \} \ (\textrm{Thm } \ref{theorem_dege_x_anye_mp_pp}) \\
    & & \{ [\textrm{deg}]^{e} \otimes [\textrm{any}]^{e},[++]  \} \ (\textrm{Thm } \ref{theorem_dege_x_anye_mp_pp}) \ar[dd]^{\textrm{Section } \ref{subsekcja_dege_x_degn}}  \\  
  &  \ar[l] \textrm{let } \mathcal{C}^{n}_{m^{\dot{A}}} \textrm{ exists}  \ar[r] & \\
  [\textrm{deg}]^{n} \otimes [\textrm{deg}]^{n} \ar[ddd] & & \{ [\textrm{deg}]^{e} \otimes [\textrm{deg}]^{n},[--] \} \ (\textrm{Thm } \ref{theorem_dege_x_degn_mm})  \\
   & & \{ [\textrm{deg}]^{e} \otimes [\textrm{deg}]^{n},[++] \} \ (\textrm{Thm } \ref{theorem_dege_x_degn_pp}) \ar[dd]^{\textrm{Section } \ref{subsekcja_dege_x_degne}}  \\ 
   &  \ar[l] \textrm{let } \mathcal{C}^{n}_{m^{\dot{A}}} \textrm{ and } \mathcal{C}^{e}_{n^{\dot{A}}} \textrm{ exist} \ar[r] & \\
 \{ [\textrm{deg}]^{n} \otimes [\textrm{deg}]^{ne},[--,--] \} &  &  \{ [\textrm{deg}]^{e} \otimes [\textrm{deg}]^{ne},[--,+-] \}  \   (\textrm{Thm } \ref{theorem_dege_x_degne_mm_pm}) \\
 \{ [\textrm{deg}]^{n} \otimes [\textrm{deg}]^{ne},[--,++] \} \ar[ddddddd] &  &  \{ [\textrm{deg}]^{e} \otimes [\textrm{deg}]^{ne},[--,-+] \} \ (\textrm{Thm } \ref{theorem_dege_x_degne_mm_mp}) \\   
         &  &  \{ [\textrm{deg}]^{e} \otimes [\textrm{deg}]^{ne},[--,++] \} \ (\textrm{Thm } \ref{theorem_dege_x_degne_mm_pp}) \\  
   &  &  \{ [\textrm{deg}]^{e} \otimes [\textrm{deg}]^{ne},[++,--] \} \ (\textrm{Thm } \ref{theorem_dege_x_degne_pp_mm}) \\  
   &  &  \{ [\textrm{deg}]^{e} \otimes [\textrm{deg}]^{ne},[++,+-] \} \ (\textrm{Thm } \ref{theorem_dege_x_degne_pp_pm}) \\  
   &  &  \{ [\textrm{deg}]^{e} \otimes [\textrm{deg}]^{ne},[++,-+] \} \ (\textrm{Cor } \ref{wniosek_o_nieistnieniu}) \  \\  
   &  &  \{ [\textrm{deg}]^{e} \otimes [\textrm{deg}]^{ne},[++,++] \}  \ (\textrm{Thm } \ref{theorem_dege_x_degne_pp_pp_and_pp_mp}) \ar[dd]^{\textrm{Sections } \ref{subsekcja_dege_x_degnn} \textrm{ and } \ref{sekcja_para_Kahler_complicated}} 
 \\   
 &  \ar[l] \textrm{let } \mathcal{C}^{n}_{m^{\dot{A}}} \textrm{ and } \mathcal{C}^{n}_{n^{\dot{A}}} \textrm{ exist}  \ar[r] & \\  
  [\textrm{deg}]^{n} \otimes [\textrm{D}]^{nn} &  & \{ [\textrm{deg}]^{e} \otimes [\textrm{D}]^{nn},[++,--]  \} \ (\textrm{Thm } \ref{theorem_dege_x_Dnn_pp_mm}) \\
   &  & \{ [\textrm{deg}]^{e} \otimes [\textrm{D}]^{nn},[++,++]  \}  \ (\textrm{Thm } \ref{theorem_dege_x_degnn_pp_pp})     
}
\end{displaymath} 
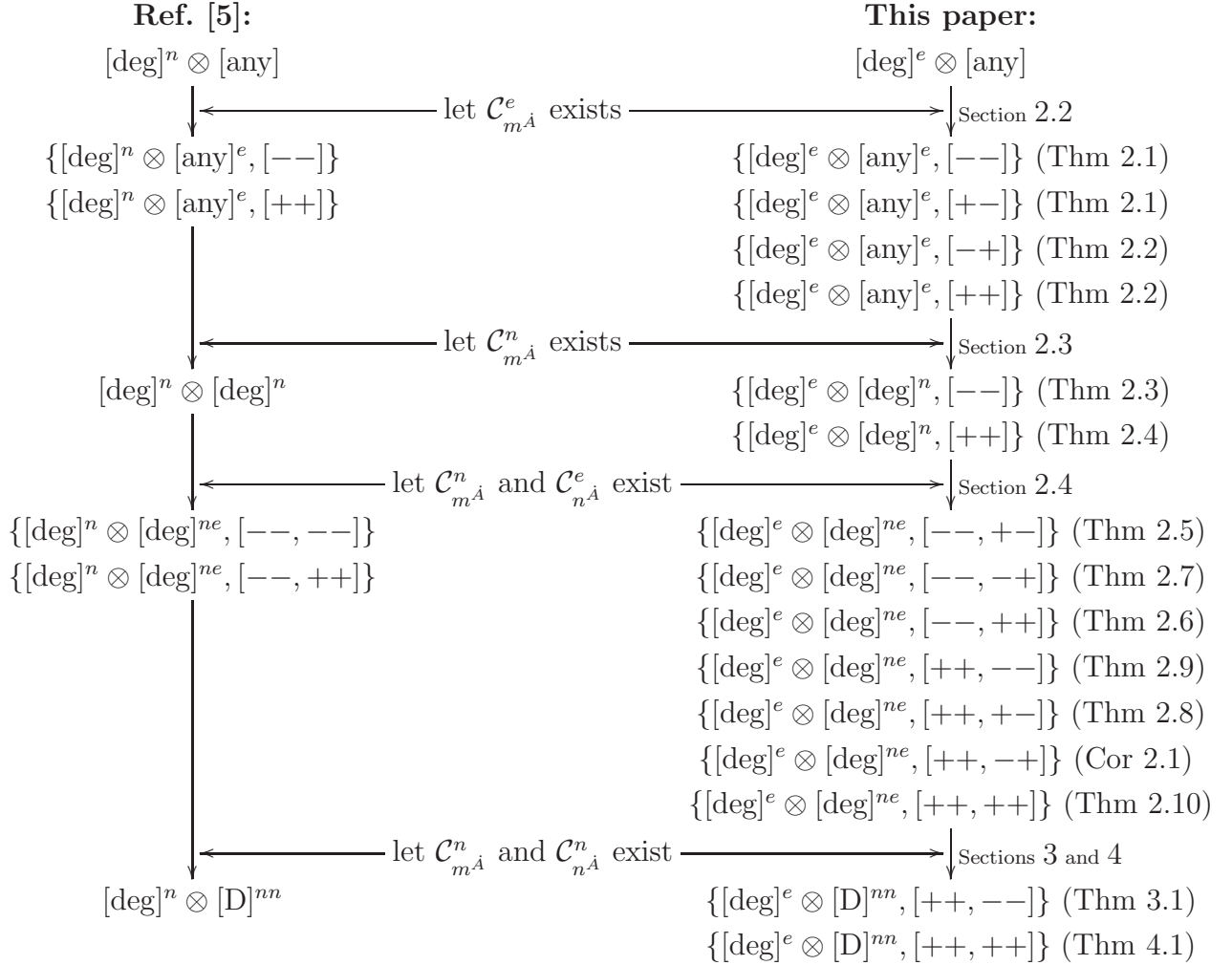
\caption{Diagram of weak $\mathcal{HH}$-spaces equipped with ASD congruences of null strings.}
\label{Structure}
\end{Scheme}

First we assume that a weak $\mathcal{HH}$-space admits also an anti-self-dual (ASD) $\mathcal{C}^{e}$ (Section \ref{subsekcja_dege_x_anye}). The existence of two $\mathcal{C}s$ of a different duality allows to introduce a sub-classification of spaces based on the properties of the intersection of $\mathcal{C}s$. Hence, there are 4 subtypes of spaces of the type $ [\textrm{deg}]^{e} \otimes [\textrm{any}]^{e}$.

Then we assume that the ASD $\mathcal{C}$ is nonexpanding (Section \ref{subsekcja_dege_x_degn}). There are 2 subtypes of such spaces. The next step is to equipped a space with the second ASD $\mathcal{C}$ which is expanding (Section \ref{subsekcja_dege_x_degne}). According to the results of Ref. \cite{Chudecki_Kahler_1} there are 7 possible subtypes of such spaces. Our considerations prove that one of these subtypes, in fact, cannot exist\footnote{This unexpected fact gives an answer to one of the questions formulated in \textsl{Concluding Remarks} of Ref. \cite{Chudecki_Kahler_1}. Namely, we asked if there are any geometrical obstacles which prevent the existence of certain subtypes listed in Tables 9-11 of Ref. \cite{Chudecki_Kahler_1}. The answer to this question is positive. In fact, the type $ \{ [\textrm{deg}]^{e} \otimes [\textrm{deg}]^{ne},[++,-+] \}$ does not exist.}. The metrics presented in Sections \ref{subsekcja_dege_x_anye} - \ref{subsekcja_dege_x_degne} are not analyzed in details because their geometrical interpretation is not clear yet. It is only worth to mention that all of the spaces which are equipped with ASD $\mathcal{C}^{n}$ are Walker spaces.

In Sections \ref{subsekcja_dege_x_degnn} and \ref{sekcja_para_Kahler_complicated} the second ASD $\mathcal{C}$ becomes nonexpanding. Thus, spaces are equipped with ASD $\mathcal{C}^{nn}$ and they are para-Kähler. There are two different subtypes of such spaces, namely $\{ [\textrm{deg}]^{e} \otimes [\textrm{D}]^{nn},[++,--]  \}$  and $\{ [\textrm{deg}]^{e} \otimes [\textrm{D}]^{nn},[++,++] \}$. 

Section \ref{subsekcja_dege_x_degnn} is devoted to the types $\{ [\textrm{deg}]^{e} \otimes [\textrm{D}]^{nn},[++,--]  \}$. These types have been completely solved (Theorem \ref{theorem_dege_x_Dnn_pp_mm}). The metric (\ref{metryka_dege_x_Dnn_pp_mm}) is a general metric of spaces of the types $\{ [\textrm{deg}]^{e} \otimes [\textrm{D}]^{nn},[++,--]  \}$. Thorough analysis of such spaces via Petrov-Penrose classification is presented. We consider also para-Kähler Einstein spaces of the types $\{ [\textrm{deg}]^{e} \otimes [\textrm{D}]^{nn},[++,--]  \}$ (Theorem \ref{theorem_Einstein_dege_x_Dnn_pp_mm}, the metric (\ref{metryka_Einstein_dege_x_Dnn_pp_mm})). 

The types $\{ [\textrm{deg}]^{e} \otimes [\textrm{D}]^{nn},[++,++]  \}$ (Section \ref{sekcja_para_Kahler_complicated}) are much more complicated. We found two forms of the metric of such spaces (Theorem \ref{theorem_dege_x_degnn_pp_pp}, metrics (\ref{metryka_ogolnyparaKahler_classIIa}) and (\ref{metryka_ogolnyparaKahler_classIIb})). The geometrical difference between metrics (\ref{metryka_ogolnyparaKahler_classIIa}) and (\ref{metryka_ogolnyparaKahler_classIIb}) remains unknown. We are not sure if they are essentially different or the metric (\ref{metryka_ogolnyparaKahler_classIIb}) is just a special case of (\ref{metryka_ogolnyparaKahler_classIIa}). Thus, we do not analyze these metrics in details. They will be analyzed elsewhere.

In Section \ref{subsekcja_degexdegnn_pp_pp_Einstein} we focus on para-Kähler Einstein spaces of the types $\{ [\textrm{deg}]^{e} \otimes [\textrm{D}]^{nn},[++,++]  \}$ and we arrive at the main result\footnote{The question which of the results presented in this paper is the "main" result is a moot point. According to our best knowledge almost all the metrics which we found are new results (see the summary Tables \ref{summary} and \ref{summary_2}). However, it seems that pKE metrics have the most transparent geometrical meaning and via their relation with twistor distributions \cite{Bor_Makhmali_Nurowski} are physically relevant. Thus, we think that the most important results of the paper are Theorems \ref{theorem_Einstein_dege_x_Dnn_pp_mm} and \ref{theorem_Einstein_dege_x_Dnn_pp_pp}. Theorem \ref{theorem_Einstein_dege_x_Dnn_pp_pp} as a more general than Theorem \ref{theorem_Einstein_dege_x_Dnn_pp_mm} deserves to be called "the main result" of the paper.} of the paper. It is Theorem \ref{theorem_Einstein_dege_x_Dnn_pp_pp} and the metric (\ref{metryka_Einstein_dege_x_Dnn_pp_pp}) which is a general metric of the para-Kähler Einstein spaces for which SD Weyl tensor is algebraically special. The metric (\ref{metryka_Einstein_dege_x_Dnn_pp_pp}) depends on 4 functions of two variables like it was proved in \cite{Bor_Makhmali_Nurowski}. The types $\{ [\textrm{III}]^{e} \otimes [\textrm{D}]^{nn},[++,++]  \}$ and $\{ [\textrm{N}]^{e} \otimes [\textrm{D}]^{nn},[++,++]  \}$ have been also found and they depend on 3 and 2 functions of two variables, respectively\footnote{We skip an analysis of the types $[\textrm{D}]^{ee} \otimes [\textrm{D}]^{nn}$. All such types were found in \cite{Bor_Makhmali_Nurowski}.}.

Concluding remarks end the paper. 

Considerations are purely local. We consider complex manifolds of dimension four equipped with a holomorphic metric. Thus, coordinates are complex and functions are holomorphic. If we replace complex coordinates by real ones and holomorphic functions by real smooth ones, we arrive at 4-dimensional real neutral geometries. Hence, all metrics presented in this paper have neutral slices.

\renewcommand{\arraystretch}{1.5}
\setlength\arraycolsep{2pt}


\section{Weak expanding hyperheavenly spaces}
\setcounter{equation}{0}
\label{weak_hyperheavenly_spaces}

\subsection{Spaces of the types $[\textrm{deg}]^{e} \otimes [\textrm{any}]$} 
\label{subsekcja_dege_x_any}

\subsubsection{Structure of weak expanding hyperheavenly spaces}

In this Section we gather the basic definitions and notations concerning weak hyperheavenly spaces. This Section resembles Section 2.6 of Ref. \cite{Chudecki_Kahler_1} but in \cite{Chudecki_Kahler_1} all the formulas were adapted to \textsl{weak nonexpanding hyperheavenly spaces} (types $[\textrm{deg}]^{n} \otimes [\textrm{any}]$). In what follows we focus on more general \textsl{weak expanding hyperheavenly spaces}, i.e., spaces of the types $[\textrm{deg}]^{e} \otimes [\textrm{any}]$. 

\begin{Definicja}
\label{definicja_weak_HH_spaces}
\textsl{Weak hyperheavenly space (weak $\mathcal{HH}$-space)} is a pair $(\mathcal{M}, ds^{2})$ where $\mathcal{M}$ is a 4-dimensional complex analytic differential manifold and $ds^{2}$ is a holomorphic metric, satysfying the following conditions
\begin{enumerate}[label=(\roman*)]
\item there exists a 2-dimensional holomorphic totally null self-dual integrable distribution given by the Pfaff system
\begin{equation}
\label{condition_1}
m_{A} \, g^{A \dot{B}} = 0  , \ m_{A} \ne 0
\end{equation}
\item the SD Weyl spinor $C_{ABCD}$ is algebraically degenerate and $m_{A}$ is a multiple Penrose spinor i.e.
\begin{equation}
\label{condition_2}
C_{ABCD} \, m^{A} m^{B} m^{C} =0 
\end{equation}
\end{enumerate}
\end{Definicja}
Definition \ref{definicja_weak_HH_spaces} appeared in \cite{Chudecki_Przanowski_Walkery} for the first time. It was also proved in \cite{Chudecki_Przanowski_Walkery} that the metric of a weak $\mathcal{HH}$-space can be brought to the form
\begin{equation}
\label{metric_weak_HH_expanding_space}
\frac{1}{2} ds^{2} = \phi^{-2} (-dp^{\dot{A}} dq_{\dot{A}} + Q^{\dot{A}\dot{B}} dq_{\dot{A}} dq_{\dot{B}}) = e^{1}e^{2} + e^{3} e^{4}
\end{equation}
where $(q_{\dot{A}}, p_{\dot{B}})$ are local complex coordinates. Coordinates $q_{\dot{A}}$ label the null string while $p_{\dot{A}}$ are coordinates on leafs of $\mathcal{C}_{m^{A}}$. $\phi$ and $Q^{\dot{A}\dot{B}} = Q^{(\dot{A}\dot{B})}$ are holomorphic functions of variables $(q^{\dot{A}}, p^{\dot{B}})$. The null tetrad $(e^{1}, e^{2}, e^{3}, e^{4})$ defined as follows
\begin{eqnarray}
\label{tetrada_spinorowa}
[ e^{3} , e^{1}  ] &=&  \phi^{-2} \, dq_{\dot{A}}
\\ \nonumber
[  e^{4} , e^{2}  ]  & =&  -dp^{\dot{A}} + Q^{\dot{A} \dot{B}} \, dq_{\dot{B}}
\end{eqnarray}
is called \textsl{Plebański tetrad}. The dual basis is given by
\begin{equation}
-\partial_{\dot{A}} = [\partial_{4}, \partial_{2} ] , \ 
\eth^{\dot{A}} = [ \partial_{3} , \partial_{1}], \ \Longrightarrow \ \partial_{A\dot{B}}  = \sqrt{2} \, [ \partial_{\dot{B}}, \eth_{\dot{B}}]
\end{equation}
where
\begin{eqnarray}
&& \partial_{\dot{A}} := \frac{\partial}{\partial p^{\dot{A}}}, \ \eth_{\dot{A}} := \phi^{2} \left( \frac{\partial}{\partial q^{\dot{A}}} - Q_{\dot{A}}^{\ \ \dot{B}} \partial_{\dot{B}} \right)
\\ \nonumber
&& \partial^{\dot{A}} := \frac{\partial}{\partial p_{\dot{A}}}, \ \eth^{\dot{A}} := \phi^{2} \left( \frac{\partial}{\partial q_{\dot{A}}} + Q^{\dot{A} \dot{B}} \partial_{\dot{B}} \right)
\end{eqnarray}
From the first structure equations we find the spinorial connection coefficients
\begin{eqnarray}
\label{weak_HH_connection}
&&\mathbf{\Gamma}_{111 \dot{D}} =0  \ \ \ \ \ \ \ \ \ \ \ \ \ \ \ \ \  \ \ \ \ \ \ \ \ \ \ \ \ \  \ \ \ \ \ \ 
\mathbf{\Gamma}_{112 \dot{D}} = -\sqrt{2} \, \phi^{-1} \, J_{\dot{D}}
\\ \nonumber
&&\mathbf{\Gamma}_{121 \dot{D}} = \frac{3}{\sqrt{2}} \phi^{-1} \, J_{\dot{D}} \ \ \ \ \ \ \ \ \ \ \ \ \ \ \ \ \ \  \ \ \ \ \ \ 
  \mathbf{\Gamma}_{122 \dot{D}} = - \frac{\phi}{\sqrt{2}} \left( \partial^{\dot{A}} (\phi \, Q_{\dot{A}\dot{D}}) - \frac{\partial \phi}{\partial q^{\dot{D}}} \right)
\\ \nonumber
&& \mathbf{\Gamma}_{221 \dot{D}} = \sqrt{2} \phi \, \left(  Q_{\dot{D}\dot{A}} J^{\dot{A}} + \frac{\partial \phi}{\partial q^{\dot{D}}}  \right)
\ \ \ \  \ \ \,
\mathbf{\Gamma}_{222 \dot{D}} = -\sqrt{2} \phi^{2} \, \eth^{\dot{A}} Q_{\dot{A}\dot{D}}
\\ \nonumber
&&\mathbf{\Gamma}_{\dot{A} \dot{B} 1 \dot{D}} = -\sqrt{2} \phi^{-1} \, J_{(\dot{A}} \in_{\dot{B})\dot{D}}
\\ \nonumber
&& \mathbf{\Gamma}_{\dot{A} \dot{B} 2 \dot{D}} =  \sqrt{2} \, \phi \left(  
\phi \, \partial_{(\dot{A}} Q_{\dot{B})\dot{D}} +  \in_{\dot{D}(\dot{B}} \left( Q_{\dot{A})\dot{C}} J^{\dot{C}} + \frac{\partial \phi}{\partial q^{\dot{A})}}\right)  \right) 
\end{eqnarray}
where
\begin{equation}
\label{spinor_J}
J_{\dot{A}} := \partial_{\dot{A}} \phi \equiv   \frac{\partial \phi}{\partial p^{\dot{A}}}
\ \ \ \ \ \Longrightarrow \ \ \ \ \ 
J^{\dot{A}} = -\partial^{\dot{A}} \phi \equiv - \frac{\partial \phi}{\partial p_{\dot{A}}}
\end{equation}
Nonzero SD curvature coefficients $C^{(i)}$ and the ASD Weyl spinor read
\begin{eqnarray}
\label{curv}
 C^{(5)}  &=& 0
\\ \nonumber
 C^{(4)}  &=& 0
\\ \nonumber
 C^{(3)}  &=& -\frac{1}{3} \, \phi^{2} \, 
\partial_{\dot{A}} \partial_{\dot{B}} Q^{\dot{A} \dot{B}}
\\ \nonumber
\frac{1}{2} \, C^{(2)}  &=& -\frac{1}{2} \, \phi^{5} \, 
\partial^{\dot{A}} (\phi^{-3} \, \eth^{\dot{B}} Q_{\dot{A} \dot{B}})
        -\frac{1}{2} \, \phi^{2} \, 
\eth_{\dot{A}} (\phi^{-3} \, \eth^{\dot{A}} \phi)
\\ \nonumber
\frac{1}{2} \, C^{(1)}  &=& 
-\phi^{4} \, \eth^{\dot{A}} (\phi^{-2} \, \eth^{\dot{B}} Q_{\dot{A} \dot{B}})  +
\phi^{4} \, (\eth^{\dot{A}} Q_{\dot{A} \dot{B}})
(\partial_{\dot{C}} Q^{\dot{B} \dot{C}})
\\ \nonumber
C_{\dot{A}\dot{B}\dot{C}\dot{D}} &=& - \phi^{2} \, 
\partial_{(\dot{A}}\partial_{\dot{B}} Q_{\dot{C} \dot{D})}
\end{eqnarray}
Curvature scalar $\mathcal{R}$ and the spinorial image of the traceless Ricci tensor $C_{AB \dot{C}\dot{D}}$ have the form
\begin{eqnarray}
\label{Einst}
- \frac{1}{2} \mathcal{R} &=&  
\phi^{2} \, \partial_{\dot{A}} \partial_{\dot{B}} Q^{\dot{A} \dot{B}} -
6 \phi^{3} \, \partial_{\dot{A}} (\phi^{-4} \, \eth^{\dot{A}} \phi)
\\ \nonumber
C_{11 \dot{A} \dot{B}} &=& - \phi^{-1} \, \partial_{\dot{A}} \partial_{\dot{B}} \phi
\\ \nonumber
C_{12 \dot{A} \dot{B}} &=& - \frac{1}{2}
\phi^{2} \, \partial_{(\dot{A}} \partial^{\dot{C}} Q_{\dot{B}) \dot{C}} -
\phi \, \partial_{(\dot{A}} (\phi^{-2} \, \eth_{\dot{B})} \phi)
\\ \nonumber
C_{22 \dot{A} \dot{B}} &=& 
- \phi^{5} \, \partial_{(\dot{A}} \phi^{-3} \, \eth^{\dot{C}} Q_{\dot{B}) \dot{C}}+
\phi \, (\eth^{\dot{C}} \phi)(\partial_{\dot{C}} Q_{\dot{A} \dot{B}})-
\phi \, \eth_{(\dot{A}} (\phi^{-2} \, \eth_{\dot{B})} \phi)
\end{eqnarray}

The SD congruence of null strings $\mathcal{C}^{e}_{m^{A}}$ is generated by the spinor $m_{A} \sim [0,m]$, $m \ne 0$, what implies that it is spanned by the vectors $(\partial_{2}, \partial_{4})$. In this sense Plebański tetrad is \textsl{adapted} to $\mathcal{C}^{e}_{m^{A}}$. The expansion of $\mathcal{C}^{e}_{m^{A}}$ reads
\begin{equation}
\label{parametry_kongruencji_SD}
M_{\dot{B}} =   -\sqrt{2} m  \phi^{-1} J_{\dot{B}} \ne 0
\end{equation}

In the rest of the paper we use the abbreviations
\begin{equation}
\label{skroty_wspolrzednych}
q^{\dot{1}} =: q, \ q^{\dot{2}} =: p, \ p^{\dot{1}} =: x, \ p^{\dot{2}} =: y, \ Q^{\dot{A} \dot{B}} =: \left[ \begin{array}{cc}
       \mathcal{A} \ & \mathcal{Q}   \\
       \mathcal{Q} \ & \mathcal{B}  
       \end{array} \right]
\end{equation}
with these abbreviations the metric (\ref{metric_weak_HH_expanding_space}) takes the form
\begin{equation}
\label{metric_weak_HH_expanding_space_abrevaitions}
\frac{1}{2} ds^{2} = \phi^{-2} (dq dy -dp dx + \mathcal{A} \, dp^{2} -2 \mathcal{Q} \, dqdp +  \mathcal{B} \, dq^{2}) 
\end{equation}
In what follows the coordinate system $(q,p,x,y)$ is called \textsl{the hyperheavenly coordinate system}.

Note, that the metric (\ref{metric_weak_HH_expanding_space_abrevaitions}) depends on four functions of four variables.

\subsubsection{Gauge transformations}

The metric (\ref{metric_weak_HH_expanding_space}) remains invariant under the following transformations of the coordinates
\begin{equation}
\label{gauge}
q'^{\dot{A}} = q'^{\dot{A}} (q^{\dot{M}}), \ p'^{\dot{A}} = \lambda^{-1} \, D^{-1 \ \ \dot{A}}_{\ \ \; \dot{B}} \, p^{\dot{B}} + \sigma^{\dot{A}}
\end{equation}
where $\lambda=\lambda(q_{\dot{B}})$ and $\sigma^{\dot{A}}=\sigma^{\dot{A}}(q_{\dot{B}})$ are arbitrary gauge functions and 
\begin{eqnarray}
\nonumber
D_{\dot{A}}^{\ \ \dot{B}} 
&:=& \frac{\partial q'_{\dot{A}}}{\partial q_{\dot{B}}} 
= \Delta \, \frac{\partial q^{\dot{B}}}{\partial q'^{\dot{A}}}
= \lambda \Delta \, \frac{\partial p'_{\dot{A}}}{\partial p_{\dot{B}}}
= \lambda^{-1} \, \frac{\partial p^{\dot{B}}}{\partial p'^{\dot{A}}} , \ \Delta := \det 
\left( \frac{\partial q'_{\dot{A}}}{\partial q_{\dot{B}}} \right)
 = \frac{1}{2} \, D_{\dot{A}\dot{B}} D^{\dot{A}\dot{B}}
\\ 
D^{-1 \ \ \dot{B}}_{\ \ \; \dot{A}} 
&=& \frac{\partial q_{\dot{A}}}{\partial q'_{\dot{B}}}
= \Delta^{-1} \, \frac{\partial q'^{\dot{B}}}{\partial q^{\dot{A}}}
= \lambda \, \frac{\partial p'^{\dot{B}}}{\partial p^{\dot{A}}}
= \lambda^{-1} \Delta^{-1} \, \frac{\partial p_{\dot{A}}}{\partial p'_{\dot{B}}}
\end{eqnarray}
Hence
\begin{eqnarray}
 D^{\, \dot{A}}_{\ \ \dot{B}} &:=& D_{\dot{M}}^{\ \ \dot{N}} \in^{\dot{M} \dot{A}} \in_{\dot{B} \dot{N}} = - \Delta \, D^{-1 \ \ \dot{A}}_{\ \ \; \dot{B}}
\\ \nonumber
D^{-1 \; \dot{A}}_{\ \ \ \ \ \dot{B}} &:=& D^{-1 \ \ \dot{N}}_{\ \ \; \dot{M}} \in^{\dot{M} \dot{A}} \in_{\dot{B} \dot{N}} = - \frac{1}{\Delta} \, D_{\dot{B}}^{\ \ \dot{A}}
\end{eqnarray}
Functions $\phi$ and $Q^{\dot{A}\dot{B}}$ transform under (\ref{gauge}) as follows
\begin{eqnarray}
\label{transformacja_Q}
\phi' &=& \lambda^{-\frac{1}{2}} \, \phi
\\ \nonumber
Q'^{\dot{A} \dot{B}} &=& \lambda^{-1} \, D^{-1 \ \ \dot{A}}_{\ \ \; \dot{R}} \, D^{-1 \ \ \dot{B}}_{\ \ \; \dot{S}}
Q^{\dot{R} \dot{S}} + D^{-1 \ \ ( \dot{A}}_{\ \ \; \dot{R}} \, 
\frac{\partial p'^{\dot{B})}}{\partial q_{\dot{R}}} 
\end{eqnarray}
Transformations (\ref{gauge}) are equivalent to the spinorial transformations 
\begin{eqnarray}
\label{transformacja_spinorowa}
L^{A}_{\ \; B} &=& \left[ \begin{array}{cc}
       \lambda^{-1} \Delta^{-\frac{1}{2}} \ \  & h \phi^{2}  \Delta^{\frac{1}{2}}   \\
       0 \ \  &  \lambda \Delta^{\frac{1}{2}}  
       \end{array} \right], \ 2 h := \frac{\partial \sigma^{\dot{R}}}{\partial q'^{\dot{R}}} - \frac{1}{\lambda^{2}\Delta}  \frac{\partial \lambda}{\partial q^{\dot{R}}} \, p^{\dot{R}}
\\ \nonumber
M^{\dot{A}}_{\ \; \dot{B}} &=& \Delta^{\frac{1}{2}} \, D^{-1 \ \ \dot{A}}_{\ \ \; \dot{B}}
\end{eqnarray}
where $L^{A}_{\ \; B}, \ M^{\dot{A}}_{\ \; \dot{B}} \in SL(2, \mathbb{C})$. Dotted and undotted spinors transform as follows
\begin{equation}
\label{jawne_wzory_na_transformacje_spinorow}
m'^{A_{1}...A_{n}} = L^{A_{1}}_{\ \; R_{1}} \, ... \, L^{A_{n}}_{\ \; R_{n}} \, m^{R_{1}...R_{n}}, \ 
m'^{\dot{A}_{1}...\dot{A}_{n}} = M^{\dot{A}_{1}}_{\ \; \dot{R}_{1}} \, ... \, M^{\dot{A}_{n}}_{\ \; \dot{R}_{n}} \, m^{\dot{R}_{1}...\dot{R}_{n}}
\end{equation}

\subsubsection{ASD null string equations} 
\label{ASD_null_strings}

Assume that a weak $\mathcal{HH}$-space admits an additional structure which is ASD $\mathcal{C}^{e}$. Let this congruence be expanding in general and let it be generated by a spinor $m_{\dot{A}}$ with an expansion given by $M_{A}$. ASD null string equations read
\begin{equation}
\label{ASD_null_strings_eqs}
m_{\dot{A}} M_{A} = m^{\dot{B}} \nabla_{A \dot{A}} m_{\dot{B}} \equiv m^{\dot{B}} \partial_{A\dot{A}} m_{\dot{B}} + m^{\dot{S}} m^{\dot{B}} \mathbf{\Gamma}_{\dot{S}\dot{B} A \dot{A}}
\end{equation}
Explicitly, Eqs. (\ref{ASD_null_strings_eqs}) take the form
\begin{subequations}
\label{ASD_null_strings_eqs_jawnie_ekspansja}
\begin{eqnarray}
\frac{\phi}{\sqrt{2}} m_{\dot{A}} M_{1} &=& \phi \, m^{\dot{B}} \frac{\partial m_{\dot{B}}}{\partial p^{\dot{A}}} + m_{\dot{A}} m^{\dot{S}} J_{\dot{S}}
\\ 
\frac{1}{\sqrt{2}\phi} m_{\dot{A}} M_{2} &=& \phi \, m^{\dot{B}} \frac{\partial m_{\dot{B}}}{\partial q^{\dot{A}}} + \phi  \, m^{\dot{S}} \frac{\partial}{\partial p^{\dot{S}}} ( m^{\dot{B}} Q_{\dot{A}\dot{B}} ) - \phi  \, m^{\dot{B}} Q_{\dot{A}\dot{B}} \frac{\partial m^{\dot{C}}}{\partial p^{\dot{C}}} 
\\  \nonumber
&& +m_{\dot{A}} m^{\dot{B}} Q_{\dot{C}\dot{B}} J^{\dot{C}} + m_{\dot{A}} m^{\dot{S}} \frac{\partial \phi}{\partial q^{\dot{S}}}
\end{eqnarray}
\end{subequations}
Eqs. (\ref{ASD_null_strings_eqs_jawnie_ekspansja}) contracted with $m^{\dot{A}}$ give ASD null string equations \textsl{sensu stricto}
\begin{subequations}
\label{ASD_null_strings_eqs_jawnie}
\begin{eqnarray}
\label{ASD_null_strings_eqs_jawnie_1}
&& m^{\dot{A}} m^{\dot{B}} \frac{\partial m_{\dot{A}}}{\partial p^{\dot{B}}} = 0
\\ 
\label{ASD_null_strings_eqs_jawnie_2}
&& m^{\dot{A}} m^{\dot{B}} \frac{\partial m_{\dot{A}}}{\partial q^{\dot{B}}} + m^{\dot{S}} \frac{\partial \mathcal{Y}}{\partial p^{\dot{S}}} - \mathcal{Y} \left( 2 \frac{\partial m^{\dot{S}}}{\partial p^{\dot{S}}} + \frac{1}{\sqrt{2}} M_{1} - \frac{1}{\phi} J_{\dot{S}} m^{\dot{S}}  \right)  =0
\\ \nonumber
&&  \mathcal{Y} := m^{\dot{A}} m^{\dot{B}} Q_{\dot{A}\dot{B}}
\end{eqnarray}
\end{subequations}
ASD null string equations (\ref{ASD_null_strings_eqs}) remain invariant if a spinor $m_{A}$ which generates $\mathcal{C}$ is re-scalled. Hence, $\mathcal{C}_{m^{\dot{A}}}$ and $\mathcal{C}_{m'^{\dot{A}}}$ are the same congruences if $m^{\dot{A}} \sim m'^{\dot{A}}$. Consequently, the spinor $m_{\dot{A}}$ can be always re-scalled to the form with $m_{\dot{1}}=1$ or $m_{\dot{2}} =1$ so it has only one component to be determined via Eq. (\ref{ASD_null_strings_eqs_jawnie_1}). Solution of Eq. (\ref{ASD_null_strings_eqs_jawnie_2}) gives $\mathcal{Y}$. Having $m_{\dot{A}}$ and $\mathcal{Y}$ one returns to Eqs. (\ref{ASD_null_strings_eqs_jawnie_ekspansja}) and calculates the expansion $M_{A}$.

An expansion $\theta$ and a twist $\varrho$ of $\mathcal{I} (\mathcal{C}^{e}_{m^{A}}, \mathcal{C}_{m^{\dot{A}}})$ read
\begin{equation}
\label{skalary_optyczne}
\theta \sim mM_{1} +  \frac{\sqrt{2} m}{\phi} m_{\dot{A}} \frac{\partial \phi}{\partial p_{\dot{A}}}, \ 
\varrho \sim mM_{1} -  \frac{\sqrt{2} m}{\phi} m_{\dot{A}} \frac{\partial \phi}{\partial p_{\dot{A}}}
\end{equation}

In the next Sections we follow the pattern of Ref. \cite{Chudecki_Kahler_1}. We equip a weak $\mathcal{HH}$-space with "the first" ASD $\mathcal{C}$ generated by a spinor $m_{\dot{A}} \sim [z,1]$, $z=z(q,p,x,y)$, with an expansion $M_{A} \ne 0$ (this congruence is denoted by $\mathcal{C}_{m^{\dot{A}}}$, Section \ref{subsekcja_dege_x_anye}). Then we consider a case with $M_{A}=0$ (Section \ref{subsekcja_dege_x_degn}). In Section \ref{subsekcja_dege_x_degne} a space is equipped with "the second" $\mathcal{C}$ generated by a spinor $n_{\dot{A}} \sim [1,w]$, $w=w(q,p,x,y)$, with an expansion $N_{A} \ne 0$ (we denote this congruence by $\mathcal{C}_{n^{\dot{A}}}$). Of course,  $\mathcal{C}_{m^{\dot{A}}}$ and $\mathcal{C}_{n^{\dot{A}}}$ are distinct so $m^{\dot{A}}n_{\dot{A}} \sim zw-1 \ne 0$. Finally we arrive at a case with ASD $\mathcal{C}^{nn}$ (Sections \ref{subsekcja_dege_x_degnn} and \ref{sekcja_para_Kahler_complicated}).

We believe that Scheme \ref{Structure} clearly presents the structure of our article.

\subsubsection{ASD congruence $\mathcal{C}^{e}_{m^{\dot{A}}}$}

Inserting $m_{\dot{A}} \sim [z,1]$ into Eqs. (\ref{ASD_null_strings_eqs_jawnie}), (\ref{ASD_null_strings_eqs_jawnie_ekspansja}) and (\ref{skalary_optyczne}) one finds
\begin{subequations}
\label{kongruencja_mdotA_rownania}
\begin{eqnarray}
\label{kongruencja_mdotA_rownania_1}
&& z_{x} - zz_{y}=0
\\ 
\label{kongruencja_mdotA_rownania_2}
&& z_{q} - zz_{p} - z_{y} \mathcal{Z} + z \frac{\partial \mathcal{Z}}{\partial y} - \frac{\partial \mathcal{Z}}{\partial x} =0 , \ \mathcal{Z} := \mathcal{B} + 2z \mathcal{Q} + z^{2} \mathcal{A}
\end{eqnarray}
\end{subequations}
Expansion of $\mathcal{C}^{e}_{m^{\dot{A}}}$ reads
\begin{subequations}
\label{ekspansja_pierwszej_ASD_struny}
\begin{eqnarray}
\label{ekspansja_pierwszej_ASD_struny_1}
\frac{\phi}{\sqrt{2}} M_{1} &=& -\phi z_{y} + z\phi_{y} - \phi_{x}
\\ 
\label{ekspansja_pierwszej_ASD_struny_2}
\frac{1}{\sqrt{2}\phi}  M_{2} &=& z \phi_{p}- \phi z_{p}  - \phi_{q} + \phi z \frac{\partial}{\partial y} (\mathcal{Q} +z \mathcal{A}) - \phi \frac{\partial}{\partial x} (\mathcal{Q} +z \mathcal{A})  
\\ \nonumber
&&  + (\mathcal{Q} +z \mathcal{A}) (\phi_{x} - \phi z_{y}) + (\mathcal{B} + z \mathcal{Q}) \phi_{y}
\end{eqnarray}
\end{subequations}
Optical scalars are given by the formulas
\begin{equation}
\label{optyczne_skalary_pierwszego_przeciecia}
\mathcal{I} (\mathcal{C}^{e}_{m^{A}}, \mathcal{C}^{e}_{m^{\dot{A}}}): \ \ \theta_{1} \sim  \phi z_{y} - 2z \phi_{y} + 2\phi_{x}, \ \varrho_{1} \sim z_{y}
\end{equation}
Transformation rule for $z$ follows from (\ref{jawne_wzory_na_transformacje_spinorow}) and it yields
\begin{equation}
\label{transformacja_na_z}
z' = \frac{p'_{q}-  p'_{p}z}{q'_{p} z - q'_{q}}
\end{equation}

\begin{Lemat}
\label{Lemat_kongruencja_pierwsza_nontwisting}
Let $\mathcal{I} (\mathcal{C}^{e}_{m^{A}}, \mathcal{C}^{e}_{m^{\dot{A}}})$ be nontwisting, $\varrho_{1}=0$. Then $z=0$ and the following equations hold true
\begin{eqnarray}
\label{uproszczone_formuly_dla_z_rowne_0}
&&  \mathcal{B}_{x}=0, \ \theta_{1} \sim \phi_{x},
\\ \nonumber
&&  M_{1} \sim \phi_{x}, \ M_{2} \sim -\phi \mathcal{Q}_{x} + \phi_{x} \mathcal{Q} + \phi_{y} \mathcal{B} - \phi_{q}
\end{eqnarray}
\end{Lemat}
\begin{proof}
If $\varrho_{1}=0$ then $z_{y}=0$ and from (\ref{kongruencja_mdotA_rownania_1}) it follows that $z_{x}=0$. Hence, $z=z(p,q)$. Finally $z$ can be gauged away without any loss of generality (compare (\ref{transformacja_na_z})). With $z=0$, Eqs. (\ref{kongruencja_mdotA_rownania}), (\ref{ekspansja_pierwszej_ASD_struny}) and (\ref{optyczne_skalary_pierwszego_przeciecia}) reduce to (\ref{uproszczone_formuly_dla_z_rowne_0}).
\end{proof}

\begin{Uwaga}
\normalfont
If $z=0$ then the gauge (\ref{gauge}) is restricted to the condition $p'_{q}=0$.
\end{Uwaga}

\begin{Lemat} 
\label{Lemat_kongruencja_pierwsza_twisting}
Let $\mathcal{I} (\mathcal{C}^{e}_{m^{A}}, \mathcal{C}^{e}_{m^{\dot{A}}})$ be twisting, $\varrho_{1} \ne 0$. Then 
\begin{equation}
\label{rozwiazanie_na_y}
y=-xz + \Sigma(q,p,z)
\end{equation}
where $\Sigma = \Sigma (q,p,z)$ is an arbitrary function. Moreover
\begin{equation}
\label{rozwiazanie_rownania_pierwszej_struny_2}
\mathcal{B} + 2z \mathcal{Q} + z^{2} \mathcal{A} = \Omega \, (x-\Sigma_{z}) + z \Sigma_{p} - \Sigma_{q}, \ \Omega = \Omega (q,p,z)
\end{equation}
and
\begin{subequations}
\label{pomocnicze_oggolne}
\begin{eqnarray}
 \label{pomocnicze_1}
 \theta_{1} &\sim& \phi - 2 \phi_{x} (x-\Sigma_{z}), 
\\ \label{pomocnicze_2}
 M_{1} &\sim& \phi - \phi_{x} (x-\Sigma_{z}),
\\ \label{pomocnicze_3}
 M_{2} &\sim& z \phi_{p} +  \frac{\phi \Sigma_{p}}{\Sigma_{z}-x} - \phi_{q} - \phi \frac{\partial}{\partial x} (\mathcal{Q} + z \mathcal{A}) 
 + (\mathcal{Q} + z \mathcal{A}) \left( \phi_{x} - \frac{\phi}{\Sigma_{z}-x} \right) -\Omega \phi_{z} \ \ \ \ \ \ \ \
\end{eqnarray}
\end{subequations}
where $\Omega = \Omega (q,p,z)$ is an arbitrary function. 
\end{Lemat}
\begin{proof}
Multiplying (\ref{kongruencja_mdotA_rownania_1}) by $dx \wedge dy \wedge dq \wedge dp$ and treating $y$ as a function of $(q,p, x,z)$ one arrives at (\ref{rozwiazanie_na_y}). Hence 
\begin{equation}
z_{x} = \frac{z}{\Sigma_{z}-x}, \ z_{y}  = \frac{1}{\Sigma_{z}-x}, \ z_{q} = -\frac{\Sigma_{q}}{\Sigma_{z}-x}, \ 
z_{p}  = -\frac{\Sigma_{p}}{\Sigma_{z}-x}
\end{equation}
Eq. (\ref{rozwiazanie_na_y}) suggests that $z$ should be treated as a new variable. Thus we pass to a new coordinate system $(q,p,x,z)$. Denote
\begin{equation}
\tilde{\mathcal{Z}} = \tilde{\mathcal{Z}} (q,p,x,z) = \tilde{\mathcal{Z}} (q,p,x,y(q,p,x,z)) =  \mathcal{Z} (q,p,x,y)
\end{equation}
then the derivatives of $\mathcal{Z}$ transform as follows
\begin{equation}
\mathcal{Z}_{y} = \tilde{\mathcal{Z}}_{z} z_{y}, \ \mathcal{Z}_{x} = \tilde{\mathcal{Z}}_{x}+\tilde{\mathcal{Z}}_{z} z_{x}, \ \mathcal{Z}_{q} = \tilde{\mathcal{Z}}_{q}+\tilde{\mathcal{Z}}_{z} z_{q}, \ \mathcal{Z}_{p} = \tilde{\mathcal{Z}}_{p}+\tilde{\mathcal{Z}}_{z} z_{p}
\end{equation}
Eq. (\ref{kongruencja_mdotA_rownania_2}) takes the form
\begin{equation}
\Sigma_{q} - z \Sigma_{p} + \tilde{\mathcal{Z}} + (\Sigma_{z} - x) \tilde{\mathcal{Z}}_{x} = 0
\end{equation} 
which is solved by (\ref{rozwiazanie_rownania_pierwszej_struny_2}). Analogously one finds (\ref{pomocnicze_oggolne}). 
\end{proof}

\begin{Uwaga}
\normalfont
In (\ref{rozwiazanie_rownania_pierwszej_struny_2}) and (\ref{pomocnicze_oggolne}) functions $\phi$, $\mathcal{A}$, $\mathcal{Q}$ and $\mathcal{B}$ are understood as functions of $(q,p,x,z)$. 
\end{Uwaga}

Considerations of this section should be supplemented by the transformation formula for $\Sigma$
\begin{equation}
\label{transformacja_na_Sigma_ogolna}
\Sigma' = - \frac{\lambda^{-1}}{q'_{p}z-q'_{q}} \Sigma + \frac{p'_{q} - p'_{p}z}{q'_{p}z-q'_{q}} \sigma^{\dot{1}} + \sigma^{\dot{2}}
\end{equation}
Transformation for $\Omega$ is much more complicated and its explicit form is not necessary at this stage of considerations.

\subsubsection{ASD congruence $\mathcal{C}^{e}_{n^{\dot{A}}}$}

"The second" ASD congruence denoted by $\mathcal{C}_{n^{\dot{A}}}^{e}$ is generated by a spinor $n_{\dot{A}}$ such that $n_{\dot{A}} \sim [1,w]$. ASD null string equations for $\mathcal{C}_{n^{\dot{A}}}^{e}$ have the form (\ref{ASD_null_strings_eqs_jawnie}) and (\ref{ASD_null_strings_eqs_jawnie_ekspansja}) with $m_{\dot{A}} \rightarrow n_{\dot{A}}$ and $M_{A} \rightarrow N_{A}$. Hence
\begin{subequations}
\label{kongruencja_CndotA}
\begin{eqnarray}
\label{kongruencja_CndotA_1}
&& w_{y}-ww_{x} = 0
\\ 
\label{kongruencja_CndotA_2}
&& w_{p} - ww_{q} +w_{x} \mathcal{W} -w \frac{\partial \mathcal{W}}{\partial x} + \frac{\partial \mathcal{W}}{\partial y} = 0, \ \mathcal{W} := \mathcal{A} + 2w \mathcal{Q} + w^{2} \mathcal{B}
\end{eqnarray}
\end{subequations}
Expansion of $\mathcal{C}^{e}_{n^{\dot{A}}}$ reads
\begin{subequations}
\label{ekspansja_CndotA}
\begin{eqnarray}
\label{ekspansja_CndotA_1}
\frac{\phi}{\sqrt{2}} N_{1} &=& \phi w_{x} - w \phi_{x} + \phi_{y}
\\
\label{ekspansja_CndotA_2}
\frac{1}{\sqrt{2} \phi} N_{2} &=& \phi w_{q} - w \phi_{q} + \phi_{p}+ \phi w \frac{\partial}{\partial x} (\mathcal{Q} + w \mathcal{B}) - \phi \frac{\partial}{\partial y} (\mathcal{Q} + w \mathcal{B})
\\ \nonumber
&& (\phi_{y} - \phi w_{x}) (\mathcal{Q} + w \mathcal{B}) + (\mathcal{A} + w \mathcal{Q}) \phi_{x}
\end{eqnarray}
\end{subequations}
Optical scalars of $\mathcal{I} (\mathcal{C}^{e}_{m^{A}}, \mathcal{C}^{e}_{n^{\dot{A}}})$ take the form
\begin{equation}
\label{optyczne_skalary_CndotA}
\mathcal{I} (\mathcal{C}^{e}_{m^{A}}, \mathcal{C}^{e}_{n^{\dot{A}}}): \ \ \theta_{2} \sim \phi w_{x} -2w \phi_{x} + 2 \phi_{y}, \ \varrho_{2} \sim w_{x}
\end{equation} 
and, finally, the transformation rule for the function $w$ reads
\begin{equation}
w' = \frac{w q'_{q} - q'_{p}}{p'_{p} - p'_{q} w}
\end{equation}

\begin{Lemat}
\label{Lemat_kongruencja_druga_nontwisting}
Let $\mathcal{I} (\mathcal{C}^{e}_{m^{A}}, \mathcal{C}^{e}_{n^{\dot{A}}})$ be nontwisting, $\varrho_{2}=0$. Then $w=0$ and the following equations hold true
\begin{eqnarray}
\label{uproszczone_formuly_dla_w_rowne_0}
&&  \mathcal{A}_{y}=0, \ \theta_{2} \sim \phi_{y},
\\ \nonumber
&&  N_{1} \sim \phi_{y}, \ N_{2} \sim -\phi \mathcal{Q}_{y} + \phi_{y} \mathcal{Q} + \phi_{x} \mathcal{A} + \phi_{p}
\end{eqnarray}
\end{Lemat}
\begin{proof}
We skip the proof due to its similarity to that of Lemma \ref{Lemat_kongruencja_pierwsza_nontwisting}.
\end{proof}
\begin{Uwaga}
\normalfont
If $w=0$ then the gauge (\ref{gauge}) is restricted to the condition $q'_{p}=0$.
\end{Uwaga}
\begin{Lemat} 
\label{Lemat_kongruencja_druga_twisting}
Let $\mathcal{I} (\mathcal{C}^{e}_{m^{A}}, \mathcal{C}^{e}_{n^{\dot{A}}})$ be twisting, $\varrho_{2} \ne 0$. Then 
\begin{equation}
\label{rozwiazanie_na_x}
x=-wy + H(q,p,w)
\end{equation}
where $H = H (q,p,w)$ is an arbitrary function. Moreover
\begin{equation}
\label{rozwiazanie_drugiego_drug_rownanianstrun}
\mathcal{A} + 2w \mathcal{Q} + w^{2} \mathcal{B} = R \, (y-H_{w}) + H_{p} - w H_{q}, \ R=R(q,p,w)
\end{equation}
and
\begin{subequations}
\label{pomocnicze_oggolne_drugie}
\begin{eqnarray}
 \label{wiezy_drugiej_kongruencji_1}
 \theta_{2} &\sim & \phi - 2 \phi_{y} (y-H_{w}), 
\\ \label{wiezy_drugiej_kongruencji}
 N_{1} &\sim & \phi -  \phi_{y} (y-H_{w}),
\\ \label{wiezy_drugiej_kongruencji_3}
 N_{2} &\sim & \phi_{p} - w \phi_{q} - \frac{\phi H_{q}}{H_{w}-y} - \phi \frac{\partial}{\partial y} (\mathcal{Q} + w \mathcal{B}) 
+ (\mathcal{Q} + w \mathcal{B}) \left( \phi_{y} - \frac{\phi}{H_{w}-y} \right) -R \phi_{w} \ \ \ \ \ \ \ \ \ 
\end{eqnarray}
\end{subequations}
where $R = R (q,p,w)$ is an arbitrary function. 
\end{Lemat}
\begin{proof}
We skip the proof due to its similarity to that of Lemma \ref{Lemat_kongruencja_pierwsza_twisting}.
\end{proof}
\begin{Uwaga}
\normalfont
In (\ref{rozwiazanie_drugiego_drug_rownanianstrun}) and (\ref{pomocnicze_oggolne_drugie}) functions $\phi$, $\mathcal{A}$, $\mathcal{Q}$ and $\mathcal{B}$ are understood as functions of $(q,p,w,y)$. 
\end{Uwaga}
Note that
\begin{equation}
  w_{x} = \frac{1}{H_{w}-y}, \ w_{y} = \frac{w}{H_{w}-y}, \ w_{q} = -\frac{H_{q}}{H_{w}-y}, \ w_{p} = -\frac{H_{p}}{H_{w}-y}
\end{equation}
and the transformation formula for $H$ reads
\begin{equation}
\label{transformacja_na_H}
H' = \frac{\lambda^{-1}}{p'_{p} - p'_{q}w} H + \frac{q'_{q} w-q'_{p}}{p'_{p} - p'_{q}w} \sigma^{\dot{2}} + \sigma^{\dot{1}}
\end{equation}
The transformation formula for $R$ is more advanced and we do not need it at this stage.

\subsection{Spaces of the types $[\textrm{deg}]^{e} \otimes [\textrm{any}]^{e}$} 
\label{subsekcja_dege_x_anye}

In this Section we assume that a weak $\mathcal{HH}$-space admits $\mathcal{C}_{m^{\dot{A}}}^{e}$.

\subsubsection{Spaces of the types $\{ [\textrm{deg}]^{e} \otimes [\textrm{any}]^{e},[+-] \}$ and $\{ [\textrm{deg}]^{e} \otimes [\textrm{any}]^{e},[--] \}$}

\begin{Twierdzenie} 
\label{theorem_dege_x_anye_mm_pm}
Let $(\mathcal{M}, ds^{2})$ be a complex (neutral) space of the type $\{ [\textrm{deg}]^{e} \otimes [\textrm{any}]^{e}, [--] \}$ or $\{ [\textrm{deg}]^{e} \otimes [\textrm{any}]^{e}, [+-] \}$. Then there exists a local coordinate system $(q,p,x,y)$ such that the metric takes the form
\begin{equation}
\label{metryka_dege_x_anye_mm_pm}
\frac{1}{2} ds^{2} = \phi^{-2} \left( dqdy - dpdx + \mathcal{A} \, dp^{2} - 2 \mathcal{Q} \, dqdp + \mathcal{B} \, dq^{2} \right)
\end{equation}
where $\phi = \phi (q,p,x,y)$, $\mathcal{A} = \mathcal{A} (q,p,x,y)$, $\mathcal{Q} = \mathcal{Q} (q,p,x,y)$, $\mathcal{B} = \mathcal{B} (q,p,y)$ are arbitrary holomorphic (real smooth) functions such that
\begin{subequations}
\begin{eqnarray}
\label{Twierdzenieeeee_1_1}
\textrm{for the type } \{ [\textrm{deg}]^{e} \otimes [\textrm{any}]^{e}, [--] \}: && \phi_{x}=0, \ \phi \mathcal{Q}_{x} \ne \phi_{y} \mathcal{B} - \phi_{q}
\\ 
\label{Twierdzenieeeee_1_2}
\textrm{for the type } \{ [\textrm{deg}]^{e} \otimes [\textrm{any}]^{e}, [+-] \}: && \phi_{x} \ne 0
\end{eqnarray}
\end{subequations}
\end{Twierdzenie}
\begin{proof}
If $\varrho_{1}=0$ (twist of $\mathcal{I}(\mathcal{C}_{m^{A}}^{e}, \mathcal{C}_{m^{\dot{A}}}^{e})$ vanishes) then Lemma \ref{Lemat_kongruencja_pierwsza_nontwisting} holds true. From (\ref{uproszczone_formuly_dla_z_rowne_0}) it follows that $\mathcal{B} = \mathcal{B} (q,p,y)$. Inserting this into (\ref{metric_weak_HH_expanding_space_abrevaitions}) one arrives at (\ref{metryka_dege_x_anye_mm_pm}).

For a subcase $[+-]$ the expansion $\theta_{1}$ is nonzero what implies that $\phi_{x} \ne 0$ and $M_{1} \ne 0$. This proves (\ref{Twierdzenieeeee_1_2}). For a subcase $[--]$ we have $\phi_{x} =0$, but $\phi \mathcal{Q}_{x} \ne \phi_{y} \mathcal{B} - \phi_{q}$, otherwise $M_{2}=0$ and a space is no longer equipped with an expanding ASD $\mathcal{C}$. Thus, (\ref{Twierdzenieeeee_1_1}) is proved.
\end{proof}

\subsubsection{Spaces of the types $\{ [\textrm{deg}]^{e} \otimes [\textrm{any}]^{e},[++] \}$ and $\{ [\textrm{deg}]^{e} \otimes [\textrm{any}]^{e},[-+] \}$}

\begin{Twierdzenie} 
\label{theorem_dege_x_anye_mp_pp}
Let $(\mathcal{M}, ds^{2})$ be a complex (neutral) space of the type $\{ [\textrm{deg}]^{e} \otimes [\textrm{any}]^{e}, [-+] \}$ or $\{ [\textrm{deg}]^{e} \otimes [\textrm{any}]^{e}, [++] \}$. Then there exists a local coordinate system $(q,p,x,z)$ such that the metric takes the form
\begin{eqnarray}
\label{metryka_dege_x_anye_mp_pp}
\frac{1}{2} ds^{2} &=& \phi^{-2} \Big( -dpdx -x \, dqdz -z \, dq dx +  ( \Sigma_{z} dz + \Sigma_{p} dp) \, dq + \mathcal{A} \, dp^{2}   \ \ \ \ \ \  
\\ \nonumber
&&  \ \ \ \ \ \ \ - 2 \mathcal{Q} \, dqdp - \left( 2z \mathcal{Q} + z^{2} \mathcal{A} - \Omega \, (x-\Sigma_{z}) -z \Sigma_{p} \right) dq^{2} \Big)
\end{eqnarray}
where $\phi = \phi (q,p,x,z)$, $\mathcal{A} = \mathcal{A} (q,p,x,z)$, $\mathcal{Q} = \mathcal{Q} (q,p,x,z)$, $\Sigma = \Sigma (q,p,z)$ and $\Omega = \Omega (q,p,z)$ are arbitrary holomorphic (real smooth) functions such that
\begin{subequations}
\begin{eqnarray}
\label{Twierdzenieeeee_2_1}
&& \textrm{for the type } \{ [\textrm{deg}]^{e} \otimes [\textrm{any}]^{e}, [-+] \}:  \phi=F \, (x-\Sigma_{z})^{\frac{1}{2}}, \ F=F (q,p,z) \ne 0
\\  
\label{Twierdzenieeeee_2_2}
&& \textrm{for the type } \{ [\textrm{deg}]^{e} \otimes [\textrm{any}]^{e}, [++] \}:  \phi - 2 \phi_{x} (x-\Sigma_{z}) \ne 0 
\\ \nonumber
&& \qquad \qquad \qquad \qquad \qquad \qquad \qquad \qquad  \qquad \textrm{ and } 
\\ \nonumber
&& \qquad \qquad \qquad \qquad \qquad \qquad  \Big[ \phi - \phi_{x} (x-\Sigma_{z}) \ne 0 \textrm{    or } 
\\ \nonumber
&& z \phi_{p} +  \frac{\phi \Sigma_{p}}{\Sigma_{z}-x} - \phi_{q} - \phi \frac{\partial}{\partial x} (\mathcal{Q} + z \mathcal{A}) + (\mathcal{Q} + z \mathcal{A}) \left( \phi_{x} - \frac{\phi}{\Sigma_{z}-x} \right) -\Omega \phi_{z} \ne 0 \Big]
\end{eqnarray}
\end{subequations}
\end{Twierdzenie}
\begin{proof}
If twist of $\mathcal{I}(\mathcal{C}_{m^{A}}^{e}, \mathcal{C}_{m^{\dot{A}}}^{e})$ is nonzero ($\varrho_{1} \ne 0$) then Lemma \ref{Lemat_kongruencja_pierwsza_twisting} holds true. Eq. (\ref{rozwiazanie_rownania_pierwszej_struny_2}) gives a solution for $\mathcal{B}$ in terms of $\mathcal{A}$, $\mathcal{Q}$, $\Sigma$ and $\Omega$. Also, (\ref{rozwiazanie_na_y}) gives $y$ in terms of $(q,p,z)$. Putting this into (\ref{metric_weak_HH_expanding_space_abrevaitions}) one arrives at (\ref{metryka_dege_x_anye_mp_pp}). 

For a subtype $[-+]$ one has $\theta_{1}=0 \ \Longrightarrow \ \phi = F(q,p,z) (x-\Sigma_{z})^{\frac{1}{2}}$ where $F=F(q,p,z)$ is an arbitrary nonzero function. Thus, (\ref{Twierdzenieeeee_2_1}) is proved. For a subtype $[++]$ we have $\theta_{1} \ne 0$ and ($M_{1} \ne 0$ or $M_{2} \ne 0$). It proves (\ref{Twierdzenieeeee_2_2}).
\end{proof}

\subsection{Spaces of the types $[\textrm{deg}]^{e} \otimes [\textrm{deg}]^{n}$} 
\label{subsekcja_dege_x_degn}

In this Section we assume that a weak $\mathcal{HH}$-space admits $\mathcal{C}_{m^{\dot{A}}}^{n}$. It means that $M_{A}=0$. If $\varrho_{1}=0$ then from Lemma \ref{Lemat_kongruencja_pierwsza_nontwisting} it follows that $M_{1}=0 \ \Longrightarrow \ \theta_{1}=0$. In this case we deal with the type $\{ [\textrm{deg}]^{e} \otimes [\textrm{deg}]^{n},[--] \}$. On the other hand, if $\varrho_{1} \ne 0$ then from Lemma \ref{Lemat_kongruencja_pierwsza_twisting} it follows that $M_{1}=0 \ \Longrightarrow \ \theta_{1} \ne 0$. Thus, the second  possible type is $\{ [\textrm{deg}]^{e} \otimes [\textrm{deg}]^{n},[++] \}$. Note, that the fact that $M_{A}=0$ automatically implies that ASD Weyl spinor is algebraically special. 

\begin{Twierdzenie} 
\label{theorem_dege_x_degn_mm}
Let $(\mathcal{M}, ds^{2})$ be a complex (neutral) space of the type $\{ [\textrm{deg}]^{e} \otimes [\textrm{deg}]^{n}, [--] \}$. Then there exists a local coordinate system $(q,p,x,y)$ such that the metric takes the form
\begin{equation}
\label{metryka_dege_x_degn_mm}
\frac{1}{2} ds^{2} = \phi^{-2} \Big( dqdy - dpdx + \mathcal{A} \, dp^{2}  + \mathcal{B} \, dq^{2} 
- 2 ( \phi^{-1} ( \phi_{y} \mathcal{B}-\phi_{q} )x + S  ) \, dqdp   \Big)
\end{equation}
where $\phi = \phi (q,p,y)$, $\mathcal{A} = \mathcal{A} (q,p,x,y)$, $\mathcal{B} = \mathcal{B} (q,p,y)$ and $S=S(q,p,y)$ are arbitrary holomorphic (real smooth) functions.
\end{Twierdzenie}

\begin{proof}
If twist of $\mathcal{I}(\mathcal{C}_{m^{A}}^{e}, \mathcal{C}_{m^{\dot{A}}}^{n})$ vanishes ($\varrho_{1}=0$) then Lemma \ref{Lemat_kongruencja_pierwsza_nontwisting} holds true. Condition $M_{A}=0$ implies $\phi=\phi(q,p,y)$, $\mathcal{B} = \mathcal{B} (q,p,y)$ and $\mathcal{Q} = \phi^{-1} (\phi_{y} \mathcal{B} - \phi_{q})x + S$ where $S=S(q,p,y)$ is an arbitrary function. Inserting this into (\ref{metric_weak_HH_expanding_space_abrevaitions}) one arrives at (\ref{metryka_dege_x_degn_mm}). 
\end{proof}

\begin{Twierdzenie} 
\label{theorem_dege_x_degn_pp}
Let $(\mathcal{M}, ds^{2})$ be a complex (neutral) space of the type $\{ [\textrm{deg}]^{e} \otimes [\textrm{deg}]^{n}, [++] \}$. Then there exists a local coordinate system $(q,p,x,z)$ such that the metric takes the form
\begin{eqnarray}
\label{metryka_dege_x_degn_pp}
\frac{1}{2} ds^{2} &=& F^{-2} (x-\Sigma_{z})^{-2} \Big\{ -(dp+z  dq)dx - (x-\Sigma_{z}) dq dz + \mathcal{A} \, (dp+z dq)^{2} 
\\ \nonumber
&&  \ \ \ \ \ \ \ \ \ \ \ \ \ \ \ \ \ \ \  - \Big( 2T (x-\Sigma_{z})^{2} +\frac{2}{F} (x-\Sigma_{z}) ( \Omega F_{z} -z F_{p} + F_{q})
\\ \nonumber
&&  \ \ \ \ \ \ \ \ \ \ \ \ \ \ \ \ \ \ \ \ \ \ -\Omega \Sigma_{zz} + z \Sigma_{zp} - \Sigma_{zq} \Big) (dp+zdq) dq+ (x-\Sigma_{z}) \Omega \, dq^{2} \Big\}
\end{eqnarray}
where $\mathcal{A} = \mathcal{A} (q,p,x,z)$, $F=F(q,p,z)$, $T=T(q,p,z)$, $\Omega=\Omega(q,p,z)$ and $\Sigma=\Sigma(q,p,z)$ are arbitrary holomorphic (real smooth) functions.
\end{Twierdzenie}

\begin{proof}
If twist of $\mathcal{I}(\mathcal{C}_{m^{A}}^{e}, \mathcal{C}_{m^{\dot{A}}}^{n})$ is nonzero ($\varrho_{1} \ne 0$) then Lemma \ref{Lemat_kongruencja_pierwsza_twisting} holds true. From (\ref{pomocnicze_1}) one finds that condition  $M_{1}=0$ implies 
\begin{equation}
\label{rozwiazanie_na_fi_pierwsze}
\phi = F(q,p,z)(x-\Sigma_{z})
\end{equation}
where $F$ is an arbitrary nonzero function. Thus, $\theta_{1}$ is automatically nonzero. Eq. $M_{2}=0$ is quite easy to solve and it yields
\begin{equation}
\label{solution_pomoc_naOiA}
\mathcal{Q} + z \mathcal{A} = T (x-\Sigma_{z})^{2} + \frac{1}{F} (\Omega F_{z} - z F_{p} + F_{q}) (x-\Sigma_{z}) + \frac{1}{2} (\Sigma_{p} + z\Sigma_{zp} - \Sigma_{zq} - \Omega \Sigma_{zz})
\end{equation}
where $T=T(q,p,z)$ is an arbitrary function. Formula (\ref{solution_pomoc_naOiA}) is a solution for $\mathcal{Q}$ in terms of $\mathcal{A}$, $\Sigma$, $\Omega$, $F$ and $T$. Inserting (\ref{solution_pomoc_naOiA}) into (\ref{rozwiazanie_rownania_pierwszej_struny_2}) one gets a solution for $\mathcal{B}$. Collecting all the results and inserting them into (\ref{metric_weak_HH_expanding_space_abrevaitions}) one arrives at (\ref{metryka_dege_x_degn_pp}).
\end{proof}

Function $F=F(q,p,z)$ (compare (\ref{rozwiazanie_na_fi_pierwsze})) transforms under (\ref{gauge}) as follows
\begin{equation}
\label{transformacja_na_F}
F' = \frac{\Delta \lambda^{\frac{1}{2}} F}{q'_{q} - q'_{p}z}
\end{equation}

\subsection{Spaces of the types $[\textrm{deg}]^{e} \otimes [\textrm{deg}]^{ne}$} 
\label{subsekcja_dege_x_degne}

In this Section we consider weak $\mathcal{HH}$-spaces equipped with ASD $\mathcal{C}^{ne}$. To be more precise, we consider spaces of the types $\{ [\textrm{deg}]^{e} \otimes [\textrm{deg}]^{n}, [--] \}$ (the metric (\ref{metryka_dege_x_degn_mm})) and $\{ [\textrm{deg}]^{e} \otimes [\textrm{deg}]^{n}, [++] \}$ (the metric (\ref{metryka_dege_x_degn_pp})) and we equip these spaces with "the second" ASD congruence, namely $\mathcal{C}^{e}_{n^{\dot{A}}}$. In what follows the metrics (\ref{metryka_dege_x_degn_mm}) and (\ref{metryka_dege_x_degn_pp}) are background for further considerations.

The classification of spaces equipped with SD $\mathcal{C}^{e}$ and ASD $\mathcal{C}^{ne}$ presented in \cite{Chudecki_Kahler_1} lists 7 subtypes of such spaces. These subtypes characterize by the following properties of $\mathcal{I}s$: $[--,++]$, $[--,+-]$, $[--,-+]$, $[++,--]$, $[++,-+]$, $[++,+-]$ and $[++,++]$. In Section \ref{subsekcja_nieistniejacy_typ} we prove that the subtype $\{ [\textrm{deg}]^{e} \otimes [\textrm{deg}]^{ne}, [++,-+] \}$ cannot, in fact, exist\footnote{It does not mean that SD $\mathcal{C}^{e}$ and ASD $\mathcal{C}^{ne}$ cannot intersect in such a manner that properties of $\mathcal{I}s$ are $[++,-+]$. It is possible within the types $\{ [\textrm{I}]^{e} \otimes [\textrm{deg}]^{ne}, [++,-+] \}$ but we do not consider such types in this article.}.

\begin{Uwaga}
\normalfont
The eighth possibility (types $\{ [\textrm{deg}]^{e} \otimes [\textrm{deg}]^{ne}, [--,--] \}$) also cannot exist. Indeed, for such types the function $\phi$ is a function of $(q,p,y)$ only (compare Theorem \ref{theorem_dege_x_degn_mm}). Moreover, $\theta_{2}=0$ what implies $\phi_{y}=0$ (compare Lemma \ref{Lemat_kongruencja_druga_nontwisting}). Consequently, $\phi=\phi(q,p) \ \Longrightarrow \ J_{\dot{A}}=0$ and SD $\mathcal{C}_{m^{A}}$ is nonexpanding what is a contradiction.
\end{Uwaga}

\subsubsection{Spaces of the types $\{ [\textrm{deg}]^{e} \otimes [\textrm{deg}]^{ne},[--,++]  \}$, $\{ [\textrm{deg}]^{e} \otimes [\textrm{deg}]^{ne},[--,+-]  \}$ and $\{ [\textrm{deg}]^{e} \otimes [\textrm{deg}]^{ne},[--,-+]  \}$} 

In this Section the metric (\ref{metryka_dege_x_degn_mm}) is a background. 

\begin{Twierdzenie} 
\label{theorem_dege_x_degne_mm_pm}
Let $(\mathcal{M}, ds^{2})$ be a complex (neutral) space of the type $\{ [\textrm{deg}]^{e} \otimes [\textrm{deg}]^{ne}, [--,+-] \}$. Then there exists a local coordinate system $(q,p,x,y)$ such that the metric takes the form
\begin{equation}
\label{metryka_dege_x_degne_mm_pm}
\frac{1}{2} ds^{2} = \phi^{-2} \Big( dqdy - dpdx + \mathcal{A} \, dp^{2}  + \mathcal{B} \, dq^{2} 
 - 2 ( \phi^{-1} (\mathcal{B} \phi_{y}-\phi_{q} )x + S  ) \, dqdp   \Big)
\end{equation}
where $\phi = \phi (q,p,y)$, $\mathcal{A} = \mathcal{A} (q,p,x)$, $\mathcal{B} = \mathcal{B} (q,p,y)$ and $S=S(q,p,y)$ are arbitrary holomorphic (real smooth) functions such that $\phi_{y} \ne 0$.
\end{Twierdzenie}
\begin{proof}
 Consider the metric (\ref{metryka_dege_x_degn_mm}) equipped with $\mathcal{C}^{e}_{n^{\dot{A}}}$ such that $\mathcal{I} (\mathcal{C}^{e}_{m^{A}},\mathcal{C}^{e}_{n^{\dot{A}}}) = \mathcal{I}^{+-}$. From Lemma \ref{Lemat_kongruencja_druga_nontwisting} one finds that $\mathcal{A}_{y}=0$ and $\theta_{2} \sim \phi_{y} \ne 0$. The condition $\phi_{y} \ne 0$ automatically implies that $N_{1} \ne 0$ so $\mathcal{C}^{e}_{n^{\dot{A}}}$ in necessarily expanding.
\end{proof}

Spaces for which twist $\varrho_{2}$ of $\mathcal{I} (\mathcal{C}^{e}_{m^{A}}, \mathcal{C}^{e}_{n^{\dot{A}}})$ is nonzero are a little more complicated. For such spaces Theorem \ref{theorem_dege_x_degn_mm} and Lemma \ref{Lemat_kongruencja_druga_twisting} hold true simultaneously. Let us consider the case $\mathcal{I} (\mathcal{C}^{e}_{m^{A}}, \mathcal{C}^{e}_{n^{\dot{A}}}) = \mathcal{I}^{++}$ first.

\begin{Twierdzenie} 
\label{theorem_dege_x_degne_mm_pp}
Let $(\mathcal{M}, ds^{2})$ be a complex (neutral) space of the type $\{ [\textrm{deg}]^{e} \otimes [\textrm{deg}]^{ne}, [--,++] \}$. Then there exists a local coordinate system $(q,p,w,y)$ such that the metric takes the form
\begin{eqnarray}
\label{metryka_dege_x_degne_mm_pp}
\frac{1}{2} ds^{2} &=& \phi^{-2} \big\{ (dq+ wdp) dy + (y-H_{w}) dpdw - H_{q} \, dqdp + \mathcal{B} \, dq^{2}
\\ \nonumber
&& \ \ \ \ \ \ -2 \left( \phi^{-1} (\mathcal{B} \phi_{y} - \phi_{q}) (H-wy) +S  \right)  (dq+wdp)dp
\\ \nonumber
&&  \ \ \ \ \ \ + \left( -w^{2} \mathcal{B} + R (y-H_{w}) -w H_{q}  \right) \, dp^{2}   \big\}
\end{eqnarray}
where $\phi = \phi (q,p,y)$, $\mathcal{B} = \mathcal{B} (q,p,y)$, $S=S(q,p,y)$, $R=R(q,p,w)$ and $H=H(q,p,w)$ are arbitrary holomorphic (real smooth) functions such that 
\begin{subequations}
\begin{eqnarray}
\label{pppomoccnicze_1}
& \phi-2\phi_{y} (y-H_{w}) \ne 0 &
\\ \nonumber
&\textrm{and}&
\\ 
\label{pppomoccnicze_2}
& \bigg[ \phi-\phi_{y} (y-H_{w}) \ne 0 \ \ \ \textrm{or}&
\\ 
\label{pppomoccnicze_3}
& \phi_{p} - w \phi_{q} - \dfrac{\phi H_{q}}{H_{w}-y} - \phi \dfrac{\partial}{\partial y} \big( \phi^{-1} (\mathcal{B} \phi_{y} - \phi_{q}) (H-wy) +S + w \mathcal{B} \big) &
\\ \nonumber
& + \big( \phi^{-1} (\mathcal{B} \phi_{y} - \phi_{q}) (H-wy) +S + w \mathcal{B} \big) \left( \phi_{y} - \dfrac{\phi}{H_{w}-y} \right) \ne 0 \ \bigg]&
\end{eqnarray}
\end{subequations}
\end{Twierdzenie}
\begin{proof}
Consider the metric (\ref{metryka_dege_x_degn_mm}) equipped with $\mathcal{C}^{e}_{n^{\dot{A}}}$ such that $\mathcal{I} (\mathcal{C}^{e}_{m^{A}},\mathcal{C}^{e}_{n^{\dot{A}}}) = \mathcal{I}^{++}$. Coordinate $x$ is replaced by (\ref{rozwiazanie_na_x}). Formula (\ref{rozwiazanie_drugiego_drug_rownanianstrun}) gives a solution for $\mathcal{A}$ in terms of $\phi$, $\mathcal{B}$, $S$, $R$ and $H$, namely
\begin{equation}
\label{rrozwiaazanieAA}
\mathcal{A} = -2w (\phi^{-1} (\phi_{y} \mathcal{B} - \phi_{q}) (H-wy) +S) - w^{2} \mathcal{B} + R(y-H_{w}) + H_{p} - w H_{q}
\end{equation}
Putting (\ref{rrozwiaazanieAA}) and (\ref{rozwiazanie_na_x}) into (\ref{metryka_dege_x_degn_mm}) one arrives at (\ref{metryka_dege_x_degne_mm_pp}).

From Theorem \ref{theorem_dege_x_degn_mm} it follows that $\phi_{x}=0$. Thus, $0=\phi_{x} = \phi_{w} w_{x} \ \Longrightarrow \ \phi_{w}=0$. With $\phi_{w}=0$, Eqs. (\ref{wiezy_drugiej_kongruencji_1})-(\ref{wiezy_drugiej_kongruencji_3}) reduce to (\ref{pppomoccnicze_1})-(\ref{pppomoccnicze_3}).
\end{proof}

\begin{Twierdzenie} 
\label{theorem_dege_x_degne_mm_mp}
Let $(\mathcal{M}, ds^{2})$ be a complex (neutral) space of the type $\{ [\textrm{deg}]^{e} \otimes [\textrm{deg}]^{ne}, [--,-+] \}$. Then there exists a local coordinate system $(q,p,w,y)$ such that the metric takes the form
\begin{equation}
\label{metryka_dege_x_degne_mm_mp}
\frac{1}{2} ds^{2} = y^{-1} \big( dqdy+w\, dpdy + y \, dpdw + \mathcal{B} \, dq^{2} + (w \mathcal{B} -2 S) \, dpdq + (-2wS +yR) \, dp^{2}   \big)
\end{equation}
where $\mathcal{B} = \mathcal{B} (q,p,y)$, $S=S(q,p,y)$ and $R=R(q,p,w)$ are arbitrary holomorphic (real smooth) functions. 
\end{Twierdzenie}

\begin{proof}
Consider the metric (\ref{metryka_dege_x_degne_mm_pp}). Condition $\theta_{2}=0$ yields $\phi-2\phi_{y} (y-H_{w}) = 0$. Because $\phi=\phi(q,p,y)$, $\phi_{y} \ne 0$ and $H=H(q,p,w)$ one finds that $H_{ww}=0 \ \Longrightarrow \ H = H_{1}(q,p)  w + H_{2}(q,p)$. Fast analysis of the transformation formula (\ref{transformacja_na_H}) proves that both $H_{1}$ and $H_{2}$ can be gauged away. With $H=0$ one gets $\phi = \widetilde{\phi} y^{\frac{1}{2}}$ where $\widetilde{\phi}$ is an arbitrary nonzero function of $(q,p)$. From the transformation formula (\ref{transformacja_Q}) it follows that $\widetilde{\phi}$ can be brought to 1 without any loss of generality. Thus, $\phi = y^{\frac{1}{2}}$ and the metric (\ref{metryka_dege_x_degne_mm_pp}) reduces to (\ref{metryka_dege_x_degne_mm_mp}).
\end{proof}

\subsubsection{Spaces of the types $\{ [\textrm{deg}]^{e} \otimes [\textrm{deg}]^{ne},[++,--]  \}$ and $\{ [\textrm{deg}]^{e} \otimes [\textrm{deg}]^{ne},[++,+-]  \}$}

In this Section spaces of the types $\{ [\textrm{deg}]^{e} \otimes [\textrm{deg}]^{n},[++]  \}$ are equipped with $\mathcal{C}^{e}_{n^{\dot{A}}}$ such that $\mathcal{I}(\mathcal{C}^{e}_{m^{A}}, \mathcal{C}^{e}_{n^{\dot{A}}}) = \mathcal{I}^{+-}$ or $\mathcal{I}(\mathcal{C}^{e}_{m^{A}}, \mathcal{C}^{e}_{n^{\dot{A}}}) = \mathcal{I}^{--}$. For such spaces Theorem \ref{theorem_dege_x_degn_pp} and Lemma \ref{Lemat_kongruencja_druga_nontwisting} hold true simultaneously. Hence, the metric (\ref{metryka_dege_x_degn_pp}) is a background.
\begin{Twierdzenie} 
\label{theorem_dege_x_degne_pp_pm}
Let $(\mathcal{M}, ds^{2})$ be a complex (neutral) space of the type $\{ [\textrm{deg}]^{e} \otimes [\textrm{deg}]^{ne}, [++,+-] \}$. Then there exists a local coordinate system $(q,p,x,z)$ such that the metric takes the form
\begin{eqnarray}
\label{metryka_dege_x_degne_pp_pm}
\frac{1}{2} ds^{2} &=& F^{-2} (x-\Sigma_{z})^{-2} \Big\{ -(dp+z  dq)dx - (x-\Sigma_{z}) dq dz + \mathcal{A} \, (dp+z dq)^{2} 
\\ \nonumber
&&  \ \ \ \ \ \ \ \ \ \ \ \ \ \ \ \ \ \ \  - \Big( 2T (x-\Sigma_{z})^{2} +\frac{2}{F} (x-\Sigma_{z}) ( \Omega F_{z} -z F_{p} + F_{q})
\\ \nonumber
&&  \ \ \ \ \ \ \ \ \ \ \ \ \ \ \ \ \ \ \ \ \ \ -\Omega \Sigma_{zz} + z \Sigma_{zp} - \Sigma_{zq} \Big) (dp+zdq) dq+ (x-\Sigma_{z}) \Omega \, dq^{2} \Big\}
\end{eqnarray}
where $\mathcal{A} = \mathcal{A} (q,p,x)$, $F=F(q,p,z)$, $T=T(q,p,z)$, $\Omega=\Omega(q,p,z)$ and $\Sigma=\Sigma(q,p,z)$ are arbitrary holomorphic (real smooth) functions such that $F_{z} \ne 0$ or $\Sigma_{zz} \ne 0$.
\end{Twierdzenie}
\begin{proof}
Consider the metric (\ref{metryka_dege_x_degn_pp}) with constraint Eqs. (\ref{uproszczone_formuly_dla_w_rowne_0}). Formulas (\ref{uproszczone_formuly_dla_w_rowne_0}) need to be transformed from the coordinate system $(q,p,x,y)$ to the coordinate system $(q,p,x,z)$. Hence, $0=\mathcal{A}_{y} = \mathcal{A}_{z} z_{y} \ \Longrightarrow \ \mathcal{A}_{z}=0$. Moreover, $\theta_{2} \ne 0$ implies $0 \ne \phi_{y} = [F_{z} (x-\Sigma_{z}) - F \Sigma_{zz}] z_{y}$. Thus, $F_{z} \ne 0$ or $\Sigma_{zz} \ne 0$. Also, $\theta_{2} \ne 0$ automatically implies that $N_{1} \ne 0$ so $\mathcal{C}_{n^{\dot{A}}}$ is necessarily expanding.
\end{proof}
\begin{Twierdzenie} 
\label{theorem_dege_x_degne_pp_mm}
Let $(\mathcal{M}, ds^{2})$ be a complex (neutral) space of the type $\{ [\textrm{deg}]^{e} \otimes [\textrm{deg}]^{ne}, [++,--] \}$. Then there exists a local coordinate system $(q,p,x,z)$ such that the metric takes the form
\begin{eqnarray}
\label{metryka_dege_x_degne_pp_mm}
\frac{1}{2} ds^{2} &=&  x^{-2} \Big\{ -(dp+z  dq)dx - x \, dq dz + \mathcal{A} \, (dp+z dq)^{2} 
\\ \nonumber
&&  \ \ \ \ \ \   -  2T x^{2}  (dp+zdq) dq  + x \Omega \, dq^{2} \Big\}
\end{eqnarray}
where $\mathcal{A} = \mathcal{A} (q,p,x)$, $T=T(q,p,z)$ and $\Omega=\Omega(q,p,z)$  are arbitrary holomorphic (real smooth) functions such that $T_{z} \ne 0$.
\end{Twierdzenie}
\begin{proof}
Consider the metric (\ref{metryka_dege_x_degne_pp_pm}) such that $\theta_{2}=0$. Thus, $F_{z}=0$ and $\Sigma_{zz}=0$. From the transformation formulas (\ref{transformacja_na_F}) and (\ref{transformacja_na_Sigma_ogolna}) it follows that $F$ can be brought to $1$ and $\Sigma$ can be gauged away without any loss of generality. Hence, one arrives at (\ref{metryka_dege_x_degne_pp_mm}).

From (\ref{solution_pomoc_naOiA}) one finds $\mathcal{Q} + z \mathcal{A} = T x^{2}$. Also, $y=-xz$ (compare (\ref{rozwiazanie_na_y}) with $\Sigma=0$). Because $\theta_{2} =0$ then $N_{1}=0$ (compare (\ref{uproszczone_formuly_dla_w_rowne_0})). Thus, $N_{2}$ must be nonzero otherwise $\mathcal{C}_{n^{\dot{A}}}$ is nonexpanding. From (\ref{uproszczone_formuly_dla_w_rowne_0}) we have
\begin{equation}
N_{2} \sim \mathcal{A} - x \mathcal{Q}_{y} = \mathcal{A} - x \frac{\partial}{\partial y} (Tx^{2} + \frac{y}{x} \mathcal{A}) = x^{2} T_{z}
\end{equation}
Hence, $T_{z} \ne 0$.
\end{proof}

\subsubsection{Spaces of the types $\{ [\textrm{deg}]^{e} \otimes [\textrm{deg}]^{ne},[++,++]  \}$}

For the most complicated types $\{ [\textrm{deg}]^{e} \otimes [\textrm{deg}]^{ne},[++,++]  \}$ twists of $\mathcal{I}(\mathcal{C}^{e}_{m^{A}}, \mathcal{C}^{e}_{n^{\dot{A}}})$ and $\mathcal{I}(\mathcal{C}^{e}_{m^{A}}, \mathcal{C}^{e}_{m^{\dot{A}}})$ are nonzero, $\varrho_{1} \ne 0$, $\varrho_{2} \ne 0$. Thus, both Lemmas \ref{Lemat_kongruencja_pierwsza_twisting} and \ref{Lemat_kongruencja_druga_twisting} hold true. In particular, Eqs. (\ref{rozwiazanie_na_y}) and (\ref{rozwiazanie_na_x}) hold true simultaneously. Hence
\begin{equation}
\label{jawne_wzory_na_xiy}
x= \frac{H-w \Sigma}{1-zw}, \ y= \frac{\Sigma -z H} {1-zw}, \ H=H(q,p,w), \ \Sigma=\Sigma (q,p,z)
\end{equation}
Additional constraints are Eqs. (\ref{rozwiazanie_rownania_pierwszej_struny_2}), (\ref{rozwiazanie_drugiego_drug_rownanianstrun}), $M_{A}=0$ (compare (\ref{pomocnicze_2})-(\ref{pomocnicze_3})), $N_{A} \ne 0$ (compare (\ref{wiezy_drugiej_kongruencji})-(\ref{wiezy_drugiej_kongruencji_3})), $\theta_{1} \ne 0$ and $\theta_{2} \ne 0$ (compare (\ref{pomocnicze_1}) and (\ref{wiezy_drugiej_kongruencji_1})).

\begin{Twierdzenie} 
\label{theorem_dege_x_degne_pp_pp_and_pp_mp}
Let $(\mathcal{M}, ds^{2})$ be a complex (neutral) space of the type $\{ [\textrm{deg}]^{e} \otimes [\textrm{deg}]^{ne}, [++,++] \}$. Then there exists a local coordinate system $(q,p,w,z)$ such that the metric takes the form
\begin{eqnarray}
\label{metryka_dege_x_degne_pp_pp_and_pp_mp}
\frac{1}{2} ds^{2} &=&  \phi^{-2} \Big\{ dqdy-dpdx -2 \mathfrak{a} (dp+zdq) dq + \mathfrak{b} \, dq^{2}
\\ \nonumber
&&  \ \ \ \ \ \   + \frac{\mathfrak{c} - w^{2} \mathfrak{b} - 2w \mathfrak{a} (1-wz)}{(1-wz)^{2}} (dp+zdq)^{2} \Big\}
\end{eqnarray}
where $x$, $y$, $\phi$, $\mathfrak{a}$, $\mathfrak{b}$ and $\mathfrak{c}$ stand for abbreviations
\begin{eqnarray}
\nonumber
&&  x:= \frac{H-w \Sigma}{1-zw}, \ y:= \frac{\Sigma -z H} {1-zw} \ , \phi := F ( x - \Sigma_{z} ) 
\\ \nonumber
&& \mathfrak{a}  := T (x-\Sigma_{z})^{2} + \frac{1}{F} (x-\Sigma_{z}) (\Omega F_{z} - z F_{p} + F_{q}) + \frac{1}{2} (\Sigma_{p} - \Omega \Sigma_{zz} + z \Sigma_{zp} - \Sigma_{zq})
\\ \nonumber
&& \mathfrak{b} := \Omega (x-\Sigma_{z}) + z \Sigma_{p} - \Sigma_{q}
\\ \nonumber
&& \mathfrak{c} := R (y-H_{w}) + H_{p} - w H_{q}
\end{eqnarray}
and $\Omega=\Omega(q,p,z)$, $\Sigma=\Sigma (q,p,z)$, $R=R(q,p,w)$, $T=T(q,p,z)$, $H=H(q,p,w)$ and $F=F(q,p,z)$ are  holomorphic (real smooth) functions such that
\begin{subequations}
\begin{eqnarray}
\label{typ_prawie_ogolny_wwar_1}
&   \dfrac{\phi}{H_{w}-y} + \dfrac{2 \phi_{z} (1-wz)}{\Sigma_{z}-x} \ne 0  &
\\ \nonumber
&\textrm{ and }& 
\\ 
\label{typ_prawie_ogolny_wwar_2}
& \Bigg[ \dfrac{\phi}{H_{w}-y} + \dfrac{\phi_{z} (1-wz)}{\Sigma_{z}-x} \ne 0 \ \ \ \ \textrm{ or } &
\\ 
\label{typ_prawie_ogolny_wwar_3}
 &- \dfrac{\phi H_{q}}{H_{w}-y} - w \phi_{q} + \phi_{p} - \phi \dfrac{1-wz}{\Sigma_{z}-x} \dfrac{\partial}{\partial z} \left( \mathfrak{a} + \dfrac{w \mathfrak{b}- z \mathfrak{c}}{1-wz} \right) - \phi_{w} R&
\\ \nonumber
& - \dfrac{\phi}{H_{w}-y} \left( \mathfrak{a} + \dfrac{w \mathfrak{b}- z \mathfrak{c}}{1-wz} \right) + \dfrac{\phi_{z}}{\Sigma_{z}-x} \left( w \Omega (x-\Sigma_{z}) + (1-wz) (\mathfrak{a}-\Sigma_{p}) \right) \ne 0 \Bigg] &
\\ \nonumber
\end{eqnarray} 
\end{subequations}
\end{Twierdzenie}
\begin{proof}
Equation $M_{1}=0$ implies $\phi = F(q,p,z) (x-\Sigma_{z})$, $F \ne 0$ and equation $M_{2} =0$ implies (\ref{solution_pomoc_naOiA}). From Eqs. (\ref{rozwiazanie_rownania_pierwszej_struny_2}), (\ref{rozwiazanie_drugiego_drug_rownanianstrun}) and (\ref{solution_pomoc_naOiA}) one finds
\begin{eqnarray}
\label{rozwiazania_na_ABQ_degexdegne}
\mathcal{A} = \frac{\mathfrak{c} - 2w \mathfrak{a} (1 -zw) - w^{2} \mathfrak{b}}{(1-zw)^{2}}, \ \mathcal{Q} = \mathfrak{a}-z \mathcal{A}, \ \mathcal{B} = \mathfrak{b} - 2z \mathfrak{a} + z^{2} \mathcal{A}
\end{eqnarray}
Inserting $\phi$ and (\ref{rozwiazania_na_ABQ_degexdegne}) into (\ref{metric_weak_HH_expanding_space_abrevaitions}) with $x$ and $y$ replaced according to (\ref{jawne_wzory_na_xiy}) one arrives at (\ref{metryka_dege_x_degne_pp_pp_and_pp_mp}). 

The expansion $\theta_{1}$ is nonzero (it follows from Eq. $M_{1}=0$). Moreover, $\theta_{2} \ne 0$ and at least one of the $N_{1}$ or $N_{2}$ must be nonzero. Formulas (\ref{wiezy_drugiej_kongruencji_1})-(\ref{wiezy_drugiej_kongruencji_3}) transformed into coordinate system $(q,p,w,z)$ lead to the Eqs. (\ref{typ_prawie_ogolny_wwar_1})-(\ref{typ_prawie_ogolny_wwar_3}).
\end{proof}

\subsubsection{Spaces of the types $\{ [\textrm{deg}]^{e} \otimes [\textrm{deg}]^{ne},[++,-+]  \}$} 
\label{subsekcja_nieistniejacy_typ}

The last types which remained to be considered are types $\{ [\textrm{deg}]^{e} \otimes [\textrm{deg}]^{ne}, [++,-+] \}$. The metric of spaces of such types can be obtained from the metric (\ref{metryka_dege_x_degne_pp_pp_and_pp_mp}) by demanding that $\theta_{2} = 0$. It gives the condition 
\begin{equation}
\label{rownanie_na_theta2_zero}
\dfrac{\phi}{H_{w}-y} + \dfrac{2 \phi_{z} (1-wz)}{\Sigma_{z}-x} = 0
\end{equation}
where $\phi = F(q,p,z) (x-\Sigma_{z})$; $x$ and $y$ are given by (\ref{jawne_wzory_na_xiy}). Explicitly, (\ref{rownanie_na_theta2_zero}) yields
\begin{equation}
\label{rrrrrownanie_1}
F(x-\Sigma_{z})^{2} + 2 (1-wz) (y-H_{w}) \left( F_{z} (x-\Sigma_{z}) + F \left(\frac{w (x-\Sigma_{z})}{1-zw} - \Sigma_{zz} \right) \right) = 0
\end{equation}
Elimination of $x$ and $y$ in (\ref{rrrrrownanie_1}) via (\ref{jawne_wzory_na_xiy}) gives an equation for $F=F(q,p,z)$, $\Sigma=\Sigma(q,p,z)$ and $H=H(q,p,w)$. 

Before we deal with Eq. (\ref{rrrrrownanie_1}) we point out an additional relation, very helpful in further analysis. For the types $\{ [\textrm{deg}]^{e} \otimes [\textrm{deg}]^{ne}, [++,-+] \}$ conditions $M_{1}=0$ (\ref{pomocnicze_2}) and $\theta_{2}=0$ (\ref{wiezy_drugiej_kongruencji_1}) hold true simultaneously. A solution of Eq. $M_{1}=0$ reads $\phi = F(q,p,z) (x-\Sigma_{z})$ (in the coordinate system $(q,p,x,z)$) and a solution of Eq. $\theta_{2}=0$ yields $\phi = G (q,p,w) ( H_{w} - y )^{\frac{1}{2}}$, $G \ne 0$ (in the coordinate system $(q,p,w,y)$). Both these solutions give an alternative forms of $\phi$. Thus, we arrive at the consistency condition
\begin{equation}
\label{warunek_konsystencjji}
F (q,p,z) \left( \frac{H-w \Sigma}{1-zw} - \Sigma_{z} \right) = G (q,p,w) \left( H_{w} - \frac{\Sigma -z H} {1-zw} \right)^{\frac{1}{2}}
\end{equation}
Of course, (\ref{rrrrrownanie_1}) and (\ref{warunek_konsystencjji}) are equivalent but it is algebraically problematic to find $F$, $\Sigma$, $H$ and $G$ using only (\ref{rrrrrownanie_1}) or only (\ref{warunek_konsystencjji}). In what follows we cleverly juggle with both (\ref{rrrrrownanie_1}) and (\ref{warunek_konsystencjji}).

\begin{Lemat}
\label{Lemat_o_nieistnieniu_rozwiazania}
Equation (\ref{rownanie_na_theta2_zero}) holds true only if $\phi = 0$.
\end{Lemat}
\begin{proof}
We use condition (\ref{warunek_konsystencjji}) to eliminate factor $(H_{w}-y)$ in (\ref{rrrrrownanie_1}) and we put an explicit form of $x$ into (\ref{rrrrrownanie_1}). After some algebraic work we find that Eq. (\ref{rrrrrownanie_1}) reduces to the following condition
\begin{equation}
\label{rrrrrownanie_2}
\widetilde{a} w^{2} + \widetilde{b}w + \widetilde{c} = (1-zw) G^{2} - H \left( \widetilde{f}_{z} + w (2 \widetilde{f} - z \widetilde{f}_{z}) \right)
\end{equation}
where
\begin{eqnarray}
\label{definicje_a_b_c_f}
&& \widetilde{c}  := -2 \widetilde{f} \Sigma_{zz} - \widetilde{f}_{z} \Sigma_{z}, \ \widetilde{b} := -2z \widetilde{c} - 2\widetilde{f} \Sigma_{z} - \widetilde{f}_{z} \Sigma
\\ \nonumber
&& \widetilde{a} := -z^{2} \widetilde{c} - z \widetilde{b} - 2\widetilde{f} \Sigma, \ \widetilde{f} := F^{2}
\end{eqnarray}
Note, that $\widetilde{a}$, $\widetilde{b}$, $\widetilde{c}$ and $\widetilde{f}$ are functions of $(q,p,z)$ only. Thus, the left hand side of (\ref{rrrrrownanie_2}) is a second order polynomial in $w$. Acting on (\ref{rrrrrownanie_2}) by $\partial_{w}^{3} \partial_{z}^{2}$ one finds that $(i)$ $H_{ww}=0$ or $(ii)$  $\widetilde{f}_{zzz}=0$.

$(i)$, case $H_{ww}=0$. From (\ref{transformacja_na_H}) it follows that if $H_{ww}=0$ then $H$ can be gauged away without any loss of generality. With $H=0$, formula (\ref{rrrrrownanie_2}) implies that $\widetilde{a}$, $ \widetilde{b}$ and $\widetilde{c}$ are all linear in $z$. After simple calculations one arrives at the solutions $\widetilde{c} = \widetilde{c} (q,p)$, $\widetilde{a} = \widetilde{\alpha} z$ and $\widetilde{b} = -\widetilde{c} z - \widetilde{\alpha}$ where $\widetilde{\alpha} = \widetilde{\alpha} (q,p)$. The last step is to put these solutions into definition of $\widetilde{a}$. Finally, one finds that $\widetilde{f} \Sigma=0$. Because $H$ and $\Sigma$ cannot simultaneously vanish, $\widetilde{f}=0$. It yields $F=0$ and, eventually, $\phi=0$.

$(ii)$, case $\widetilde{f}_{zzz}=0$. If $\widetilde{f}_{zzz}=0$ then $\widetilde{f} = \widetilde{m}(q,p) z^{2} + \widetilde{n}(q,p) z + \widetilde{s}(q,p)$ and  $\widetilde{a}_{zz} = \widetilde{b}_{zz} =\widetilde{c}_{zz} =0$. Hence, $\widetilde{c} = -\widetilde{\alpha} z - \widetilde{\beta}$, $\widetilde{b} = \widetilde{\gamma} z + \widetilde{\delta}$, $\widetilde{a} = \widetilde{\mu} z + \widetilde{\nu}$ where $\widetilde{\alpha}$, $\widetilde{\beta}$, $\widetilde{\gamma}$, $\widetilde{\delta}$, $\widetilde{\mu}$ and $\widetilde{\nu}$ are functions of $(q,p)$ only. Inserting these formulas into (\ref{definicje_a_b_c_f}) one arrives at the conditions
\begin{subequations}
\begin{eqnarray}
\label{ROWnanie_pomocnicze_1}
2\widetilde{f} \Sigma_{zz} + \widetilde{f}_{z} \Sigma_{z} &=& \widetilde{\alpha} z + \widetilde{\beta}
\\
\label{ROWnanie_pomocnicze_2}
2 \widetilde{f} \Sigma_{z} + \widetilde{f}_{z} \Sigma &=& 2 \widetilde{\alpha} z^{2} + (2\widetilde{\beta} - \widetilde{\gamma})z - \widetilde{\delta}
\\ 
\label{ROWnanie_pomocnicze_3}
2\widetilde{f} \Sigma &=& \widetilde{\alpha} z^{3} + (\widetilde{\beta}-\widetilde{\gamma}) z^{2} - (\widetilde{\delta}+\widetilde{\mu}) z - \widetilde{\nu}
\end{eqnarray}
\end{subequations}
(\ref{ROWnanie_pomocnicze_2}) - $\partial_{z}$(\ref{ROWnanie_pomocnicze_3}) gives $\widetilde{f}_{z} \Sigma = \widetilde{\alpha} z^{2} - \widetilde{\gamma} z - \widetilde{\mu}$. Using this result to eliminate factor $\widetilde{f}_{z} \Sigma$ in (\ref{ROWnanie_pomocnicze_2}) one finds the condition $2\widetilde{f} \Sigma_{z} = \widetilde{\alpha} z^{2} + 2 \widetilde{\beta} z +\widetilde{\mu} - \widetilde{\delta}$. Inserting all these formulas into (\ref{ROWnanie_pomocnicze_1}) we get $\widetilde{f} \Sigma_{zz}=0$. Thus, $\widetilde{f}=0 \ \Longrightarrow \ \phi=0$, or $\Sigma_{zz}=0$. 

If $\Sigma_{zz}=0$ then $\Sigma$ can be gauged away without any loss of generality (compare transformation formula (\ref{transformacja_na_Sigma_ogolna})). $\Sigma=0$ implies $\widetilde{a}=\widetilde{b}=\widetilde{c}=0$ and (\ref{rrrrrownanie_2}) simplifies to the formula $H (\widetilde{m} + \widetilde{n} w + \widetilde{s} w^{2})$. Because $\Sigma$ and $H$ cannot be simultaneously equal $0$, then $\widetilde{m} =\widetilde{n} = \widetilde{s}=0$. Finally, $\widetilde{f}=0 \ \Longrightarrow \ \phi=0$. Hence, the proof is completed.
\end{proof}

\begin{Wniosek}
\label{wniosek_o_nieistnieniu}
Spaces of the types $\{ [\textrm{deg}]^{e} \otimes [\textrm{deg}]^{ne}, [++,-+] \}$ do not exist.
\end{Wniosek}
\begin{proof}
Eq. (\ref{rownanie_na_theta2_zero}) is a necessary condition for spaces of the types $\{ [\textrm{deg}]^{e} \otimes [\textrm{deg}]^{ne}, [++,-+] \}$ to exist. According to Lemma \ref{Lemat_o_nieistnieniu_rozwiazania}, Eq. (\ref{rownanie_na_theta2_zero}) implies $\phi=0$ what is a contradiction.
\end{proof}


\section{Spaces of the types $\{ [\textrm{deg}]^{e} \otimes [\textrm{D}]^{nn},[++,--] \}$} 
\label{subsekcja_dege_x_degnn}
\setcounter{equation}{0}

\subsection{Non-Einstein case} 
\label{subsekcja_general_result}

\subsubsection{General results}
\label{subsubsekcja_general_results}

In this Section we finally arrive at the first class of para-Kähler spaces. To be more precise we find a general metric of spaces of the types $\{ [\textrm{deg}]^{e} \otimes [\textrm{D}]^{nn},[++,--] \}$. Such a metric can be extracted from Theorem \ref{theorem_dege_x_degne_pp_mm} by demanding that the expansion of $\mathcal{C}_{n^{\dot{A}}}$ vanishes. It is equivalent to the fact that $T_{z}=0 \ \Longleftrightarrow \ T=T(q,p)$. Hence, the metric takes the form (\ref{metric_weak_HH_expanding_space_abrevaitions}) with
\begin{eqnarray}
\label{funkcje_metryczne_paraKahler_class_1}
&& \phi = x, \ \mathcal{A}=\mathcal{A}(q,p,x), \ \mathcal{Q} = \frac{y}{x} \, \mathcal{A} + x^{2} \, T, \ T=T(q,p)
\\ \nonumber
&& \mathcal{B} = \frac{y^{2}}{x^{2}} \, \mathcal{A} + 2y x \, T   + x \, \Omega , \ \Omega = \Omega (q,p,z), \ z:=-\frac{y}{x}
\end{eqnarray}
The metric (\ref{metric_weak_HH_expanding_space_abrevaitions}) with (\ref{funkcje_metryczne_paraKahler_class_1}) remains invariant under transformations (\ref{gauge}) such that 
\begin{eqnarray}
\label{gauge_paraKahler_class_1}
&& \sigma^{\dot{A}}=0, \ q'=q'(q) , \ p'=p'(q,p), \ \lambda^{-1} = {p'_{p}}^{2}
\\ \nonumber
&& x'=p'_{p} \, x, \ y' = \frac{p'_{p}}{q'_{q}} (p'_{q} \, x + p'_{p} \, y), \ z'=\frac{1}{q'_{q}} (p'_{p} \, z - p'_{q} )
\end{eqnarray}
Metrical functions $\mathcal{A}$, $T$ and $\Omega$ transform under (\ref{gauge_paraKahler_class_1}) as follows
\begin{eqnarray}
\label{gauge_paraKahler_class_1_funkcje}
&& \mathcal{A}'=\mathcal{A} + \frac{p'_{pp}}{p'_{p}} \, x, \ T' = \frac{1}{p'_{p}q'_{q}} \, T,
\\ \nonumber  
&& p'_{p} \Omega' = \frac{{p'_{p}}^{2}}{{q'_{q}}^{2}} \Omega - \frac{p'_{p}p'_{pp}}{{q'_{q}}^{2}} \, z^{2} + \frac{1}{q'_{q}} \frac{\partial}{\partial q} \left( \frac{{p'_{p}}^{2}}{q'_{q}} \right) z - \frac{p'_{p}}{q'_{q}} \frac{\partial}{\partial q} \left( \frac{p'_{q}}{q'_{q}} \right) 
\end{eqnarray}

For further purposes it is necessary to calculate conformal curvature coefficients, curvature scalar and traceless Ricci tensor. Curvature scalar reads
\begin{equation}
\label{deg_x_Dnn_mm_pp_R}
\mathcal{R} = -2x^{2} \left( \mathcal{A}_{xx} + \frac{1}{x} \Omega_{zz} - \frac{4 \mathcal{A}_{x}}{x} + \frac{6 \mathcal{A}}{x^{2}} \right)
\end{equation}
Nonzero ASD curvature coefficients are\footnote{Obviously, the condition $2 C_{\dot{1}\dot{1}\dot{1}\dot{2}}^{2} - 3 C_{\dot{1}\dot{1}\dot{1}\dot{1}} C_{\dot{1}\dot{1}\dot{2}\dot{2}} = 0$ for the ASD Weyl spinor to be of the type [D] is identically satisfied.}
\begin{equation}
\label{deg_x_Dnn_mm_pp_ASDWeyl}
C_{\dot{1}\dot{1}\dot{2}\dot{2}} = \frac{\mathcal{R}}{12}, \ C_{\dot{1}\dot{1}\dot{1}\dot{2}} = z \frac{\mathcal{R}}{4}, \ C_{\dot{1}\dot{1}\dot{1}\dot{1}} = z^{2} \frac{\mathcal{R}}{2},
\end{equation}
SD conformal curvature coefficients takes the form
\begin{eqnarray}
 C^{(3)} &=& \frac{\mathcal{R}}{6} +2\mathcal{A}-2x\mathcal{A}_{x} = -\frac{x}{3} ( \Omega_{zz} + x \mathcal{A}_{xx} + 2 \mathcal{A}_{x}  )
\\ \nonumber
 -x^{-4} C^{(2)} &=& 4T_{p} x-\Omega_{zp} -z \mathcal{A}_{xp} -\frac{z}{x} \mathcal{A}_{p} + \frac{1}{x} \mathcal{A}_{q} + \mathcal{A}_{xq} + T (x \Omega_{zz} - x^{2}\mathcal{A}_{xx} + 6\mathcal{A})
\\ \nonumber
\frac{1}{2} x^{-6} C^{(1)} &=& - x\Omega_{pp}+ 2 T x^{2} \Omega_{zp}+ T^{2} x^{2} (-x \Omega_{zz} - x^{2} \mathcal{A}_{xx} + 6x\mathcal{A}_{x} - 12 \mathcal{A})
\\ \nonumber
&& + \Omega_{z} (\mathcal{A}_{q} - z\mathcal{A}_{p} -2T x^{2} \mathcal{A}_{x} + 4Tx \mathcal{A} + 2x^{2} T{_p})+ T_{q} (2x^{2} \mathcal{A}_{x} - 4x \mathcal{A})
\\ \nonumber
&&+ T_{p} (-4Tx^{3} - 2zx^{2} \mathcal{A}_{x} + 4zx \mathcal{A} ) -2T_{pq} x^{2} + 2zx^{2} T_{pp} -2Tzx^{2} \mathcal{A}_{xp} 
\\ \nonumber
&& + \mathcal{A}_{p} (\Omega + 8T xz)+ x\mathcal{A}_{x}\Omega_{p}- \mathcal{A}_{qq}+ 2z \mathcal{A}_{pq}- 8T x \mathcal{A}_{q}+ 2T x^{2} \mathcal{A}_{xq} -z^{2}\mathcal{A}_{pp}
\end{eqnarray}

Nonzero coefficients of the traceless Ricci tensor are
\begin{equation}
\label{deg_x_Dnn_mm_pp_tracelessRicci}
-\frac{4}{x} C_{12\dot{1}\dot{2}} = - \frac{2}{zx} C_{12\dot{1}\dot{1}} = \alpha, \ 
\frac{1}{zx^{3}} C_{22\dot{1}\dot{1}} = \frac{2}{x^{3}} C_{22\dot{1}\dot{2}} = \beta
\end{equation}
where
\begin{eqnarray}
\alpha &:=& \Omega_{zz} -x \mathcal{A}_{xx} + 2\mathcal{A}_{x}
\\ \nonumber
\beta &:=& \partial_{q} (x\mathcal{A}_{x}-3\mathcal{A}) - z \partial_{p} (x\mathcal{A}_{x}-3\mathcal{A}) +x \Omega_{pz} + \frac{1}{2} T x \mathcal{R}
\end{eqnarray}
Equivalently
\begin{equation}
C_{AB \dot{C}\dot{D}} = ( x^{3} \beta \, \delta_{A}^{2} \delta_{B}^{2} -x \alpha \, \delta_{(A}^{1} \delta_{B)}^{2} ) \, m_{(\dot{C}} n_{\dot{D})}, \ m_{\dot{A}}=[z,1], n_{\dot{A}}=[1,0]
\end{equation}
If $\alpha \ne 0$ then the traceless Ricci tensor has four different eigenvectors and two double eigenvalues (complex type is $[2N_{1}-2N]_{2}$, see \cite{Przanowski_classification}). If $\alpha=0$ and $\beta \ne 0$ then the traceless Ricci tensor has two different eigenvectors and one quadruple eigenvalue (complex type $^{(2)}[4N]^{a}_{2}$). Algebraic types of the traceless Ricci tensor in neutral spaces are slightly more complicated and we do not discuss this issue here (for details see \cite{Chudecki_struny,Chudecki_Ricci}).

\begin{Twierdzenie} 
\label{theorem_dege_x_Dnn_pp_mm}
Let $(\mathcal{M}, ds^{2})$ be a complex (neutral) space of the type $\{ [\textrm{deg}]^{e} \otimes [\textrm{D}]^{nn}, [++,--] \}$ ($\{ [\textrm{deg}]^{e} \otimes [\textrm{D}_{r}]^{nn}, [++,--] \}$). Then there exist a local coordinate system $(q,p,x,z)$ such that the metric takes the form
\begin{eqnarray}
\label{metryka_dege_x_Dnn_pp_mm}
\frac{1}{2} ds^{2} &=&  x^{-2} \Big\{ -(dp+z  dq)dx - x \, dq dz + \mathcal{A} \, (dp+z dq)^{2} 
\\ \nonumber
&&  \ \ \ \ \ \   -  2T_{0} x^{2}  (dp+zdq) dq  + x \Omega \, dq^{2} \Big\}
\end{eqnarray}
where $T_{0}=\{ -1, 0,1 \}$ is a constant, $\mathcal{A} = \mathcal{A} (q,p,x)$ and $\Omega=\Omega(q,p,z)$ are arbitrary holomorphic (real smooth) functions such that
\begin{equation}
\label{warunek_na_niezerowosc_skkalara}
 \Omega_{zz} + x \mathcal{A}_{xx}  - 4 \mathcal{A}_{x} + \frac{6 \mathcal{A}}{x} \ne 0
\end{equation}
\end{Twierdzenie}
\begin{proof}
Inserting (\ref{funkcje_metryczne_paraKahler_class_1}) into (\ref{metric_weak_HH_expanding_space_abrevaitions}) one arrives at (\ref{metryka_dege_x_Dnn_pp_mm}). From (\ref{gauge_paraKahler_class_1_funkcje}) it follows that the function $T$ can be brought to a constant value, namely $T= \{ -1, 0,  1 \}$. The curvature scalar $\mathcal{R}$ given by (\ref{deg_x_Dnn_mm_pp_R}) must be nonzero, otherwise the ASD conformal curvature vanishes. Hence, (\ref{warunek_na_niezerowosc_skkalara}) is valid.
\end{proof}

As long as $C^{(3)} \ne 0$ and $2C^{(2)} C^{(2)} - 3 C^{(3)} C^{(1)} \ne 0$ the metric (\ref{metryka_dege_x_Dnn_pp_mm}) is of the type $\{ [\textrm{II}]^{e} \otimes [\textrm{D}]^{nn},[++,--] \}$.

\subsubsection{Spaces of the type $\{ [\textrm{D}]^{e} \otimes [\textrm{D}]^{nn},[++,--] \}$}

If $C^{(3)} \ne 0$ and $2C^{(2)} C^{(2)} - 3 C^{(3)} C^{(1)} =0$ the metric (\ref{metryka_dege_x_Dnn_pp_mm}) is of the type $\{ [\textrm{D}]^{e} \otimes [\textrm{D}]^{nn},[++,--] \}$. However, the condition $2C^{(2)} C^{(2)} - 3 C^{(3)} C^{(1)} =0$ written explicitly is even more complicated then its counterpart for the space of the type $ [\textrm{D}]^{n} \otimes [\textrm{D}]^{nn}$ (Eq. (3.20) of Ref. \cite{Chudecki_Kahler_1}). It will be analyzed elsewhere.

\begin{Uwaga}
\normalfont An example of a (non-Einstein) space of the type $\{ [\textrm{D}]^{e} \otimes [\textrm{D}]^{nn},[++,--] \}$ cannot be obtained even if one additionally assumes that the traceless Ricci tensor has two eigenvectors. Indeed, conditions  $2C^{(2)} C^{(2)} - 3 C^{(3)} C^{(1)} =0$ and $\alpha = 0$ imply $\beta=0$. Hence, a space becomes an Einstein space. Such a solution reads
\begin{equation}
\label{rozwiazanie_szczegolnego_typu_DnaD}
\mathcal{A} = \frac{\Lambda}{3} + \frac{x^{3}}{(a_{0} p + b_{0} q + c_{0})^{3}}, \ \Omega=T=0
\end{equation}
where $a_{0}$ is an arbitrary constant, $b_{0}= \{ -1,0,1 \}$ and $c_{0}$ is a constant which is nonzero only if both  $a_{0}$ and $b_{0}$ vanish. The metric given by (\ref{rozwiazanie_szczegolnego_typu_DnaD}) reads
\begin{equation}
\label{metryka_Einstein_Dee_x_Dnn_pp_mm}
\frac{1}{2} ds^{2} = x^{-2} \left( dqdy-dpdx + \left( \frac{\Lambda}{3x} + \frac{x^{2}}{(a_{0} p + b_{0} q + c_{0})^{3}} \right) ( x dp- y dq )^{2} \right)
\end{equation}
To be more precise, solution (\ref{metryka_Einstein_Dee_x_Dnn_pp_mm}) is a metric of the type $\{ [\textrm{D}]^{ee} \otimes [\textrm{D}]^{nn}, [++,--,--,++] \}$. It belongs to one of the classes described in \cite{Bor_Makhmali_Nurowski}.
\end{Uwaga}

\subsubsection{Spaces of the type $\{ [\textrm{III}]^{e} \otimes [\textrm{D}]^{nn},[++,--] \}$} 

In this Section we do not assume that the function $T$ is gauged to a constant value, like it was done in Theorem \ref{theorem_dege_x_Dnn_pp_mm}. Thus, $T=T(q,p)$ but the gauge function $p'=p'(q,p)$ is still arbitrary at our disposal.

The type $\{ [\textrm{III}]^{e} \otimes [\textrm{D}]^{nn},[++,--] \}$ is given by $C^{(3)} = 0$, $C^{(2)} \ne 0$ and $\mathcal{R} \ne 0$. The condition $C^{(3)} = 0$ implies
\begin{equation}
\label{pomoccnicze_rozwiazanie_na_Omega}
\mathcal{A} = Ax +B +\frac{C}{x}, \ \Omega = -A z^{2} + Mz + N
\end{equation}
where $A$, $B$, $C$, $M$ and $N$ are functions of $(q,p)$. The transformation formula for $A$ reads
\begin{equation}
p'_{p} A = A + \frac{p'_{pp}}{p'_{p}}
\end{equation}
Thus, $A$ can be gauged away without any loss of generality. However, from now on the gauge (\ref{gauge_paraKahler_class_1}) is restricted to the transformations such that 
\begin{equation}
\label{restrykcja_na_p}
p'=g(q) \, p + h(q)
\end{equation}

The fact that $\Omega$ becomes a first order polynomial in $z$ (and $z=-y/x$) allows us to return to the hyperheavenly coordinate system $(q,p,x,y)$. In terms of $(q,p,x,y)$ functions $\mathcal{A}$, $\mathcal{Q}$ and $\mathcal{B}$ take the form
\begin{eqnarray}
\label{funkcje metryczne_dla_ttypu_III}
\mathcal{A} = B + \frac{C}{x}, \ \mathcal{Q} = B \, \frac{y}{x} + C \, \frac{y}{x^{2}} + T  x^{2} 
\\ \nonumber
\mathcal{B} = B \, \frac{y^2}{x^2}  + C \, \frac{y^2}{x^3}  + 2T  xy - M  y + N x
\end{eqnarray}
Hence
\begin{eqnarray}
\label{formuly_dla_ttypu_III}
-\frac{\mathcal{R}}{12} &=& B + \frac{2C}{x}  \\ \nonumber
-x^{-4} C^{(2)} &=& B_{p} \, \frac{y}{x^2} + (B_{q} + 4TC) \, \frac{1}{x} + 4T_{p} x + 6T B  - M_{p}
\\ \nonumber
\alpha &=& -4x^{-2} C
\\ \nonumber
\beta &=& -\frac{4}{x^2} (xC_{q} + y C_{p}) + x (M_{p} - 6TB) - 3B_{p} \, \frac{y}{x} - 3 (4TC + B_{q})
\end{eqnarray}
(Coefficient $C^{(1)}$ is more complicated and we do not need its explicit form at this stage). Thus one formulates
\begin{Twierdzenie} 
\label{theorem_IIIe_x_Dnn_pp_mm}
Let $(\mathcal{M}, ds^{2})$ be a complex (neutral) space of the type $\{ [\textrm{III}]^{e} \otimes [\textrm{D}]^{nn}, [++,--] \}$ ($\{ [\textrm{III}_{r}]^{e} \otimes [\textrm{D}_{r}]^{nn}, [++,--] \}$). Then there exist a local coordinate system $(q,p,x,y)$ such that the metric takes the form
\begin{eqnarray}
\label{metryka_IIIe_x_Dnn_pp_mm}
\frac{1}{2} ds^{2} &=&  x^{-2} \Big\{ dqdy-dp dx + \left( B + \frac{C}{x} \right) dp^{2} - 2 \left( B \, \frac{y}{x} + C \, \frac{y}{x^{2}} + T  x^{2} \right) dqdp
\\ \nonumber
&&  \ \ \ \ \ \ + \left(  B \, \frac{y^2}{x^2}  + C \, \frac{y^2}{x^3}  + 2T  xy - M  y + N x \right)   \, dq^{2} \Big\}
\end{eqnarray}
where $B=B(q,p)$, $C=C(q,p)$, $M=M(q,p)$, $N=N(q,p)$ and $T=T(q,p)$ are arbitrary holomorphic (real smooth) functions such that
\begin{subequations}
\begin{eqnarray}
&( B \ne 0 \textrm{ or } C \ne 0 )&
\\ \nonumber
&\textrm{ and }&
\\
 &( B_{p} \ne 0 \textrm{ or } T_{p} \ne 0 \textrm{ or } B_{q} + 4 T C \ne 0 \textrm{ or } 6TB  - M_{p} \ne 0 )&
\end{eqnarray}
\end{subequations}
If $C \ne 0$ the traceless Ricci tensor has four eigenvectors and if $C=0$ the traceless Ricci tensor has two eigenvectors.
\end{Twierdzenie}
\begin{proof}
The proof follows immediately from (\ref{funkcje metryczne_dla_ttypu_III}), (\ref{formuly_dla_ttypu_III}) and (\ref{metric_weak_HH_expanding_space_abrevaitions}).
\end{proof}

\subsubsection{Spaces of the type $\{ [\textrm{N}]^{e} \otimes [\textrm{D}]^{nn},[++,--] \}$}

\begin{Twierdzenie} 
\label{theorem_Ne_x_Dnn_pp_mm}
Let $(\mathcal{M}, ds^{2})$ be a complex (neutral) space of the type $\{ [\textrm{N}]^{e} \otimes [\textrm{D}]^{nn}, [++,--] \}$ ($\{ [\textrm{N}_{r}]^{e} \otimes [\textrm{D}_{r}]^{nn}, [++,--] \}$). Then there exist a local coordinate system $(q,p,x,y)$ such that the metric takes the form
\begin{eqnarray}
\label{metryka_Ne_x_Dnn_pp_mm}
\frac{1}{2} ds^{2} &=&  x^{-2} \Big\{ dqdy-dp dx + \left( B + \frac{C}{x} \right) dp^{2} - 2 \left( B \, \frac{y}{x} + C \, \frac{y}{x^{2}} + T_{0}  x^{2} \right) dqdp
\\ \nonumber
&&  \ \ \ \ \ \ + \left(  B \, \frac{y^2}{x^2}  + C \, \frac{y^2}{x^3}  + 2T_{0}  xy - 6T_{0} B p  y + N x \right)   \, dq^{2} \Big\}
\end{eqnarray}
where $T=T_{0} = \{ -1,0,1 \}$ as a constant, $B=B(q)$, $C=C(q,p)$ and $N=N(q,p)$ are arbitrary holomorphic (real smooth) functions such that
\begin{subequations}
\begin{eqnarray}
\label{typ_NxD_wwwarunek_1}
&(4CT_{0} = - B_{q})&
\\ \nonumber
&\textrm{ and }&
\\ 
\label{typ_NxD_wwwarunek_2}
&(B \ne 0 \textrm{ or } C \ne 0 )&
\\ \nonumber
&\textrm{ and }&
\\
\label{typ_NxD_wwwarunek_3}
 & (C_{pp} \ne 0 \textrm{ or } 4C_{pq} + 3BB_{q} \ne 0 \textrm{ or } 24T_{0}^{2} B^{2}p-N_{pp}-3 T_{0} B_{q} \ne 0 \textrm{ or }&
 \\ \nonumber
 &2NC_{p} - 2CN_{p} - 2C_{qq} - 3p BB_{qq} \ne 0  \textrm{ or }  B_{qq} - 2T_{0}p BB_{q}  \ne 0)&
\end{eqnarray}
\end{subequations}
If $C \ne 0$ the traceless Ricci tensor has four eigenvectors and if $C=0$ the traceless Ricci tensor vanishes.
\end{Twierdzenie}
\begin{proof}
For the type $\{ [\textrm{N}]^{e} \otimes [\textrm{D}]^{nn},[++,--] \}$ the coefficient $C^{(2)}$ vanishes. Hence, we obtain the conditions (compare (\ref{formuly_dla_ttypu_III}))
\begin{equation}
B=B(q),\ T=T(q), \ B_{q} + 4TC=0, \ M_{p} = 6TB
\end{equation}
From (\ref{gauge_paraKahler_class_1_funkcje}) and (\ref{restrykcja_na_p}) one finds the transformation for $T$  
\begin{equation}
T' = \frac{1}{g q'_{q}} T
\end{equation}
Because $T=T(q)$ then it can be gauged to a constant value, $T=T_{0} = \{ -1, 0 ,1 \}$. From now on the gauge functions $g$ and $q'_{q}$ are restricted to such that $g q'_{q}=1$. 

With $T=T_{0}$ we find $M=6T_{0} B p + m(q)$ where $m$ is an arbitrary function which transform as follows
\begin{equation}
 m' = g m - 6T_{0} B h + 3 g_{q}
\end{equation}
Obviously, $m$ can be gauged away without any loss of generality. Hence, $M=6T_{0} B p$ and (\ref{metryka_Ne_x_Dnn_pp_mm}) is proved. The condition (\ref{typ_NxD_wwwarunek_2}) assures that the curvature scalar $\mathcal{R}$ is nonzero. The coefficient $C^{(1)}$ takes the form
\begin{eqnarray}
C^{(1)} &=& -C_{pp} \frac{z^{2}}{x} + \left( 2C_{pq} + \frac{3}{2} BB_{q} \right) \frac{z}{x} + \left( 24T_{0}^{2} B^{2}p-N_{pp}-3 T_{0} B_{q} \right) x
\\ \nonumber
&& + \left( NC_{p} - CN_{p} - C_{qq} - \frac{3}{2}p BB_{qq} \right) \frac{1}{x} + \frac{3}{2} B_{qq} - 3T_{0}p BB_{q} 
\end{eqnarray}
Thus, $C^{(1)} \ne 0$ implies (\ref{typ_NxD_wwwarunek_3}). 

Finally, from (\ref{formuly_dla_ttypu_III}) one finds that $C\ne 0 \ \Longleftrightarrow \ \alpha \ne 0$. However, if  $C=0$ then $\alpha=\beta=0$ and the traceless Ricci tensor vanishes.  
\end{proof}

\subsection{Einstein case} 
\label{sekcja_Einstein}

Results from Section \ref{subsubsekcja_general_results} can be easily specialized to the Einstein case. $\alpha=\beta=0$ and $\mathcal{R}=-4 \Lambda \ne 0$ imply the solutions for $\mathcal{A}$, $\Omega$ and $T$
\begin{equation}
\mathcal{A} = Ax^{3} + Bx + \frac{\Lambda}{3}, \ \Omega = - B z^{2} + Mz +N, \ T = \frac{M_{p}-2B_{q}}{2 \Lambda}
\end{equation}
Transformation formula for $B$ reads
\begin{equation}
p'_{p} B' = B + \frac{p'_{pp}}{p'_{p}}
\end{equation}
Thus, $B$ can be gauged away without any loss of generality but from now on $p'=g(q)p+ h(q)$. With $B=0$ the metrical functions $\phi$, $\mathcal{A}$, $\mathcal{Q}$ and $\mathcal{B}$ take the form
\begin{eqnarray}
\label{ffunkcje_metryczne_dla_Einstein_dexe_x_Dnn}
&& \phi=x, \ \mathcal{A} = A x^{3} + \frac{\Lambda}{3}, \ \mathcal{Q} = A yx^{2} + \frac{\Lambda}{3} \frac{y}{x} + \frac{M_{p}}{2 \Lambda} x^{2},
\\ \nonumber
&& \mathcal{B} = A y^{2} x + \frac{\Lambda}{3} \frac{y^2}{x^2} + \frac{M_{p}}{\Lambda} xy + Nx-My
\end{eqnarray}
and
\begin{eqnarray}
\label{krzywizzzna_dege_x_Dnn_Einstein}
C_{\dot{1}\dot{1}\dot{2}\dot{2}} &=& -\frac{\Lambda}{3}, \ C_{\dot{1}\dot{1}\dot{1}\dot{2}} = \Lambda \frac{y}{x}, \
C_{\dot{1}\dot{1}\dot{1}\dot{1}} = -2 \Lambda \frac{y^2}{x^2}
\\ \nonumber
C^{(3)} &=& -4A x^{3}
\\ \nonumber
-x^{-5} C^{(2)} &=& 4A_{p}y + 4A_{q} x + \frac{2}{\Lambda} M_{pp} 
\end{eqnarray}
(Coefficient $C^{(1)}$ is much more complicated and its explicit form is not necessary at this stage.)

\begin{Twierdzenie} 
\label{theorem_Einstein_dege_x_Dnn_pp_mm}
Let $(\mathcal{M}, ds^{2})$ be an Einstein complex (neutral) space of the type $\{ [\textrm{deg}]^{e} \otimes [\textrm{D}]^{nn}, [++,--] \}$ ($\{ [\textrm{deg}]^{e} \otimes [\textrm{D}_{r}]^{nn}, [++,--] \}$). Then there exist a local coordinate system $(q,p,x,y)$ such that the metric takes the form
\begin{eqnarray}
\label{metryka_Einstein_dege_x_Dnn_pp_mm}
\frac{1}{2} ds^{2} &=&  x^{-2} \Big\{ dqdy-dp dx + \left( A x^{3} + \frac{\Lambda}{3} \right) dp^{2} - 2 \left( A yx^{2} + \frac{\Lambda}{3} \frac{y}{x} + \frac{M_{p}}{2 \Lambda} x^{2}  \right) dqdp \ \ \ \ \ 
\\ \nonumber
&&  \ \ \ \ \ \ + \left( A y^{2} x + \frac{\Lambda}{3} \frac{y^2}{x^2} + \frac{M_{p}}{\Lambda} xy + Nx-My  \right)   \, dq^{2} \Big\}
\end{eqnarray}
where $\Lambda \ne 0$ is a cosmological constant, $A = A (q,p)$, $M = M (q,p)$ and  $N = N (q,p)$ are arbitrary holomorphic (real smooth) functions.
\end{Twierdzenie}
\begin{proof}
Inserting (\ref{ffunkcje_metryczne_dla_Einstein_dexe_x_Dnn}) into (\ref{metric_weak_HH_expanding_space_abrevaitions}) one arrives at (\ref{metryka_Einstein_dege_x_Dnn_pp_mm}).
\end{proof}
As long as $A \ne 0$ and $2C^{(2)}C^{(2)} - 3C^{(3)}C^{(1)} \ne 0$ the metric (\ref{metryka_Einstein_dege_x_Dnn_pp_mm}) is of the type $\{ [\textrm{II}]^{e} \otimes [\textrm{D}]^{nn}, [++,--] \}$. We skip an analysis of the type $ [\textrm{D}]^{ee} \otimes [\textrm{D}]^{nn}$ (all such types were found in \cite{Bor_Makhmali_Nurowski}).

\begin{Twierdzenie} 
\label{theorem_Einstein_IIIe_x_Dnn_pp_mm}
Let $(\mathcal{M}, ds^{2})$ be an Einstein complex (neutral) space of the type $\{ [\textrm{III}]^{e} \otimes [\textrm{D}]^{nn}, [++,--] \}$ ($\{ [\textrm{III}_{r}]^{e} \otimes [\textrm{D}_{r}]^{nn}, [++,--] \}$). Then there exist a local coordinate system $(q,p,x,y)$ such that the metric takes the form
\begin{eqnarray}
\label{metryka_Einstein_IIIe_x_Dnn_pp_mm}
\frac{1}{2} ds^{2} &=&  x^{-2} \Big\{ dqdy-dp dx +   \frac{\Lambda}{3}  dp^{2} - 2 \left(  \frac{\Lambda}{3} \frac{y}{x} + \frac{M_{p}}{2 \Lambda} x^{2}  \right) dqdp \ \ \ \ \ 
\\ \nonumber
&&  \ \ \ \ \ \ + \left(  \frac{\Lambda}{3} \frac{y^2}{x^2} + \frac{M_{p}}{\Lambda} xy + Nx-My  \right)   \, dq^{2} \Big\}
\end{eqnarray}
where $\Lambda \ne 0$ is a cosmological constant, $M = M (q,p)$ and  $N = N (q,p)$ are arbitrary holomorphic (real smooth) functions such that $M_{pp} \ne 0$.
\end{Twierdzenie}
\begin{proof}
$C^{(3)}=0$ implies $A=0$ and $C^{(2)} \ne 0$ implies $M_{pp} \ne 0$ (compare (\ref{krzywizzzna_dege_x_Dnn_Einstein})).
\end{proof}

\begin{Twierdzenie} 
\label{theorem_Einstein_Ne_x_Dnn_pp_mm}
Let $(\mathcal{M}, ds^{2})$ be an Einstein complex (neutral) space of the type $\{ [\textrm{N}]^{e} \otimes [\textrm{D}]^{nn}, [++,--] \}$ ($\{ [\textrm{N}_{r}]^{e} \otimes [\textrm{D}_{r}]^{nn}, [++,--] \}$). Then there exist a local coordinate system $(q,p,x,y)$ such that the metric takes the form
\begin{eqnarray}
\label{metryka_Einstein_Ne_x_Dnn_pp_mm}
\frac{1}{2} ds^{2} &=&  x^{-2} \Big\{ dqdy-dp dx +  \frac{\Lambda}{3}  dp^{2} - 2 \left(  \frac{\Lambda}{3} \frac{y}{x} + \frac{M_{0}}{2 \Lambda} x^{2}  \right) dqdp \ \ \ \ \ 
\\ \nonumber
&&  \ \ \ \ \ \ + \left(  \frac{\Lambda}{3} \frac{y^2}{x^2} + \frac{M_{0}}{\Lambda} xy + Nx-M_{0}py  \right)   \, dq^{2} \Big\}
\end{eqnarray}
where $\Lambda \ne 0$ is a cosmological constant, $M_{0}= \{ -1,0,1 \}$ is a constant and $N = N (q,p)$ is an arbitrary holomorphic (real smooth) function such that $3N_{pp} -2 M_{0}^{2} p\ne 0$.
\end{Twierdzenie}
\begin{proof}
$C^{(2)}=0$ implies $M_{pp}=0$. Hence, $M = m(q) p + n(q)$. Transformation laws for $m$ and $n$ read
\begin{eqnarray}
q'_{q} g m' &=& m
\\ \nonumber
q'_{q} n' &=& n - \frac{h}{g} m + \frac{q'_{q}}{g^2} \frac{\partial}{\partial q} \left( \frac{g^2}{q'_{q}} \right)
\end{eqnarray}
Arbitrariness of the gauge functions $q'_{q}$ and $g$ allows to gauge away $n$ and brought $m$ to a constant value, $m=M_{0} = \{ -1, 0, 1 \}$. Hence, $M=M_{0}p$. Coefficient $C^{(1)}$ reads
\begin{equation}
-\frac{1}{2} x^{-7} C^{(1)} = N_{pp} - \frac{2}{3} M_{0}^{2}p
\end{equation}
Thus, $3N_{pp} - 2 M_{0}^{2}p \ne 0$.
\end{proof}


\section{Spaces of the types $\{ [\textrm{deg}]^{e} \otimes [\textrm{D}]^{nn},[++,++] \}$} 
\label{sekcja_para_Kahler_complicated}
\setcounter{equation}{0}

\subsection{Non-Einstein case} 
\label{sekcja_para_Kahler_complicated_nonEinstein}

In this Section we consider the second class of para-Kähler spaces, i.e., spaces of the types $\{ [\textrm{deg}]^{e} \otimes [\textrm{D}]^{nn},[++,++] \}$. Twists $\varrho_{1}$ and $\varrho_{2}$ are nonzero within this class. Thus,  Lemmas \ref{Lemat_kongruencja_pierwsza_twisting} and \ref{Lemat_kongruencja_druga_twisting} hold true simultaneously. Hence, the coordinate system $(q,p,w,z)$ seems to be the optimal one. For the sake of transparency of the considerations we proceed according to the following strategy\footnote{Alternatively, the metric of spaces of the types  $\{ [\textrm{deg}]^{e} \otimes [\textrm{D}]^{nn},[++,++] \}$ can be obtained from Theorem \ref{theorem_dege_x_degne_pp_pp_and_pp_mp} by demanding that $\mathcal{C}_{n^{\dot{A}}}$ is nonexpanding. However, such an approach is more complicated.}
\begin{enumerate}[label=(\roman*)]
\item We deal with Eqs. $M_{1}=0$ (\ref{pomocnicze_2}) and $N_{1}=0$ (\ref{wiezy_drugiej_kongruencji}). Each of these equations gives a solution for $\phi$. Consequently, we find $\phi$ and an equation which we call "the first consistency condition".
\item We solve the first consistency condition.
\item We consider Eqs. (\ref{rozwiazanie_rownania_pierwszej_struny_2}) and (\ref{rozwiazanie_drugiego_drug_rownanianstrun}) with $M_{2}=0$ (\ref{pomocnicze_3}) and $N_{2}=0$ (\ref{wiezy_drugiej_kongruencji_3}). These are four conditions for three functions $\mathcal{A}$, $\mathcal{Q}$ and $\mathcal{B}$. We obtain solutions for $\mathcal{A}$, $\mathcal{Q}$ and $\mathcal{B}$ and a single equation which we call "the second consistency condition".
\item We solve the second consistency condition.
\end{enumerate}
If (i)-(iv) are completed we are ready to write down the metric. Note, that equations $M_{1}=0$ and $N_{1}=0$ automatically imply that $\theta_{1} \ne 0$ and $\theta_{2} \ne 0$.

\subsubsection{The first consistency condition}

From equation $M_{1}=0$ it follows that $\phi = F (q,p,z) ( x - \Sigma_{z} )$ and this solution is written in the coordinate system $(q,p,x,z)$. Analogously, from equation $N_{1}=0$ one finds that $\phi = G (q,p,w) ( H_{w} - y )$ and this solution is expressed in the coordinate system $(q,p,w,y)$. Both these solutions can be easily transformed into the coordinate system $(q,p,w,z)$ via (\ref{jawne_wzory_na_xiy}). Thus we find the first consistency condition
\begin{equation}
\label{consistency_condition_1}
F (q,p,z) \left( \frac{H-w \Sigma}{1-zw} - \Sigma_{z} \right) = G (q,p,w) \left( H_{w} - \frac{\Sigma -z H} {1-zw} \right)
\end{equation}
Transformation formulas for $\Sigma (q,p,z)$, $H(q,p,w)$ and $F(q,p,z)$ are given by (\ref{transformacja_na_Sigma_ogolna}), (\ref{transformacja_na_H}) and (\ref{transformacja_na_F}), respectively. Transformation for $G(q,p,w)$ reads
\begin{equation}
\label{transformacja_na_G}
G' = \frac{\lambda^{\frac{1}{2}} \Delta G}{p'_{p} - p'_{q} w}
\end{equation}

Playing algebraically with (\ref{consistency_condition_1}) we find two conditions helpful in further analysis
\begin{subequations}
\begin{eqnarray}
\label{wniosek_z_consistency_condition_1_1}
&& F_{zz} H_{ww} = - \Sigma_{zz} G_{ww}
\\ 
\label{wniosek_z_consistency_condition_1_2}
&& G^{3} H_{ww}  = - F^{3} \Sigma_{zz}
\end{eqnarray}
\end{subequations}
(we skip the proofs of (\ref{wniosek_z_consistency_condition_1_1}) and (\ref{wniosek_z_consistency_condition_1_2}) as they are quite straightforward). 

\begin{Lemat}
\label{Lemat_o_wyborze_dowolnych_funkcji_specjalnyprzypadek}
If $\Sigma_{zz}=0 \ \Longleftrightarrow \ H_{ww}=0$ then the functions $F$, $\Sigma$, $G$ and $H$ can be brought to the form 
\begin{equation}
\label{gaug_dla_classI}
\Sigma=0, \ H=w, \ F=1, \ G=w
\end{equation}
without any loss of generality.
\end{Lemat}
\begin{proof}
If $\Sigma_{zz}=0$ then from (\ref{wniosek_z_consistency_condition_1_2}) it follows that $H_{ww}=0$ and vice versa. Thus, the equivalence $\Sigma_{zz}=0 \ \Longleftrightarrow \ H_{ww}=0$ is proved.

If $\Sigma_{zz}=0$ then $\Sigma$ is linear in $z$ with coefficients which depend on $(q,p)$ only. Hence, the gauge functions $\sigma^{\dot{A}}$ can be used to gauge away $\Sigma$ (compare (\ref{transformacja_na_Sigma_ogolna})). After this step the gauge (\ref{gauge}) is restricted to such that $\sigma^{\dot{A}}=0$. Function $H$ becomes linear in $w$. Consequently, it can be simplified to the form $H=w$ (compare (\ref{transformacja_na_H})). The choice $H=w$ imposes additional conditions on gauge (\ref{gauge}), namely, $q'_{p}=0$ and $\lambda^{-1} =q'_{q}$. With $\Sigma=0$ and $H=w$ the first consistency condition (\ref{consistency_condition_1}) reduces to the equation $wF(q,p,z)=G(q,p,w)$. Thus, $F=F(q,p)$. Finally, $F$ can be brought to $1$ (compare (\ref{transformacja_na_F})) what yields $G=w$.
\end{proof}

\begin{Uwaga}
\normalfont
The choice (\ref{gaug_dla_classI}) leaves us with the gauge (\ref{gauge}) restricted to the conditions 
\begin{eqnarray}
\label{gauge_restricted_case_pppp_special}
&& q'=q'(q), \ \frac{dq'}{dq} := f(q), \ \lambda = \frac{1}{f}, \ p'=f^{\frac{1}{2}} p + h(q), 
\\ \nonumber
&& x'=f^{\frac{1}{2}} x, \ y' = y + f^{-\frac{1}{2}} p'_{q}x , \ \sigma^{\dot{A}}=0
\end{eqnarray}
where $f=f(q) \ne 0$ and $h=h(q)$ are arbitrary gauge functions.
\end{Uwaga}

\begin{Uwaga}
\normalfont
Note, that the choice (\ref{gaug_dla_classI}) implies
\begin{equation}
\label{x_i_y_dla_class_I}
\phi=x ; \ \ \ \ \ \ x=\frac{w}{1-zw}, \ y=- \frac{zw}{1-zw} \ \Longleftrightarrow \ z=-\frac{y}{x}, \ w= \frac{x}{1-y}
\end{equation}
\end{Uwaga}

\begin{Uwaga}
\normalfont
 The choice (\ref{gaug_dla_classI}) is not the only one possible for the case with $\Sigma_{zz}=H_{ww}=0$. However, the other choices (including "the most symmetric one", namely $\Sigma=H=1$, $F=z-1$, $G=1-w$) lead to metrics of more complicated form.
 \end{Uwaga}

\begin{Lemat}
\label{lemat_o_specjalnym_podprzypadku}
The following statements are equivalent
\begin{eqnarray}
\nonumber
(i) && \Sigma_{zz}=0 
\\ \nonumber
(ii) && C_{11\dot{A}\dot{B}}=0
\end{eqnarray}
\end{Lemat}
\begin{proof}
If $\Sigma_{zz}=0$ then (\ref{gaug_dla_classI}) holds true and $\phi=x$. Thus, $-\phi C_{11\dot{A}\dot{B}} = \partial_{\dot{A}} \partial_{\dot{B}} x = \partial_{\dot{A}} ( \delta_{\dot{B}}^{\dot{1}})=0$. It proves implication $(i) \ \Longrightarrow \ (ii)$.

If $C_{11\dot{A}\dot{B}}=0$ then $\phi$ takes the form $\phi = J_{\dot{A}} p^{\dot{A}} + \kappa$ where $J_{\dot{A}}=J_{\dot{A}} (q,p) \ne 0$ and $\kappa = \kappa (q,p)$ are arbitrary functions. $J_{\dot{A}}$ must be nonzero, otherwise SD $\mathcal{C}$ is nonexpanding what is a contradiction. From (\ref{gauge}) we find the transformation formulas for $J_{\dot{A}}$ and $\kappa$
\begin{equation}
\label{transformacja_na_J_kappa}
J'_{\dot{A}} = \lambda^{\frac{1}{2}} J_{\dot{B}} \frac{\partial q'_{\dot{A}}}{\partial q_{\dot{B}}}, \ \kappa' =  \lambda^{-\frac{1}{2}} \kappa - J'_{\dot{A}} \sigma^{\dot{A}}
\end{equation}
Hence, one puts $J_{\dot{1}}=1$, $J_{\dot{2}}=0=\kappa$ without any loss of generality. It implies $\phi=x$. It corresponds to the solution $\phi=F(x-\Sigma_{z})$ with $F=1$ and $\Sigma_{z}=0$. Thus, the proof of implication $(ii) \ \Longrightarrow \ (i)$ is completed.
\end{proof}

From Lemma (\ref{lemat_o_specjalnym_podprzypadku}) it follows that the case $\Sigma_{zz}=H_{ww}=0$ ($\Longrightarrow$ (\ref{gaug_dla_classI})) is not a general case because it implies $C_{11\dot{A}\dot{B}}=0$. Hence, the traceless Ricci tensor is not general anymore. Thus, (\ref{gaug_dla_classI}) does not lead to the general metric of para-Kähler, type $\{ [\textrm{deg}]^{e} \otimes [\textrm{D}]^{nn},[++,++] \}$ spaces. However, the choice (\ref{gaug_dla_classI}) will be treated as a starting point for para-Kähler, Einstein spaces of the types $\{ [\textrm{deg}]^{e} \otimes [\textrm{D}]^{nn},[++,++] \}$ (see Section \ref{subsekcja_degexdegnn_pp_pp_Einstein}).

In the rest of this Section we assume $\Sigma_{zz}\ne 0$, $H_{ww} \ne 0$. 

\begin{Lemat}
\label{Lemat_o_ogolnym_przypadku}
Assume $\Sigma_{zz}\ne 0$, $F_{zz} \ne 0$, $G_{ww} \ne 0$ and $H_{ww} \ne 0$. Then functions $\Sigma$, $F$, $G$ and $H$ can be brought to the form
\begin{equation}
\label{gaug_dla_classII_a}
\Sigma=F = A(q,p) \, z^{\frac{1}{2}}, \ -G = H = A(q,p) \, w^{\frac{1}{2}}
\end{equation}
without any loss of generality. $A=A(p,q)$ is an arbitrary nonzero function.
\end{Lemat}
\begin{proof}
If $\Sigma_{zz}\ne 0$ and $H_{ww} \ne 0$ then from (\ref{wniosek_z_consistency_condition_1_2}) it follows that $\Sigma_{zz} = \widetilde{\alpha} / F^{3}$, $H_{ww} = - \widetilde{\alpha} / G^{3}$, where $\widetilde{\alpha} = \widetilde{\alpha} (q,p)$ is an arbitrary nonzero function. Eq. (\ref{wniosek_z_consistency_condition_1_1}) reduces to $G^{3}G_{ww} = F^{3}F_{zz} =: \widetilde{\beta} (q,p)$ where $\widetilde{\beta}$ is an arbitrary nonzero function. Hence, $\Sigma_{zz} = aF_{zz}$ and $H_{ww} = -aG_{ww}$ where $a := \widetilde{\alpha} / \widetilde{\beta}$. Thus 
\begin{equation}
\label{wstepne_rozwiazanie_na_Sigma_przypadek_2_1}
\Sigma = aF + bz + c, \ H = -aG +mw +n
\end{equation}
where $a \ne 0$, $b$, $c$, $m$ and $n$ are functions of $(q,p)$. Inserting (\ref{wstepne_rozwiazanie_na_Sigma_przypadek_2_1}) into consistency condition (\ref{consistency_condition_1}) one arrives at the equation
\begin{eqnarray}
\label{consistency_condition_postep_1}
&& w \left( -a F^{2} + mF -cF +zaFF_{z}  \right) + nF - aFF_{z} -bF =
\\ \nonumber
&& = z \left( -bG - aG^{2}  +nG + awGG_{w}  \right) -cG - aGG_{w} + mG
\end{eqnarray}
$\partial_{w}^{2}$(\ref{consistency_condition_postep_1}) implies $m=c$ and $\partial^{2}_{z}$(\ref{consistency_condition_postep_1}) gives $n=b$. Inserting these conditions into (\ref{consistency_condition_postep_1}) and remembering that $a \ne 0$ we get
\begin{equation}
\label{consistency_condition_postep_2}
w (zFF_{z} - F^{2}) - FF_{z} = z(wGG_{w} - G^{2} ) - GG_{w}
\end{equation}
From (\ref{consistency_condition_postep_2}) it follows that $\partial_{z}^{3} (F^{2})=0$ and $\partial_{w}^{3} (G^{2})=0$. Thus, $F^{2} = Mz^{2} + 2Az +N$ and $G^{2} = P w^{2} + 2Bw + S$, where $M$, $A$, $N$, $P$, $B$ and $S$ are arbitrary functions of $(q,p)$. Feeding (\ref{consistency_condition_postep_2}) with these forms of $F^{2}$ and $G^{2}$ one finds $B=A$, $P=N$ and $S=M$. Finally, we arrive at the solutions
\begin{eqnarray}
&& F^{2} = Mz^{2} + 2A z +N, \ \Sigma = a F + bz +c,
\\ \nonumber
&& G^{2} = Nw^{2} + 2Aw +M, \  H= -aG + cw +b
\end{eqnarray}
Transformation formulas for $a$, $b$, $c$, $M$, $A$ and $N$ reads
\begin{eqnarray}
\label{transformacje_na_funkcjee_pomocnicze}
a' &=& \Delta^{-1} \lambda^{-\frac{3}{2}}   a
\\ \nonumber
b' &=& \lambda^{-1} \Delta^{-1} (q'_{q} b + q'_{p} c) + \sigma^{\dot{1}}
\\ \nonumber
c' &=& \lambda^{-1} \Delta^{-1} (p'_{q} b + p'_{p} c) + \sigma^{\dot{2}}
\\ \nonumber
M' &=& \lambda ({q'_{q}}^{2} M + 2q'_{q}q'_{p} A + {q'_{p}}^{2}N)
\\ \nonumber
N' &=& \lambda ({p'_{q}}^{2} M + 2p'_{q}p'_{p} A + {p'_{p}}^{2}N)
\\ \nonumber
A' &=& \lambda ( q'_{q} p'_{q} M + p'_{p} q'_{p} N + (q'_{q} p'_{p} + q'_{p} p'_{q}) A   )
\end{eqnarray}
Thorough analysis of (\ref{transformacje_na_funkcjee_pomocnicze}) proves that one puts $a=1$ and $b=c=M=N=0$ without any loss of generality. Re-defining $A$, namely $\sqrt{2A} \rightarrow A$, we arrive at (\ref{gaug_dla_classII_a}).
\end{proof}
\begin{Uwaga}
\normalfont
Note, that (\ref{gaug_dla_classII_a}) implies 
\begin{equation}
\label{x_i_y_dla_class_IIa}
x=\frac{Aw^{\frac{1}{2}}}{1+z^{\frac{1}{2}}w^{\frac{1}{2}}}, \ y=\frac{Az^{\frac{1}{2}}}{1+z^{\frac{1}{2}}w^{\frac{1}{2}}}, \ \phi= - \frac{A^{2}}{2} \frac{1-z^{\frac{1}{2}}w^{\frac{1}{2}}}{1+z^{\frac{1}{2}}w^{\frac{1}{2}}}
\end{equation}
\end{Uwaga}

Lemma \ref{Lemat_o_ogolnym_przypadku} holds true if $F_{zz} \ne 0$ and $G_{ww} \ne 0$. The case $F_{zz} = 0$ and $G_{ww} = 0$ must be analyzed separately.

\begin{Lemat}
Assume $\Sigma_{zz}\ne 0$, $H_{ww} \ne 0$, $F_{zz} = 0$ and $G_{ww} = 0$. Then the functions $F$, $\Sigma$, $H$ and $G$ can be brought to the form
\begin{equation}
\label{gaug_dla_classII_b}
F=1, \ G=-w, \ \Sigma = A  z^{2}, \ H = \frac{A}{w}
\end{equation}
without any loss of generality. $A=A(p,q)$ is an arbitrary nonzero function.
\end{Lemat}
\begin{proof}
$F_{zz}=0$ implies $F=a(q,p) z +b(q,p)$. From (\ref{transformacja_na_F}) it follows that without any loss of generality $F$ can be brought to a constant value, namely $F=1$. Eq. (\ref{wniosek_z_consistency_condition_1_2}) reduces to $G^{3} H_{ww} = - \Sigma_{zz}$. Thus, $\Sigma_{zz}$ is a function of $(q,p)$ only so $\Sigma = A(q,p)z^{2} +B(q,p)z+C(q,p)$, $A \ne 0$. From the formula (\ref{transformacja_na_Sigma_ogolna}) it follows that $B$ and $C$ can be gauged away, leaving us with $\Sigma = A z^{2}$. With these forms of $F$ and $\Sigma$, solutions for $G$ and $H$ can be easily obtained from (\ref{consistency_condition_1}). 
\end{proof}

\begin{Uwaga}
\normalfont
Note, that (\ref{gaug_dla_classII_b}) implies 
\begin{equation}
\label{x_i_y_dla_class_IIb}
x=\frac{1+zw}{w} A, \ y=-\frac{z}{w}A, \ \phi = \frac{1-zw}{w} A
\end{equation}
\end{Uwaga}

\subsubsection{The second consistency condition}

The next step is to solve the null string equations (\ref{kongruencja_mdotA_rownania_2}) and (\ref{kongruencja_CndotA_2}). Solutions of these equations are given by (\ref{rozwiazanie_rownania_pierwszej_struny_2}) (expressed in the coordinate system $(q,p,x,z)$) and (\ref{rozwiazanie_drugiego_drug_rownanianstrun}) (expressed in the coordinate system $(q,p,w,y)$).

Vanishing of the expansions of $\mathcal{C}_{m^{\dot{A}}}$ and $\mathcal{C}_{n^{\dot{A}}}$ implies $M_{2}=0$ (\ref{pomocnicze_3}) and $N_{2}=0$ (\ref{wiezy_drugiej_kongruencji_3}). Solution of Eq. (\ref{pomocnicze_3}) is given by the formula (\ref{solution_pomoc_naOiA}) (compare the proof of Theorem \ref{theorem_dege_x_degn_pp}). Solution of Eq. (\ref{wiezy_drugiej_kongruencji_3}) is obtained analogously. Thus
\begin{equation}
\label{rozwiazania_M20_N20}
\mathcal{Q} + z \mathcal{A} = \mathfrak{t}, \ \mathcal{Q} + w\mathcal{B} = \mathfrak{e}
\end{equation}
$\mathfrak{t}$ and $\mathfrak{e}$ stand for abbreviations
\begin{eqnarray}
\label{definicje_e_i_t}
\mathfrak{t} &:=& T \, (x-\Sigma_{z})^{2} + \frac{1}{F} (\Omega F_{z} - zF_{p} +F_{q}) (x-\Sigma_{z}) + \frac{1}{2} (\Sigma_{p} - \Omega \Sigma_{zz} +z \Sigma_{zp} - \Sigma_{zq})
\\ \nonumber
\mathfrak{e} &:=& E \, (y-H_{w})^{2} + \frac{1}{G} (R G_{w} +w G_{q} -G_{p}) (y-H_{w}) - \frac{1}{2} (RH_{ww} +H_{q} +wH_{wq}-H_{wp})
\end{eqnarray}
where $T=T(q,p,z)$, $E=E(q,p,w)$, $\Omega=\Omega (q,p,z)$ and $R=R(q,p,w)$. Combining (\ref{rozwiazania_M20_N20}) with (\ref{rozwiazanie_rownania_pierwszej_struny_2}) and (\ref{rozwiazanie_drugiego_drug_rownanianstrun}) one arrives at the solutions for $\mathcal{A}$, $\mathcal{B}$ and $\mathcal{Q}$
\begin{eqnarray}
\label{formulas_for_AQB_condition_2}
\mathcal{A} &=& -w (\mathcal{Q} + \mathfrak{e}) + R (y-H_{w}) +H_{p} - wH_{q} \equiv  \frac{1}{z} (\mathfrak{t} - \mathcal{Q})
\\ \nonumber
\mathcal{B} &=& -z( \mathcal{Q} + \mathfrak{t}) + \Omega (x-\Sigma_{z}) + z\Sigma_{p} - \Sigma_{q} \equiv  \frac{1}{w} (\mathfrak{e} - \mathcal{Q})
\\ \nonumber
(1-wz) \mathcal{Q} &=& \mathfrak{t}+zw \mathfrak{e} - z (R (y-H_{w}) +H_{p} - wH_{q})
\\ \nonumber
                   &=&  \mathfrak{e} + wz \mathfrak{t} - w (\Omega (x-\Sigma_{z}) + z\Sigma_{p} - \Sigma_{q})
\end{eqnarray}
Clearly, the second consistency condition reads
\begin{equation}
\label{consistency_condition_2}
(1-wz) (\mathfrak{t} - \mathfrak{e}) -z (R (y-H_{w}) +H_{p} - wH_{q}) +w (\Omega (x-\Sigma_{z}) + z\Sigma_{p} - \Sigma_{q}) =0
\end{equation}

\begin{Lemat}
\label{drugi_warunek_konsystencji_klasa_I}
Assume that (\ref{gaug_dla_classI}) holds true. Then (\ref{formulas_for_AQB_condition_2}) and (\ref{consistency_condition_2}) imply
\begin{eqnarray}
\label{wzory_na_AQB_Class_I}
(1-wz)^{3} \mathcal{Q} &=& zR - w^{2} \Omega
\\ \nonumber
(1-wz)^{3} \mathcal{A} &=& -R + w^{3} \Omega
\\ \nonumber
(1-wz)^{3} \mathcal{B} &=& -z^{2} R + w \Omega + w (1-wz)^{2} V
\end{eqnarray}
where $R=R(q,p,w)$, $\Omega  = \Omega (q,p,z)$ and $V=V(q,p)$ are arbitrary functions. 
\end{Lemat}
\begin{proof}
With (\ref{gaug_dla_classI}) and (\ref{x_i_y_dla_class_I}), formulas (\ref{definicje_e_i_t}) simplify considerably and they yield
\begin{equation}
\label{te_uproszczone}
\mathfrak{t} = \frac{w^{2}T}{(1-zw)^{2}} , \ \mathfrak{e} = \frac{E}{(1-zw)^{2}}  - \frac{R}{w(1-zw)} 
\end{equation}
Consistency condition (\ref{consistency_condition_2}) gives
\begin{equation}
\label{rozwiazania_na_T_E_Class_I}
T + \Omega = \frac{1}{w^{2}} \left( E - \frac{1}{w}R \right) =: V(q,p) \ \ \Longrightarrow \ \ T=V-\Omega, \ E = w^{2} V + \frac{1}{w}R
\end{equation}
Thus, $\Omega$, $R$ and $V$ remain as independent functions. Inserting (\ref{rozwiazania_na_T_E_Class_I}) into (\ref{te_uproszczone}) and then into (\ref{formulas_for_AQB_condition_2}) one gets
\begin{eqnarray}
\label{wzory_na_AQB_Class_I_przed_przedefiniowaniem}
(1-wz)^{3} \mathcal{Q} &=& z(R+w^{3}V) - w^{2} (\Omega-V)
\\ \nonumber
(1-wz)^{3} \mathcal{A} &=& -(R+w^{3}V) + w^{3} (\Omega-V)
\\ \nonumber
(1-wz)^{3} \mathcal{B} &=& -z^{2} (R+w^{3}V) + w (\Omega-V) + w (1-wz)^{2} V
\end{eqnarray}
Re-defining functions $\Omega$ and $R$ according to $R + w^{3} V \rightarrow R$, $\Omega - V \rightarrow \Omega$ one arrives at (\ref{wzory_na_AQB_Class_I}).
\end{proof}

\begin{Lemat}
\label{drugi_warunek_konsystencji_klasa_IIa}
Assume that (\ref{gaug_dla_classII_a}) holds true. Then (\ref{formulas_for_AQB_condition_2}) and (\ref{consistency_condition_2}) imply
\begin{eqnarray}
\label{wzory_na_AQB_Class_IIa}
 (1+ z^{\frac{1}{2}}w^{\frac{1}{2}})^{2}(1-wz)\mathcal{Q} &=& \frac{1}{2} \Omega A w z^{-\frac{1}{2}} + \frac{1}{2} RA z w^{-\frac{1}{2}} + \frac{1}{4} V(1+zw) (1- z^{\frac{1}{2}}w^{\frac{1}{2}})^{2} \ \ \ \ \
 \\ \nonumber
 && + (2+ z^{\frac{1}{2}} w^{\frac{1}{2}} +2wz) (1-z^{\frac{1}{2}}w^{\frac{1}{2}}) (A_{p} z^{\frac{1}{2}} - A_{q} w^{\frac{1}{2}})
\\ \nonumber
 z(1+ z^{\frac{1}{2}}w^{\frac{1}{2}})^{2}(1-wz) \mathcal{A} &=& -\frac{1}{2} R A z w^{-\frac{1}{2}} - \frac{1}{2} \Omega A w^{2}z^{\frac{1}{2}} - \frac{1}{2} V zw (1-w^{\frac{1}{2}}z^{\frac{1}{2}})^{2}
 \\ \nonumber
  && + 3A_{q}  z w^{\frac{3}{2}} (1-z^{\frac{1}{2}}w^{\frac{1}{2}}) + A_{p} z w^{\frac{1}{2}} (1-z^{\frac{1}{2}}w^{\frac{1}{2}}) (1-z^{\frac{1}{2}}w^{\frac{1}{2}} + zw)
\\ \nonumber
 w(1+ z^{\frac{1}{2}}w^{\frac{1}{2}})^{2}(1-wz) \mathcal{B} &=& -\frac{1}{2} \Omega A w z^{-\frac{1}{2}} - \frac{1}{2}RA z^{2} w^{\frac{1}{2}} - \frac{1}{2} V zw (1-z^{\frac{1}{2}}w^{\frac{1}{2}})^{2}
 \\ \nonumber
 && -3A_{p} w z^{\frac{3}{2}} (1-z^{\frac{1}{2}}w^{\frac{1}{2}}) - A_{q} w z^{\frac{1}{2}} (1-z^{\frac{1}{2}}w^{\frac{1}{2}}) (1- z^{\frac{1}{2}}w^{\frac{1}{2}} + zw)
\end{eqnarray}
where $R=R(q,p,w)$, $\Omega  = \Omega (q,p,z)$, $A=A(q,p) \ne 0$ and $V=V(q,p)$ are arbitrary functions. 
\end{Lemat}
\begin{proof}
We skip the proof due to its similarity to that of Lemma \ref{drugi_warunek_konsystencji_klasa_I} (although it is much more tedious).
\end{proof}

\begin{Lemat}
\label{drugi_warunek_konsystencji_klasa_IIb}
Assume that (\ref{gaug_dla_classII_b}) holds true. Then (\ref{formulas_for_AQB_condition_2}) and (\ref{consistency_condition_2}) imply
\begin{eqnarray}
\label{wzory_na_AQB_Class_IIb}
 (1-wz)\mathcal{Q} &=& \frac{(1+zw)(1-zw)^{2}}{w^{2}}A^{2}V - A\Omega - \frac{z}{w^{2}} AR - \frac{z(1-zw)(2+zw)}{2w} A_{p} \ \ \ 
 \\ \nonumber
 (1-wz) \mathcal{A} &=& - \frac{2(1-zw)^{2}}{w}A^{2}V + wA \Omega +\frac{1}{w^{2}}AR - (1-wz)A_{q} + \frac{(1-wz)(1+2zw)}{w} A_{p}
 \\ \nonumber
  (1-wz) \mathcal{B} &=& - \frac{2z(1-zw)^{2}}{w^{2}} A^{2}V  + \frac{1}{w} A \Omega + \frac{z^{2}}{w^{2}} AR + \frac{z^{2}(1-zw)}{w} A_{p}
\end{eqnarray}
where $R=R(q,p,w)$, $\Omega  = \Omega (q,p,z)$, $A=A(q,p) \ne 0$ and $V=V(q,p)$ are arbitrary functions. 
\end{Lemat}
\begin{proof}
We skip the proof due to its similarity to that of Lemma \ref{drugi_warunek_konsystencji_klasa_I} (although it is much more tedious).
\end{proof}

\subsubsection{The metrics}

\begin{Twierdzenie} 
\label{theorem_dege_x_degnn_pp_pp}
Let $(\mathcal{M}, ds^{2})$ be a complex (neutral) space of the type $\{ [\textrm{deg}]^{e} \otimes [\textrm{D}]^{nn}, [++,++] \}$ ($\{ [\textrm{deg}]^{e} \otimes [\textrm{D}_{r}]^{nn}, [++,++] \}$). Then there exist a local coordinate system $(q,p,w,z)$ such that the metric takes the form
\begin{eqnarray}
\label{metryka_ogolnyparaKahler_classIIa}
\frac{1}{2} ds^{2} &=& \frac{1}{(1-w^{\frac{1}{2}}z^{\frac{1}{2}})^{2}} \Big\{ z^{-\frac{1}{2}} A \,   (dq+w dp)dz - w^{-\frac{1}{2}} A \,  (dp+zdq) dw
\\ \nonumber
&& \ \ \ \ \ \ \ \ \   - \frac{1}{1-wz} \left[ R \, (dp+zdq)^{2} + \Omega \, (dq+wdp)^{2} \right.
\\ \nonumber
&&  \ \ \ \ \ \ \ \ \ \ \ \ \ \ \ \ \ \ \ \ \ \ + (1-w^{\frac{1}{2}} z^{\frac{1}{2}})^{2} V  (dp+zdq)(dq+wdp)  
\\ \nonumber
&& \ \ \ \ \ \ \ \ \ \ \ \ \ \ \ \ \ \ \ \ \ \  \left. -(1-w^{\frac{1}{2}} z^{\frac{1}{2}}) ( 2z^{\frac{1}{2}} \,A_{p}    - 2w^{\frac{1}{2}} \, A_{q}  ) (dp+zdq)(dq+wdp)   \right] \Big\}
\end{eqnarray}
or
\begin{eqnarray}
\label{metryka_ogolnyparaKahler_classIIb}
\frac{1}{2} ds^{2} &=& \dfrac{1}{(1-wz)^{2}} \Big\{ w  A \, (dq+w dp) dz - A \,  (dp+zdq) dw
\\ \nonumber
&& \ \ \ \ \ \ \ \ \ \ \ \ \ \  - \frac{1}{1-wz} \Big[ R \, (dp+zdq)^{2} + w \Omega \,  (dq+wdp)^{2} 
\\ \nonumber
&&  \ \ \ \ \ \ \ \ \ \ \ \ \ \ \ \ \ \ \ \ \ \ \ \ \ \ \ + (1-wz)^{2} V \, (dp+zdq)(dq+wdp)  
\\ \nonumber
&& \ \ \ \ \ \ \ \ \ \ \ \ \ \ \ \ \ \ \ \ \ \ \ \ \ \ \   -(1-wz) w ( z A_{p}  - A_{q}  ) (dp+zdq)(dq+wdp)    \Big] \Big\}
\end{eqnarray}
where $R=R(q,p,w)$, $\Omega=\Omega (q,p,z)$, $V=V(q,p)$ and $A=A(q,p) \ne 0$ are arbitrary holomorphic (real smooth) functions.
\end{Twierdzenie}
\begin{proof}
Inserting (\ref{x_i_y_dla_class_IIa}) and (\ref{wzory_na_AQB_Class_IIa}) into (\ref{metric_weak_HH_expanding_space_abrevaitions}) and re-defining the functions according to
\begin{equation}
\frac{2}{A^{3}} \rightarrow A, \ \frac{4V}{2A^{4}} \rightarrow V, \ \frac{2w^{-\frac{1}{2}}R}{A^3} \rightarrow R, \ 
\frac{2z^{-\frac{1}{2}}\Omega}{A^3} \rightarrow \Omega
\end{equation}
one arrives at (\ref{metryka_ogolnyparaKahler_classIIa}). 

Similarly, inserting (\ref{x_i_y_dla_class_IIb}) and (\ref{wzory_na_AQB_Class_IIb}) into (\ref{metric_weak_HH_expanding_space_abrevaitions}) and re-defining the functions according to
\begin{equation}
-\frac{1}{A} \rightarrow A, \ - \frac{R}{A} \rightarrow R, \ - \frac{\Omega}{A} \rightarrow \Omega, \ 2V \rightarrow V
\end{equation}
one arrives at (\ref{metryka_ogolnyparaKahler_classIIb}). 
\end{proof}

We end this section at this point. Detailed analysis of the metrics (\ref{metryka_ogolnyparaKahler_classIIa}) and (\ref{metryka_ogolnyparaKahler_classIIb}) will be presented elsewhere. The main question is whether the metrics (\ref{metryka_ogolnyparaKahler_classIIa}) and (\ref{metryka_ogolnyparaKahler_classIIb}) are essentially different or (\ref{metryka_ogolnyparaKahler_classIIb}) is just a special case of (\ref{metryka_ogolnyparaKahler_classIIa}). The answer to this question remains unknown and as long as it is unknown it does not make any sense to discuss the structure of the traceless Ricci tensor or Petrov-Penrose types of the corresponding spaces.

\subsection{Einstein case} 
\label{subsekcja_degexdegnn_pp_pp_Einstein}

In this Section we focus on pKE spaces of the types $\{ [\textrm{deg}]^{e} \otimes [\textrm{D}]^{nn}, [++,++] \}$. We have to impose the conditions $C_{AB \dot{A}\dot{B}}=0$, $\mathcal{R}=-4 \Lambda \ne 0$. The first step has been already done in Section \ref{sekcja_para_Kahler_complicated_nonEinstein} (Lemmas \ref{Lemat_o_wyborze_dowolnych_funkcji_specjalnyprzypadek} and \ref{lemat_o_specjalnym_podprzypadku}). Namely, we proved that the condition $C_{11\dot{A}\dot{B}}=0$ implies (\ref{x_i_y_dla_class_I}) and (\ref{wzory_na_AQB_Class_I}). 

The second step is to consider the conditions $C_{12 \dot{A}\dot{B}}=0$. After some straightforward but tedious work one finds
\begin{eqnarray}
&& \frac{2}{y(1-y)} C_{12\dot{1}\dot{1}} =\frac{4}{x(2y-1)}C_{12\dot{1}\dot{2}}  =-\frac{2}{x^{2}} C_{12\dot{2}\dot{2}} =\alpha
\\ \nonumber
&& \alpha := \Omega_{zz} + \frac{2w}{1-zw} \Omega_{z} + wR_{ww} - \frac{2(1-2zw)R_{w}}{1-zw} - \frac{6zR}{1-zw}
\end{eqnarray}
The curvature scalar reads
\begin{equation}
-\frac{1}{2} (1-zw) \mathcal{R} = w\Omega_{zz} - w^{2} R_{ww} + 4wR_{w} - 6R
\end{equation}
Equations $C_{12\dot{A}\dot{B}}=0$ and $\mathcal{R}=-4\Lambda$ yield
\begin{subequations}
\begin{eqnarray}
\label{rownanie_C11dot1dot2}
&& (1-zw) (\Omega_{zz} + w R_{ww}) + 2w \Omega_{z} - 6zR - 2(1-2zw) R_{w}=0
\\
\label{rownanie_skalar_4lambda}
&& (1-zw) 2 \Lambda = w \Omega_{zz}-w^{2} R_{ww} + 4wR_{w} - 6R
\end{eqnarray}
\end{subequations}
$\partial_{z}^{2}$(\ref{rownanie_skalar_4lambda}) gives $\Omega_{zzzz}=0$ and $\partial_{w}^{2}$(\ref{rownanie_skalar_4lambda}) yields $R_{wwww}=0$. Thus, $\Omega$ and $R$ are third order polynomials in $z$ and $w$, respectively. Further calculations are straightforward and they lead to the formulas
\begin{equation}
\label{rozwiazania_na_R_i_Omega}
\Omega = -\frac{\Lambda}{3} z^{3} + B z^{2} + C z +M, \ R=(M-A) w^{3} + Cw^{2} + Bw - \frac{\Lambda}{3}
\end{equation}
where $A$, $B$, $C$ and $M$ are arbitrary functions of $(q,p)$. We leave the coordinate system $(q,p,w,z)$ and we return to the hyperheavenly coordinate system $(q,p,x,y)$ (compare (\ref{x_i_y_dla_class_I})). Feeding (\ref{wzory_na_AQB_Class_I}) with (\ref{x_i_y_dla_class_I}) and (\ref{rozwiazania_na_R_i_Omega}) one arrives at the formulas
\begin{eqnarray}
\label{rozwiazania_na_AQB_we_wspolrzednych_HH}
\mathcal{A} &=& A x^{3} -Cx^{2} - Bx(1-2y) + \frac{\Lambda}{3} (1-3y+3y^{2})
\\ \nonumber
\mathcal{Q} &=& A x^{2}y - By(1-y) -Mx^{2} + \frac{\Lambda}{3} \frac{y(1-y)(1-2y)}{x}
\\ \nonumber
\mathcal{B} &=& A x y^{2} - Cy(1-y) + Mx (1-2y) + \frac{\Lambda}{3} \frac{y^{2} (1-y)^{2}}{x^{2}} + Vx
\end{eqnarray}

Equations $C_{22 \dot{A}\dot{B}}=0$ are much more tedious. To calculate $C_{22 \dot{A}\dot{B}}$ we need quantities $\eth^{\dot{A}}Q_{\dot{A}\dot{B}}$. They read
\begin{eqnarray}
\label{definicje_Xi}
\Xi_{\dot{1}} &:=& x^{-2} \eth^{\dot{A}} Q_{\dot{A}\dot{1}} = \mathcal{B}_{p} + \mathcal{Q}_{q} - V (Cx^{2} - \Lambda y(1-y))
\\ \nonumber
\Xi_{\dot{2}} &:=& x^{-2} \eth^{\dot{A}} Q_{\dot{A}\dot{2}} = -\mathcal{Q}_{p} - \mathcal{A}_{q} +x V (2Bx + \Lambda (2y-1))
\end{eqnarray}
Finally we get
\begin{eqnarray}
&& \frac{1}{x^{3}y(1-y)} C_{22 \dot{1}\dot{1}} = \frac{2}{x^{4}(2y-1)} C_{22\dot{1}\dot{2}} = - \frac{1}{x^{5}} C_{22\dot{2}\dot{2}} = \beta,
\\ \nonumber
&& \beta := 2\Lambda V - 2B_{q} - C_{p}
\end{eqnarray}
Conditions $C_{22 \dot{A}\dot{B}}=0$ imply the solution for $V$ in terms of $B$ and $C$
\begin{equation}
\label{solution_for_V}
V = \frac{2B_{q} + C_{p}}{2 \Lambda}
\end{equation}
Inserting $\mathcal{A}$, $\mathcal{Q}$ and $\mathcal{B}$ ((\ref{rozwiazania_na_AQB_we_wspolrzednych_HH}) with $V$ given by (\ref{solution_for_V})) into (\ref{curv}) one finds
\begin{eqnarray}
\label{SD_conformal_curvatura_dege_x_Dnn_pppp}
C^{(3)} &=& -4A x^{3}
\\ \nonumber
x^{-5} C^{(2)} &=& -4A_{p} y -4A_{q} x + 4M_{p} +  2C_{q} + \frac{2}{\Lambda} B (2B_{q}+C_{p})
\\ \nonumber
\frac{1}{2} x^{-6} C^{(1)} &=& \frac{\partial \Xi_{\dot{2}}}{\partial q} - \frac{\partial \Xi_{\dot{1}}}{\partial p} - \mathcal{A} \frac{\partial \Xi_{\dot{1}}}{\partial x} - \mathcal{B} \frac{\partial \Xi_{\dot{2}}}{\partial y} - \mathcal{Q} \left( \frac{\partial \Xi_{\dot{2}}}{\partial x} + \frac{\partial \Xi_{\dot{1}}}{\partial y}  \right)
\\ \nonumber
&& + (\mathcal{A}_{x} + \mathcal{Q}_{y}) \Xi_{\dot{1}} + (\mathcal{B}_{y} + \mathcal{Q}_{x}) \Xi_{\dot{2}}
\end{eqnarray}

The ASD part of the conformal curvature read
\begin{eqnarray}
\label{wspolczynniki_ASD_dla_deg_x_D_pppp}
-x^{-2}C_{\dot{1}\dot{1}\dot{1}\dot{1}} &=& \frac{2y^{2}(1-y)^{2}}{x^{4}} \Lambda
\\ \nonumber
-x^{-2}C_{\dot{1}\dot{1}\dot{1}\dot{2}} &=& \frac{(1-y)(2y^{2}-y)}{x^{3}} \Lambda
\\ \nonumber
-x^{-2}C_{\dot{1}\dot{1}\dot{2}\dot{2}} &=& \frac{6y^{2}-6y+1}{3x^{2}} \Lambda
\\ \nonumber
-x^{-2}C_{\dot{1}\dot{2}\dot{2}\dot{2}} &=& \frac{1-2y}{x} \Lambda
\\ \nonumber
-x^{-2}C_{\dot{2}\dot{2}\dot{2}\dot{2}} &=& 2 \Lambda
\end{eqnarray}
Equivalently, (\ref{wspolczynniki_ASD_dla_deg_x_D_pppp}) can be rewritten in the form
\begin{equation}
\label{ASD_dla_deg_x_D_pppp_postac}
-x^{-2} C_{\dot{A}\dot{B}\dot{C}\dot{D}} \xi^{\dot{A}} \xi^{\dot{B}} \xi^{\dot{C}} \xi^{\dot{D}} = 2 \Lambda \left( \xi - \frac{y}{x} \right)^{2} \left( \xi - \frac{y-1}{x} \right)^{2} = 2\Lambda (\xi +z)^2 \left(\xi + \frac{1}{w} \right)^2
\end{equation}
where $\xi^{\dot{A}} = [1, \xi]$ is a spinor. Formula (\ref{ASD_dla_deg_x_D_pppp_postac}) confirms that ASD conformal curvature is of the type [D] and necessarily $\Lambda \ne 0$. Otherwise ASD conformal curvature vanishes.

Transformation formulas for the functions $A$, $B$, $C$ and $M$ are essential for further considerations. The gauge (\ref{gauge}) is restricted to the transformations such that (\ref{gauge_restricted_case_pppp_special}) holds true. Finally one arrives at the formulas
\begin{eqnarray}
f^{\frac{3}{2}} A' &=& A
\\ \nonumber
f^{\frac{1}{2}} B' &=& B - \Lambda f^{-\frac{1}{2}}p'_{q}
\\ \nonumber
f C' &=& C + 2f^{-\frac{1}{2}} p'_{q} B - \Lambda f^{-1} {p'_{q}}^{2}
\\ \nonumber
f M' &=& f^{-\frac{1}{2}} M + f^{-1} {p'_{q}}^{2} B + f^{-1} p'_{q} C - \frac{\Lambda}{3} f^{-2} {p'_{q}}^{3}
\end{eqnarray}

Now we are ready to formulate the main theorem of this article.

\begin{Twierdzenie} 
\label{theorem_Einstein_dege_x_Dnn_pp_pp}
Let $(\mathcal{M}, ds^{2})$ be an Einstein complex (neutral) space of the type $\{ [\textrm{deg}]^{e} \otimes [\textrm{D}]^{nn}, [++,++] \}$ ($\{ [\textrm{deg}]^{e} \otimes [\textrm{D}_{r}]^{nn}, [++,++] \}$). Then there exist a local coordinate system $(q,p,x,y)$ such that the metric takes the form
\begin{eqnarray}
\label{metryka_Einstein_dege_x_Dnn_pp_pp}
\frac{1}{2} ds^{2} &=&  x^{-2} \Big\{ dqdy-dpdx + \left( A x^{3} -Cx^{2} - Bx(1-2y) + \frac{\Lambda}{3} (1-3y+3y^{2}) \right) dp^{2} \ \ \ \ 
\\ \nonumber
&&  \ \ \ \ \ \ -2 \left( A x^{2}y - By(1-y) -Mx^{2} + \frac{\Lambda}{3} \frac{y(1-y)(1-2y)}{x} \right) dqdp
\\ \nonumber
&&  \ \ \ \ \ \ + \left( A x y^{2} - Cy(1-y) + Mx (1-2y) + \frac{\Lambda}{3} \frac{y^{2} (1-y)^{2}}{x^{2}} + \frac{2B_{q} + C_{p}}{2 \Lambda}x   \right) dq^{2}   \Big\}
\end{eqnarray}
where $\Lambda \ne 0$ is a cosmological constant, $A = A (q,p)$, $B = B (q,p)$, $C=C(q,p)$ and  $M = M (q,p)$ are arbitrary holomorphic (real smooth) functions.
\end{Twierdzenie}
\begin{proof}
Inserting (\ref{rozwiazania_na_AQB_we_wspolrzednych_HH}) and (\ref{solution_for_V}) into (\ref{metric_weak_HH_expanding_space_abrevaitions}) one arrives at (\ref{metryka_Einstein_dege_x_Dnn_pp_pp}).
\end{proof}
As long as $A \ne 0$ and $2C^{(2)}C^{(2)} - 3C^{(3)}C^{(1)} \ne 0$ the metric (\ref{metryka_Einstein_dege_x_Dnn_pp_pp}) is of the type $\{ [\textrm{II}]^{e} \otimes [\textrm{D}]^{nn}, [++,++] \}$. As before, we skip an analysis of the type $ [\textrm{D}]^{ee} \otimes [\textrm{D}]^{nn}$.

\begin{Twierdzenie} 
\label{theorem_Einstein_IIIe_x_Dnn_pp_pp}
Let $(\mathcal{M}, ds^{2})$ be an Einstein complex (neutral) space of the type $\{ [\textrm{III}]^{e} \otimes [\textrm{D}]^{nn}, [++,++] \}$ ($\{ [\textrm{III}_{r}]^{e} \otimes [\textrm{D}_{r}]^{nn}, [++,++] \}$). Then there exist a local coordinate system $(q,p,x,y)$ such that the metric takes the form
\begin{eqnarray}
\label{metryka_Einstein_IIIe_x_Dnn_pp_pp}
\frac{1}{2} ds^{2} &=&  x^{-2} \Big\{ dqdy-dpdx + \left(  -Cx^{2} - Bx(1-2y) + \frac{\Lambda}{3} (1-3y+3y^{2}) \right) dp^{2} \ \ \ \ 
\\ \nonumber
&&  \ \ \ \ \ \ -2 \left(  - By(1-y) -Mx^{2} + \frac{\Lambda}{3} \frac{y(1-y)(1-2y)}{x} \right) dqdp
\\ \nonumber
&&  \ \ \ \ \ \ + \left(  - Cy(1-y) + Mx (1-2y) + \frac{\Lambda}{3} \frac{y^{2} (1-y)^{2}}{x^{2}} + \frac{2B_{q} + C_{p}}{2 \Lambda}x   \right) dq^{2}   \Big\}
\end{eqnarray}
where $\Lambda \ne 0$ is a cosmological constant, $B = B (q,p)$, $C=C(q,p)$ and  $M = M (q,p)$ are arbitrary holomorphic (real smooth) functions such that 
\begin{equation}
\label{warunek_na_typ_IIIe_x_Dnn_pppp}
 2M_{p} +  C_{q} + \frac{1}{\Lambda} B (2B_{q}+C_{p}) \ne 0
\end{equation}
\end{Twierdzenie}
\begin{proof}
$C^{(3)}=0$ implies $A=0$ and $C^{(2)} \ne 0$ yields (\ref{warunek_na_typ_IIIe_x_Dnn_pppp}) (compare (\ref{SD_conformal_curvatura_dege_x_Dnn_pppp})).
\end{proof}

Type $\{ [\textrm{N}]^{e} \otimes [\textrm{D}]^{nn}, [++,++] \}$ is a little more complicated. Eq. $C^{(2)}=0$ implies
\begin{equation}
\label{warunek_na_typ_Ne_x_Dnn_pppp_dorozwiazania}
 2M_{p} +  C_{q} + \frac{1}{\Lambda} B (2B_{q}+C_{p}) = 0
\end{equation}
Coefficient $C^{(1)}$ takes relatively simple form
\begin{equation}
\label{C_jeden_dla_typu_dege_x_Dnn_pppp}
\frac{1}{2} x^{-7} C^{(1)} = -\frac{1}{2\Lambda} \partial_{p}^{2} (2B_{q} + C_{p}) - \partial_{p} (MB) + \frac{1}{3} C( B_{q} + 2C_{p}) + \frac{\Lambda}{3} M_{q}
\end{equation}
Those, who prefer the hyperheavenly coordinate system must be satisfied with the metric (\ref{metryka_Einstein_IIIe_x_Dnn_pp_pp}) restricted to the conditions (\ref{warunek_na_typ_Ne_x_Dnn_pppp_dorozwiazania}) and $C^{(1)} \ne 0$. Otherwise, one formulates

\begin{Twierdzenie} 
\label{theorem_Einstein_Ne_x_Dnn_pp_pp}
Let $(\mathcal{M}, ds^{2})$ be an Einstein complex (neutral) space of the type $\{ [\textrm{N}]^{e} \otimes [\textrm{D}]^{nn}, [++,++] \}$ ($\{ [\textrm{N}_{r}]^{e} \otimes [\textrm{D}_{r}]^{nn}, [++,++] \}$). Then there exist a local coordinate system $(q,t,x,y)$ such that the metric takes the form
\begin{eqnarray}
\label{metryka_Einstein_Ne_x_Dnn_pp_pp}
\frac{1}{2} ds^{2} &=&  x^{-2} \Big\{ dqdy-(P_{t}dt + P_{q} dq)dx 
\\ \nonumber
&& \ \ \ \ \ \ + \left(  -\left( t-\frac{B^{2}}{\Lambda} \right) x^{2} - Bx(1-2y) + \frac{\Lambda}{3} (1-3y+3y^{2}) \right) (P_{t}dt + P_{q} dq)^{2} \ \ \ \ 
\\ \nonumber
&&  \ \ \ \ \ \ -2 \left(  - By(1-y) -\left( \frac{N}{2} + \frac{B^{3}}{3 \Lambda^{2}}  \right)  x^{2} + \frac{\Lambda}{3} \frac{y(1-y)(1-2y)}{x} \right) dq (P_{t}dt + P_{q} dq)
\\ \nonumber
&&  \ \ \ \ \ \ + \left[  - \left( t-\frac{B^{2}}{\Lambda} \right) y(1-y) + \left( \frac{N}{2} + \frac{B^{3}}{3 \Lambda^{2}}  \right) x (1-2y) + \frac{\Lambda}{3} \frac{y^{2} (1-y)^{2}}{x^{2}} \right.
\\ \nonumber
&& \ \ \ \ \ \ \left. +  \frac{1}{2\Lambda}  \left( 2B_{q} + \frac{1}{P_{t}} - \frac{2B_{t}}{P_{t}} (2P_{q} - N_{t})  \right) x   \right] dq^{2}   \Big\}
\end{eqnarray}
where $\Lambda \ne 0$ is a cosmological constant, $B := \Lambda (P_{q} - N_{t})$; $N=N(q,t)$ and $P=P(q,t)$ are arbitrary holomorphic (real smooth) functions such that $P_{t} \ne 0$ and
\begin{eqnarray}
\label{warunek_na_nieznikanie_C_jeden}
&& -\frac{1}{2 \Lambda} \frac{1}{P_{t}} \frac{\partial}{\partial t} \left( \frac{1}{P_{t}} \frac{\partial}{\partial t} \left( 2 B_{q} -2\frac{P_{q}}{P_{t}} B_{t} + \frac{1}{P_{t}} - \frac{2}{\Lambda} \frac{1}{P_{t}} BB_{t} \right) \right) - \frac{1}{2} \frac{1}{P_{t}} \partial_{t} (NB) 
\\ \nonumber
&& + \frac{\Lambda}{6} \left( N_{q} - \frac{P_{q}}{P_{t}} N_{t} \right) -\frac{4t}{3 \Lambda} \frac{1}{P_{t}} BB_{t} + \frac{2t}{3} \frac{1}{P_{t}} + \frac{t}{3} \left( B_{q} - \frac{P_{q}}{P_{t}} B_{t} \right) - \frac{2}{3 \Lambda} \frac{1}{P_{t}} B^{2} \ne 0
\end{eqnarray}
\end{Twierdzenie}
\begin{proof}
Substitute $2M =: N + \dfrac{2B^{3}}{3 \Lambda^{2}}$ and $C =: t- \dfrac{B^{2}}{\Lambda}$ into (\ref{warunek_na_typ_Ne_x_Dnn_pppp_dorozwiazania}) ($N$ and $t$ are new functions of $(q,p)$). Eq. (\ref{warunek_na_typ_Ne_x_Dnn_pppp_dorozwiazania}) simplifies to the form
\begin{equation}
\label{waruneknatypN_uproszczonee_1}
N_{p} + t_{q} + \frac{B}{\Lambda}t{_p} = 0
\end{equation}
Multiplying (\ref{waruneknatypN_uproszczonee_1}) by $dp \wedge dq$ and treating $(q,t)$ as independent variables we find
\begin{equation}
\label{rozwiazanie_warunku_na_typ_N}
\widetilde{B} (q,t) = \Lambda (\widetilde{P}_{q} - \widetilde{N}_{t}), \ N (q,p) =: \widetilde{N} (q,t), \ p =: \widetilde{P} (q,t)
\end{equation}
Note that 
\begin{equation}
\label{pomocnicze_wzory_na_transformacje_pochodnych}
\frac{\partial B(q,p)}{\partial q} = \widetilde{B}_{q} + \widetilde{B}_{t} t_{q}, \ 
\frac{\partial B(q,p)}{\partial p} = \widetilde{B}_{t} t_{p}, \ t_{p} = \frac{1}{\widetilde{P}_{t}}, \ t_{q} =- \frac{\widetilde{P}_{q}}{\widetilde{P}_{t}}
\end{equation}
Similar formulas hold true for derivatives of $C(q,p)$ and $N(q,p)$. Inserting (\ref{rozwiazanie_warunku_na_typ_N}) and (\ref{pomocnicze_wzory_na_transformacje_pochodnych}) into (\ref{metryka_Einstein_IIIe_x_Dnn_pp_pp}), passing to the coordinate system $(q,t,x,y)$ and dropping tildes one arrives at (\ref{metryka_Einstein_Ne_x_Dnn_pp_pp}).

$C^{(1)}$ transformed into the coordinate system $(q,t,x,y)$ takes the form
\begin{eqnarray}
\frac{1}{2} x^{-7} C^{(1)} &=& -\frac{1}{2 \Lambda} \frac{1}{P_{t}} \frac{\partial}{\partial t} \left( \frac{1}{P_{t}} \frac{\partial}{\partial t} \left( 2 B_{q} -2\frac{P_{q}}{P_{t}} B_{t} + \frac{1}{P_{t}} - \frac{2}{\Lambda} \frac{1}{P_{t}} BB_{t} \right) \right) - \frac{1}{2} \frac{1}{P_{t}} \partial_{t} (NB) \ \ \ \ \ \ \
\\ \nonumber
&& + \frac{\Lambda}{6} \left( N_{q} - \frac{P_{q}}{P_{t}} N_{t} \right) -\frac{4t}{3 \Lambda} \frac{1}{P_{t}} BB_{t} + \frac{2t}{3} \frac{1}{P_{t}} + \frac{t}{3} \left( B_{q} - \frac{P_{q}}{P_{t}} B_{t} \right) - \frac{2}{3 \Lambda} \frac{1}{P_{t}} B^{2} 
\end{eqnarray}
(where we already dropped tildes over $B$, $P$ and $N$). Because $C^{(1)} \ne 0$, (\ref{warunek_na_nieznikanie_C_jeden}) is proved.
\end{proof}



\section{Concluding remarks.}

This article is the second part of the research programme devoted to para-Kähler, para-Kähler Einstein and Walker geometries. In this part we analyzed complex and real neutral, 4-dimensional spaces equipped with a single, expanding congruence of SD null strings. Thus, a weak and expanding $\mathcal{HH}$-space was a "starting point" for further considerations. Then we equipped a weak $\mathcal{HH}$-space with ASD congruences of null strings according to the following scheme: $\mathcal{C}^{e} \rightarrow \mathcal{C}^{n} \rightarrow \mathcal{C}^{ne}\rightarrow \mathcal{C}^{nn}$ (compare Scheme \ref{Structure}). Finally we arrived at pK-spaces. The path from weak $\mathcal{HH}$-spaces to pK-spaces is full of geometrically interesting spaces equipped with $\mathcal{C}s$ and $\mathcal{I}s$ of different properties. The metrics of all these spaces were found in all the generality. The results are gathered
in the Table \ref{summary} (the metric marked by $^{*}$ has been already known, the rest of the metrics
are new results).

\begin{table}[H]
 \footnotesize
\begin{center}
\begin{tabular}{|c|c|c|}   \hline
 Type  &   Metric   & Functions in the metric        \\ \hline \hline
   \multicolumn{3}{|c|}{Spaces with $C_{ab} \ne 0$}  \\ \cline{1-3}
 $ [\textrm{deg}]^{e} \otimes [\textrm{any}] $ & (\ref{metric_weak_HH_expanding_space_abrevaitions})$^{*}$ & 4 functions of 4 variables     \\ \hline
 $\{ [\textrm{deg}]^{e} \otimes [\textrm{any}]^{e},[++]  \}$ & (\ref{metryka_dege_x_anye_mp_pp}) & 3 functions of 4 variables, 2 functions of 3 variables     \\ \hline
 $\{ [\textrm{deg}]^{e} \otimes [\textrm{any}]^{e},[+-]  \}$ & (\ref{metryka_dege_x_anye_mm_pm}) & 3 functions of 4 variables, 1 function of 3 variables     \\ \hline
 $\{ [\textrm{deg}]^{e} \otimes [\textrm{any}]^{e},[-+]  \}$ & (\ref{metryka_dege_x_anye_mp_pp})  & 2 functions of 4 variables, 3 functions of 3 variables     \\ \hline
 $\{ [\textrm{deg}]^{e} \otimes [\textrm{any}]^{e},[--]  \}$ & (\ref{metryka_dege_x_anye_mm_pm}) & 2 functions of 4 variables, 2 functions of 3 variables     \\ \hline
 $\{ [\textrm{deg}]^{e} \otimes [\textrm{any}]^{n},[++]  \}$ & (\ref{metryka_dege_x_degn_pp})  & 1 function of 4 variables, 4 functions of 3 variables     \\ \hline
 $\{ [\textrm{deg}]^{e} \otimes [\textrm{any}]^{n},[--]  \}$ & (\ref{metryka_dege_x_degn_mm}) & 1 function of 4 variables, 3 functions of 3 variables     \\ \hline
 $\{ [\textrm{deg}]^{e} \otimes [\textrm{deg}]^{ne},[--,++]  \}$ & (\ref{metryka_dege_x_degne_mm_pp}) &  5 functions of 3 variables   \\ \hline
 $\{ [\textrm{deg}]^{e} \otimes [\textrm{deg}]^{ne},[--,+-]  \}$ & (\ref{metryka_dege_x_degne_mm_pm}) &  4 functions of 3 variables   \\ \hline
  $\{ [\textrm{deg}]^{e} \otimes [\textrm{deg}]^{ne},[--,-+]  \}$ & (\ref{metryka_dege_x_degne_mm_mp}) & 3 functions of 3 variables    \\ \hline
 $\{ [\textrm{deg}]^{e} \otimes [\textrm{deg}]^{ne},[++,++]  \}$ & (\ref{metryka_dege_x_degne_pp_pp_and_pp_mp}) & 6 functions of 3 variables     \\ \hline
 $\{ [\textrm{deg}]^{e} \otimes [\textrm{deg}]^{ne},[++,-+]  \}$ & - &  such spaces do not exist   \\ \hline
 $\{ [\textrm{deg}]^{e} \otimes [\textrm{deg}]^{ne},[++,+-]  \}$ & (\ref{metryka_dege_x_degne_pp_pm}) &  5 functions of 3 variables   \\ \hline
 $\{ [\textrm{deg}]^{e} \otimes [\textrm{deg}]^{ne},[++,--]  \}$ & (\ref{metryka_dege_x_degne_pp_mm}) &  3 functions of 3 variables   \\ \hline
\end{tabular}
\caption{Summary of main results.}
\label{summary}
\end{center}
\end{table}

The strongest result of this paper concerns pKE-spaces. We found all pKE metrics which are algebraically special (see Table \ref{summary_2}). The most general space which belong to such class is a space of the type $\{ [\textrm{II}]^{e} \otimes [\textrm{D}]^{nn},[++,++]  \}$. The metric of such a space depends on 4 functions of 2 variables. The metrics of the types $\{ [\textrm{III,N}]^{e} \otimes [\textrm{D}]^{nn},[++,++]  \}$ also were found in all the generality. 

Slightly less general sub-types of pKE-spaces are types $\{ [\textrm{II,III,N}]^{e} \otimes [\textrm{D}]^{nn},[++,--]  \}$). Some of these types (namely, $\{ [\textrm{III,N}]^{e} \otimes [\textrm{D}]^{nn},[++,--]  \}$) were found earlier in \cite{Chudecki_przyklady} (the metric (5.16) of \cite{Chudecki_przyklady}). However, we did not know it yet in \cite{Chudecki_przyklady} that the metric (5.16) is equipped with $\mathcal{I}s$ with properties $[++,--]$. On the other hand, the metric of the type $\{ [\textrm{II}]^{e} \otimes [\textrm{D}]^{nn},[++,--]  \}$ is a new result. 

The only type of pKE-space which metric is still unknown\footnote{According to our best knowledge there is no even a single explicit example of such space.} is the type $ [\textrm{I}] \otimes [\textrm{D}]^{nn}$. The third part of our work will be devoted to space of such a type.

\begin{table}[H]
 \footnotesize
\begin{center}
\begin{tabular}{|c|c|c|}   \hline
 Type  &   Metric   & Functions in the metric        \\ \hline \hline
   \multicolumn{3}{|c|}{Spaces with $C_{ab} \ne 0$}  \\ \cline{1-3}
 $\{ [\textrm{II}]^{e} \otimes [\textrm{D}]^{nn},[++,--]  \}$ & (\ref{metryka_dege_x_Dnn_pp_mm}) &  2 functions of 3 variables, one constant   \\ \hline
 $\{ [\textrm{III}]^{e} \otimes [\textrm{D}]^{nn},[++,--]  \}$ & (\ref{metryka_IIIe_x_Dnn_pp_mm}) &  5 functions of 2 variables   \\ \hline
 $\{ [\textrm{N}]^{e} \otimes [\textrm{D}]^{nn},[++,--]  \}$ & (\ref{metryka_Ne_x_Dnn_pp_mm}) &  2 functions of 2 variables, 1 function of 1 variable, one constant   \\ \hline
 $\{ [\textrm{II}]^{e} \otimes [\textrm{D}]^{nn},[++,++]  \}$ & (\ref{metryka_ogolnyparaKahler_classIIa}) or (\ref{metryka_ogolnyparaKahler_classIIb}) &  2 functions of 3 variables, 2 functions of 2 variables    \\ \hline
 \multicolumn{3}{|c|}{Spaces with $C_{ab} = 0$}  \\ \cline{1-3}
 $\{ [\textrm{II}]^{e} \otimes [\textrm{D}]^{nn},[++,--]  \}$ & (\ref{metryka_Einstein_dege_x_Dnn_pp_mm}) &  3 functions of 2 variables   \\ \hline
 $\{ [\textrm{III}]^{e} \otimes [\textrm{D}]^{nn},[++,--]  \}$ & (\ref{metryka_Einstein_IIIe_x_Dnn_pp_mm})$^{*}$ &  2 functions of 2 variables   \\ \hline
 $\{ [\textrm{N}]^{e} \otimes [\textrm{D}]^{nn},[++,--]  \}$ & (\ref{metryka_Einstein_Ne_x_Dnn_pp_mm})$^{*}$ &  1 functions of 2 variables, one constant   \\ \hline
  $\{ [\textrm{II}]^{e} \otimes [\textrm{D}]^{nn},[++,++]  \}$ & (\ref{metryka_Einstein_dege_x_Dnn_pp_pp}) &  4 functions of 2 variables   \\ \hline
 $\{ [\textrm{III}]^{e} \otimes [\textrm{D}]^{nn},[++,++]  \}$ & (\ref{metryka_Einstein_IIIe_x_Dnn_pp_pp}) &  3 functions of 2 variables   \\ \hline
 $\{ [\textrm{N}]^{e} \otimes [\textrm{D}]^{nn},[++,++]  \}$ & (\ref{metryka_Einstein_Ne_x_Dnn_pp_pp}) &  2 functions of 2 variables   \\ \hline
\end{tabular}
\caption{Summary of main results. Para-Kähler spaces.}
\label{summary_2}
\end{center}
\end{table}



\end{document}